\documentclass[10pt]{article}
\usepackage{amsthm}
\usepackage{caption}
\usepackage{subcaption}
\usepackage[utf8]{inputenc}
\usepackage[T1]{fontenc}
\usepackage[english]{babel}
\usepackage{csquotes}
\usepackage{cite}

\usepackage[shortlabels]{enumitem}
\usepackage{amsmath}
\usepackage{amsthm}
\usepackage{stmaryrd}
\usepackage{amssymb}
\usepackage{mathrsfs}
\usepackage{tikz-cd}
\usepackage{tikz}
\usetikzlibrary{shapes}
\usepackage{hyperref}
\usepackage[nameinlink, capitalize]{cleveref}
\usepackage{ebproof}

\usepackage[margin=1.5in, bottom=1.25in]{geometry}

\usepackage{LICS23}

\setlist[enumerate,1]{label={(\arabic*)}}
\setlist[itemize]{itemsep = 0.2em}

\title{Cartesian Coherent Differential Categories}

\author{Thomas Ehrhard\\
  \small{Université Paris Cité, CNRS, Inria, IRIF, F-75013, Paris, France}\\
  Aymeric Walch\\
  \small{Universié Paris Cité, CNRS, IRIF, F-75013, Paris, France}
}

\definecolor{blue(pigment)}{rgb}{0.2, 0.2, 0.6}
\definecolor{darkgreen}{rgb}{0.0, 0.5, 0.0}
\definecolor{darkred}{HTML}{9F000F}   
\hypersetup
{
    colorlinks=true,
    linkcolor=darkred, 
    filecolor=magenta,            
    urlcolor=blue(pigment),
    citecolor=darkgreen
}

\begin{document}

\maketitle

\begin{abstract}
We extend to general cartesian categories the idea of Coherent Differentiation recently introduced by Ehrhard in the setting of categorical models of Linear Logic. The first ingredient is a summability structure which induces a partial left-additive structure on the category. Additional functoriality and naturality assumptions on this summability structure implement a differential calculus which can also be presented in a formalism close to Blute, Cockett and Seely's cartesian differential categories. We show that a simple term language equipped with a natural notion of differentiation can easily be interpreted in such a category.

\end{abstract}

\tableofcontents

\newcommand\Real{\mathbb{R}}
\newcommand\Realp{\Real_{\geq 0}}
\newcommand\COH{\mathbf{Coh}}
\newcommand\cC{\mathcal C}
\newcommand\Ob[1]{\mathsf{Ob}(#1)}
\newcommand\Comp{\mathrel\circ}

\renewcommand\sharp{\#}

\newcommand{\Iff}{\quad\hbox{iff}\quad}
\newcommand{\Ifsimple}{\quad\hbox{if}\quad}
\newcommand{\Implies}{\Rightarrow}
\newcommand{\Implsymbol}{\Rightarrow}
\newcommand\Equiv{\Leftrightarrow}
\newcommand{\St}{\mid}

\newcommand{\Ro}{\circ}
\newcommand{\Inf}{\bigwedge}
\newcommand{\Infi}{\wedge}
\newcommand{\Sup}{\bigvee}
\newcommand{\Supi}{\vee}

\renewcommand{\Bot}{{\mathord{\perp}}}
\newcommand{\Top}{\top}

\newcommand\Seqempty{\Tuple{}}

\newcommand\cA{\mathcal{A}}
\newcommand\cB{\mathcal{B}}
\newcommand\cD{\mathcal{D}}
\newcommand\cE{\mathcal{E}}
\newcommand\cF{\mathcal{F}}
\newcommand\cG{\mathcal{G}}
\newcommand\cH{\mathcal{H}}
\newcommand\cI{\mathcal{I}}
\newcommand\cJ{\mathcal{J}}
\newcommand\cK{\mathcal{K}}
\newcommand\cL{\mathcal{L}}
\newcommand\cM{\mathcal{M}}
\newcommand\cN{\mathcal{N}}
\newcommand\cO{\mathcal{O}}
\newcommand\cP{\mathcal{P}}
\newcommand\cQ{\mathcal{Q}}
\newcommand\cR{\mathcal{R}}
\newcommand\cS{\mathcal{S}}
\newcommand\cT{\mathcal{T}}
\newcommand\cU{\mathcal{U}}
\newcommand\cV{\mathcal{V}}
\newcommand\cW{\mathcal{W}}
\newcommand\cX{\mathcal{X}}
\newcommand\cY{\mathcal{Y}}
\newcommand\cZ{\mathcal{Z}}

\newcommand\Fini{{\mathrm{fin}}}

\def\frownsmile{%
\mathrel{\vbox{\hbox{${\frown}$}\vspace{-2ex}\hbox{${\smile}$}\vspace{-.5ex}}}}
\def\smilefrown{%
\mathrel{\vbox{\hbox{${\smile}$}\vspace{-2ex}\hbox{${\frown}$}\vspace{-.5ex}}}}

\newcommand\Part[1]{{\mathcal P}\left({#1}\right)}
\newcommand\Parti[1]{{\mathcal I}({#1})}

\newcommand\Union{\bigcup}

\newcommand{\Linarrow}{\multimap}

\def\frownsmile{%
\mathrel{\vbox{\hbox{${\frown}$}\vspace{-2ex}\hbox{${\smile}$}\vspace{-.5ex}}}}
\def\smilefrown{%
\mathrel{\vbox{\hbox{${\smile}$}\vspace{-2ex}\hbox{${\frown}$}\vspace{-.5ex}}}}

\newcommand\CScoh[3]{{#2}\mathrel{\frownsmile_{{#1}}}{#3}}
\newcommand\CScohs[3]{{#2}\mathrel{{\frown}_{#1}}{#3}}
\newcommand\CScohstr[3]{\CScohs{#1}{#2}{#3}}
\newcommand\CSincoh[3]{{#2}\mathrel{\smilefrown_{{#1}}}{#3}}
\newcommand\CSincohs[3]{{#2}\mathrel{{\smile}_{#1}}{#3}}
\newcommand\CSeq[3]{{#2}\mathrel{{=}_{#1}}{#3}}

\newcommand\Myleft{}
\newcommand\Myright{}

\newcommand\Web[1]{\Myleft|{#1}\Myright|}
\newcommand\Suppsymb{\operatorname{\mathsf{supp}}}
\newcommand\Supp[1]{\operatorname{\mathsf{supp}}({#1})}
\newcommand\Suppms[1]{\operatorname{\mathsf{set}}({#1})}

\newcommand\Emptymset{[\,]}
\newcommand\Mset[1]{[{#1}]}
\newcommand\Sset[1]{({#1})}

\newcommand\MonP{P}
\newcommand\MonPZ{\MonP_0}

\newcommand\Cl[1]{\mbox{\textrm{Cl}}({#1})}
\newcommand\ClP[1]{\mbox{\textrm{Cl}}_{\MonP}({#1})}

\newcommand\Star{\star}
\newcommand\CohName{\mathbf{Coh}}
\newcommand\NCohName{\mathbf{NCoh}}
\newcommand\NCOH{\NCohName}
\newcommand\COHR[2]{\CohName(#1,#2)}

\newcommand\Par[2]{{#1}\mathrel{\IPar}{#2}}
\newcommand\Parp[2]{\left({#1}\mathrel{\IPar}{#2}\right)}
\newcommand\ITens{\otimes}
\newcommand\Tens[2]{{#1}\ITens{#2}}
\newcommand\Tensp[2]{({#1}\ITens{#2})}
\newcommand\IWith{\mathrel{\&}}
\newcommand\With[2]{{#1}\IWith{#2}}
\newcommand\Withe[2]{{#1}\IWith^\oc{#2}}
\newcommand\Withp[2]{\left({#1}\IWith{#2}\right)}
\newcommand\IPlus{\oplus}
\newcommand\Plus[2]{{#1}\IPlus{#2}}
\newcommand\Orth[2][]{#2^{\Bot_{#1}}}
\newcommand\Orthp[2][]{(#2)^{\Bot_{#1}}}

\newcommand\Bwith{\mathop{\&}}
\newcommand\Bplus{\mathop\oplus}
\newcommand\Bunion{\mathop\cup}

\newcommand\Scal[2]{\langle{#1}\mid{#2}\rangle}

\newcommand\Pair[2]{\langle{#1},{#2}\rangle}

\newcommand\Inj[1]{\overline\pi_{#1}}

\newcommand\GlobalIndex{I}
\newcommand\Index{\GlobalIndex}
\newcommand\Relbot{\Bot_\Index}
\newcommand\MonPZI{{\MonPZ}^\Index}
\newcommand\ClPI[1]{{\ClP{{#1}}}^\Index}
\newcommand\WebI[1]{\Web{#1}^\Index}

\newcommand\Scalb[2]{\Scal{\Bar{#1}}{\Bar{#2}}}

\newcommand\Ortho[2]{{#1}\mathrel{\bot}{#2}}
\newcommand\Orthob[2]{\Bar{#1}\mathrel{\bot}\Bar{#2}}

\newcommand\Biorth[1]{#1^{\Bot\Bot}}
\newcommand\Biorthp[3]{{#1}^{\Bot_{#2}\Bot_{#3}}}
\newcommand\Triorth[1]{{#1}^{\Bot\Bot\Bot}}
\newcommand\Triorthp[4]{{#1}^{\Bot_{#2}\Bot_{#3}\Bot_{#4}}}

\newcommand\Relpretens[2]{\cR_{#1}\bar\ITens\cR_{#2}}
\newcommand\Relpreplus[2]{\cR_{#1}\bar\IPlus\cR_{#2}}

\newcommand\RelP[1]{\widetilde{#1}}

\newcommand\Eqw[2]{\delta({#1},{#2})}
\newcommand\Eqwb[2]{\Eqw{\Bar{#1}}{\Bar{#2}}}

\newcommand\PFacts[1]{\cF({#1})}

\newcommand\Facts{\cF(\MonPZI)}

\newcommand\RelL[1]{\overline{#1}}

\newcommand\PRel[1]{R_{#1}}

\newcommand\PFamb[2]{[\Bar{#1},\Bar{#2}]}

\newcommand\Fplus[2]{\Bar{#1}+\Bar{#2}}

\newcommand\Char[1]{\epsilon_{#1}}

\newcommand\Fproj[2]{\pi_{#1}(\Bar{#2})}

\newcommand\One{1}

\newcommand\Pbot[1]{\Bot_{#1}}
\newcommand\PBot[1]{\Bot_{#1}}
\newcommand\PRBot[1]{\Bot_{#1}}
\newcommand\PROne[1]{1_{#1}}

\newcommand\Pproj[1]{\pi_{#1}}

\newcommand\Zext[1]{\zeta_{\Bar{#1}}}
\newcommand\Aext[1]{\bar\zeta_{\Bar{#1}}}

\newcommand\Mall{\hbox{\textsf{MALL}}}
\newcommand\RMall{\Mall(\Index)}
\newcommand\RMallr[1]{\Mall({#1})}
\newcommand\FDom[1]{d({#1})}

\newcommand\RBot[1]{\Bot_{#1}}
\newcommand\ROne[1]{\One_{#1}}

\newcommand\Seq[1]{\vdash{#1}}

\newcommand\RSeq[2]{\vdash_{#1}{#2}}

\newcommand\Restr[2]{{#1}|_{#2}}
\newcommand\FRestr[2]{{#1}|_{#2}}

\newcommand\FSem[1]{{#1}^{*}}
\newcommand\PSem[1]{{#1}^{*}}

\newcommand\FFamb[2]{{#1}\langle\Bar{#2}\rangle}

\newcommand\Premskip{\hskip1cm}
\newcommand\Forg[1]{\underline{#1}}

\newcommand\Phase[1]{{#1}^\bullet}

\newcommand\Punit[1]{1^{#1}}

\newcommand\Reg[1]{R_{#1}}

\newcommand\Cont[1]{{#1}^\circ}

\newcommand\Neutral{e}
\newcommand\RNeutral[1]{\Neutral_{#1}}

\newcommand\POne{1}

\newcommand\relstack[2]{\underset{#2}{#1}}
\newcommand\RPlus[4]{{#3}\relstack\IPlus{{#1},{#2}}{#4}}
\newcommand\RWith[4]{{#3}\relstack\IWith{{#1},{#2}}{#4}}

\newcommand\PReq[1]{=_{#1}}

\newcommand\Fam[2]{\hbox{\textrm{Fam}}_{#1}(#2)}
\newcommand\Famfunc[1]{\hbox{\textrm{Fam}}_{#1}}

\newcommand\Pval[1]{\rho_{#1}}

\newcommand\Fcoh{\mathsf C}
\newcommand\Fincoh{\overline\Fcoh}
\newcommand\Feq{\mathsf E}

\newcommand\RZero{\Zero}
\newcommand\RTop{\Top}
\newcommand\EmptyFam{\emptyset}
\newcommand\Partial[1]{\Omega_{#1}}

\newcommand\PMall{\hbox{\textsf{MALL}}_\Omega}
\newcommand\PRMall{\Mall_\Omega(\Index)}
\newcommand\Homo[1]{H(#1)}
\newcommand\RProd[2]{\Lambda_{#1}{({#2})}}
\newcommand\FProd[1]{\widetilde{#1}}

\newcommand\CltoPCl[1]{{#1}^P}

\newcommand\Partfin[1]{{\cP_\Fini}({#1})}

\newcommand\LL{\hbox{\textsf{LL}}}

\newcommand\RLL{\LL(\Index)}
\newcommand\RLLP{\LL^+(\Index)}
\newcommand\RLLext{\LL^\mathrm{ext}(\Index)}

\newcommand\IExcl{{\mathord{!}}}
\newcommand\IInt{{\mathord{?}}}

\newcommand\RExcl[1]{\IExcl_{#1}}
\newcommand\RInt[1]{\IInt_{#1}}

\newcommand\Fempty[1]{0_{#1}}

\newcommand\FAct[1]{{#1}_*}

\newcommand\Noindex[1]{\langle{#1}\rangle}

\newcommand\Card[1]{\#{#1}}
\newcommand\Multi[1]{\mathsf m({#1})}

\newcommand\FamRestr[2]{{#1}|_{#2}}

\newcommand\Pact[1]{{#1}_*}

\newcommand\Pinj[1]{\zeta_{#1}}

\newcommand\Locun[1]{1^J}

\newcommand\Isom\simeq
\newcommand\Webisom{\simeq_{\mathsf w}}

\newcommand\FGraph[1]{\mathsf{g}({#1})}
\newcommand\GPFact[1]{\mathsf f_0({#1})}
\newcommand\GFact[1]{\mathsf f({#1})}

\newcommand\NUCS{\mathbf{nCoh}}

\newcommand\Funinv[1]{{#1}^{-1}}

\newcommand\Reindex[1]{{#1}^*}

\newcommand\Locbot[1]{\Bot({#1})}
\newcommand\Locone[1]{1({#1})}
\newcommand\LocBot{\Locbot}
\newcommand\LocOne{\Locone}

\newcommand\INJ[1]{\mathcal I({#1})}
\newcommand\COHP[1]{\mathcal C({#1})}
\newcommand\FuncFam[1]{\mathrm{Fam}_{#1}}
\newcommand\SET{\mathbf{Set}}
\newcommand\PSET{\mathbf{Set}_0}
\newcommand\Locmod[2]{{#1}({#2})}
\newcommand\FuncPhase[1]{\mathcal{F}_{#1}}
\newcommand\Trans[2][]{{\widehat{#2}}_{#1}}
\newcommand\Useful[1]{\Web{#1}_{\mathrm u}}

\newcommand\Limpl[2]{{#1}\Linarrow{#2}}
\newcommand\Limplp[2]{\left({#1}\Linarrow{#2}\right)}

\newcommand\ConstFam[2]{{#1}^{#2}}

\newcommand\Projtr[1]{{\mathsf{pr}}^{#1}}
\newcommand\Injtr[1]{{\mathsf{in}}^{#1}}

\newcommand\Center[1]{\mathcal C({#1})}

\newcommand\Derel[1]{\mathrm{d}_{#1}}

\newcommand\Mspace[2]{{#2}^*_{#1}}

\newcommand\CStoMS[1]{{#1}^+}
\newcommand\MStoCS[1]{{#1}^-}

\newcommand\PhaseCoh[1]{{\mathbf{Coh}}_{#1}}

\newcommand\Nat{{\mathbb{N}}}
\newcommand\Relint{{\mathbb{Z}}}
\newcommand\Natnz{{\Nat^+}}

\newcommand\Rien{\Omega}

\newcommand\DenEq[2]{{#1}\sim{#2}}

\newcommand\VMon[2]{\langle{#1},{#2}\rangle}
\newcommand\VMonc[3]{\langle{#1},{#2},{#3}\rangle}

\newcommand\Biind[2]{\genfrac{}{}{0pt}{1}{#1}{#2}}

\newcommand\Adapt[1]{\cA_{#1}}

\newcommand\Rllp[1]{\widetilde{#1}}

\newcommand\Emptyfun[1]{0_{#1}}

\newcommand\Multisetfin[1]{\mathcal{M}_{\mathrm{fin}}({#1})}

\newcommand\Myparag[1]{\noindent\textbf{#1.}\ }

\newcommand\Webinf[1]{\Web{#1}_\infty}
\newcommand\Fin[1]{\mathsf{F}({#1})}

\newcommand\Phasefin{\mathbb{F}}

\newcommand\Fspace[1]{\mathrm f{#1}}
\newcommand\Space[1]{\mathrm m{#1}}

\newcommand\Ssupp[1]{\Supp{#1}}

\newcommand\FIN{\mathbf{Fin}}
\newcommand\FINH[2]{\FIN({#1},{#2})}
\newcommand\FINF[1]{\FIN[{#1}]}
\newcommand\SETINJ{\mathbf{Inj}}
\newcommand\EMB{\FIN_{\mathrm e}}

\newcommand\SNat{\mathsf N}
\newcommand\Snat{\mathsf N}
\newcommand\Snatcoh{\mathsf N^\NUCS}

\newcommand\Iter[1]{\mathsf{It}_{#1}}
\newcommand\Piter[2]{\Iter{#1}^{(#2)}}
\newcommand\Case[1]{\mathsf{Case}_{#1}}
\newcommand\Pfix[2]{\mathsf{Y}_{#1}^{(#2)}}
\newcommand\Ifthenelse[3]{\mathtt{if}\,{#1}\,\mathtt{then}\,{#2}\,%
\mathtt{else}\,{#3}}

\newcommand\Ptrace[1]{\|{#1}\|}
\newcommand\Enat[1]{\overline{#1}}
\newcommand\FA[1]{\forall{#1}}
\newcommand\Trcl[2]{\langle {#1},{#2}\rangle}
\newcommand\Faproj[1]{\epsilon^{#1}}
\newcommand\Faintro[3]{\Lambda^{{#1},{#2}}({#3})}

\newcommand\Bool{\mathbf{Bool}}
\newcommand\True{\mathbf t}
\newcommand\False{\mathbf f}

\newcommand\Tarrow{\arrow}
\newcommand\Diffsymb{\mathrm D}
\newcommand\Diffapp[2]{\mathsf D{#1}\cdot{#2}}
\newcommand\Diff[3]{\mathrm D_{#1}{#2}\cdot{#3}}
\newcommand\Diffexp[4]{\mathrm D_{#1}^{#2}{#3}\cdot{#4}}
\newcommand\Diffvar[3]{\mathrm D_{#1}{#2}\cdot{#3}}
\newcommand\Diffterm[3]{\mathrm D_{#1}{#2}\cdot{#3}}
\newcommand\Zero{0}
\newcommand\App[2]{({#1}){#2}}
\newcommand\Apppref[1]{\left({#1}\right)}
\newcommand\Appp[3]{\left({#1}\right){#2}\,{#3}}
\newcommand\Abs[2]{\lambda{#1}\,{#2}}
\newcommand\Abst[3]{\lambda#1^{#2}\,{#3}}
\newcommand\Diffp[3]{\frac{\partial{#1}}{\partial{#2}}\cdot{#3}}
\newcommand\Diffpev[4]{\frac{\partial{#1}}{\partial{#2}}(#3)\cdot{#4}}
\newcommand\Diffr[3]{d{\partial{#1}}{d{#2}}\cdot{#3}}
\newcommand\Derp[3]{\frac{\partial{#1}}{\partial{#2}}\cdot{#3}}
\newcommand\Derd[3]{\frac{d{#1}}{d{#2}}\cdot{#3}}
\newcommand\Derdn[4]{\frac{d^{#1}{#2}}{d{#3}^{#1}}\cdot{#4}}
\newcommand\Derdev[4]{\frac{d{#1}}{d{#2}}(#3)\cdot{#4}}
\newcommand\Derplist[4]{%
\frac{\partial^{#4}{#1}}{\partial#2_1\cdots\partial#2_{#4}}%
\cdot\left(#3_1,\dots,#3_{#4}\right)}
\newcommand\Derplistexpl[6]{%
\frac{\partial^{#6}{#1}}{\partial{#2}\cdots\partial{#3}}%
\cdot\left({#4},\dots,{#5}\right)}
\newcommand\Derpmult[4]{%
\frac{\partial^{#4}{#1}}{\partial{#2}^{#4}}%
\cdot\left(#3_1,\dots,#3_{#4}\right)}
\newcommand\Derpm[4]{%
\frac{\partial^{#4}{#1}}{\partial{#2}^{#4}}%
\cdot\left({#3}\right)}
\newcommand\Derpmultgros[4]{%
\frac{\partial^{#4}}{\partial{#2}^{#4}}\left({#1}\right)%
\cdot\left(#3_1,\dots,#3_{#4}\right)}
\newcommand\Derpmultbis[4]{%
\frac{\partial^{#4}{#1}}{\partial{#2}^{#4}}%
\cdot{#3}}
\newcommand\Derpmultbisgros[4]{%
\frac{\partial^{#4}}{\partial{#2}^{#4}}\left({#1}\right)%
\cdot{#3}}
\newcommand\Derppar[3]{\left(\Derp{#1}{#2}{#3}\right)}
\newcommand\Derpgros[3]{\frac{\partial}{\partial{#2}}\Bigl({#1}\Bigr)%
\cdot{#3}}
\newcommand\Derpdeux[4]{\frac{\partial^2{#1}}{\partial{#2}\partial{#3}}%
\cdot{#4}}
\newcommand\Derpdeuxbis[5]{\frac{\partial^{#5}{#1}}{\partial{#2}\partial{#3}}%
\cdot{#4}}
\newcommand\Diffpv[3]{\frac{\partial}{\partial{#2}}{#1}\cdot{#3}}
\newcommand\Diffpd[5]{\frac{\partial^2{#1}}{\partial{#2}\partial{#3}}%
\cdot\left({#4},{#5}\right)}

\newcommand\Paragraph[1]{\smallbreak\noindent\textbf{#1}}
\newcommand\Parag[1]{\S\ref{#1}}

\newcommand\Redamone{\mathrel{\beta^{0,1}}}
\newcommand\Redonecan{\mathrel{{\bar\beta}^1_{\mathrm D}}}
\newcommand\Redpar{\mathrel{\rho}}
\newcommand\Redparcan{\mathrel{\bar\rho}}

\newcommand\Sat[1]{#1^*}

\newcommand\Can[1]{\left\langle{#1}\right\rangle}
\newcommand\Candiff[3]{\Delta_{#1}{#2}\cdot{#3}}

\newcommand\Eq{\simeq}

\newcommand\List[3]{#1_{#2},\dots,#1_{#3}}
\newcommand\Listbis[3]{#1_{#2}\dots #1_{#3}}
\newcommand\Listc[4]{{#1}_{#2},\dots,{#4},\dots,{#1}_{#3}}
\newcommand\Listbisc[4]{{#1}_{#2}\dots{#4}\dots{#1}_{#3}}
\newcommand\Absent[1]{\widehat{#1}}
\newcommand\Kronecker[2]{\delta_{{#1},{#2}}}

\newcommand\Eqindent{\quad}

\newcommand\Subst[3]{{#1}\left[{#2}/{#3}\right]}
\newcommand\Substv[3]{{#1}\left[{#2\,}/{#3}\right]}
\newcommand\Substz[2]{{#1}\left[0/{#2}\right]}
\newcommand\Substpar[3]{\Bigl({#1}\Bigr)\left[{#2}/{#3}\right]}
\newcommand\Substbis[2]{{#1}\left[{#2}\right]}
\newcommand\Substapp[2]{{#1}{#2}}
\newcommand\Substlinapp[2]{{#2}[{#1}]}

\newcommand\Span[1]{\overline{#1}}
\newcommand\SN{\mathcal{N}}
\newcommand\WN{\mathcal{N}}
\newcommand\Extred[1]{\mathop{\mathrm R}^{\textrm{ext}}{(#1)}}
\newcommand\Onered[1]{\mathop{\mathrm R}^1{(#1)}}
\newcommand\NO[1]{\mathop{\mathrm N}({#1})}
\newcommand\NOD[1]{\mathop{\mathrm N_0}({#1})}

\newcommand\Freemod[2]{{#1}\left\langle{#2}\right\rangle}
\newcommand\Mofl[1]{{#1}^\#}
\newcommand\Mlext[1]{\widetilde{#1}}
\newcommand\Difflamb{\Lambda_D}
\newcommand\Terms{\Lambda_{\mathrm D}}
\newcommand\Diffmod{\Freemod R\Difflamb}

\newcommand\Factor[1]{{#1}!}
\newcommand\Binom[2]{\genfrac{(}{)}{0pt}{}{#1}{#2}}

\newcommand\Multinom[2]{\mathsf{mn}(#2)}

\newcommand\Multinomb[2]{\genfrac{[}{]}{0pt}{}{#1}{#2}}
\newcommand\Suite[1]{\bar#1}
\newcommand\Head[1]{\mathrel{\tau^{#1}}}

\newcommand\Headlen[2]{{\mathrm{L}}({#1},{#2})}
\newcommand\Betaeq{\mathrel{\mathord\simeq_\beta}}
\newcommand\Betadeq{\mathrel{\mathord\simeq_{\beta_{\mathrm D}}}}

\newcommand\Vspace[1]{E_{#1}}
\newcommand\Realpto[1]{(\Realp)^{#1}}
\newcommand\Realto[1]{\Real^{#1}}
\newcommand\Realptow[1]{(\Realp)^{\Web{#1}}}
\newcommand\Realpc{\overline{\Realp}}
\newcommand\Realpcto[1]{\Realpc^{#1}}
\newcommand\Izu{[0,1]}
\newcommand\Intercc[2]{[#1,#2]}
\newcommand\Interoc[2]{(#1,#2]}
\newcommand\Interco[2]{[#1,#2)}
\newcommand\Interoo[2]{(#1,#2)}
\newcommand\Rational{\mathbb Q}

\newcommand\Ring{R}
\newcommand\Linapp[2]{\left\langle{#1}\right\rangle{#2}}
\newcommand\Linappbig[2]{\Bigl\langle{#1}\Bigr\rangle{#2}}
\newcommand\Fmod[2]{{#1}\langle{#2}\rangle}
\newcommand\Fmodr[1]{{\Ring}\langle{#1}\rangle}
\newcommand\Imod[2]{{#1}\langle{#2}\rangle_\infty}
\newcommand\Imodr[1]{{\Ring}\langle{#1}\rangle_\infty}
\newcommand\Res[2]{\langle{#1},{#2}\rangle}
\newcommand\Transp[1]{\Orth{#1}}
\newcommand\Idmat{\operatorname{I}}
\newcommand\FINMOD[1]{\operatorname{\mathbf{Fin}\,}({#1})}
\newcommand\Bcanon[1]{e_{#1}}

\newcommand\Mfin[1]{\mathcal M_\Fini({#1})}
\newcommand\Mfinc[2]{\mathcal M_{#1}({#2})}
\newcommand\Expt[2]{\operatorname{exp}_{#1}({#2})}
\newcommand\Dexpt[2]{\operatorname{exp}'_{#1}({#2})}

\newcommand\Semtype[1]{{#1}^*}
\newcommand\Semterm[2]{{#1}^*_{#2}}

\newcommand\Elem[1]{\operatorname{Elem}({#1})}
\newcommand\Fcard[1]{\operatorname{Card}({#1})}
\newcommand\Linhom[2]{\SCOTTLIN({#1},{#2})}
\newcommand\Hom[2]{\operatorname{Hom}(#1,#2)}
\newcommand\Compana{\circ_\cA}

\newcommand\IConv{\mathrel{*}}
\newcommand\Exclun[1]{\operatorname{u}^{#1}}
\newcommand\Derzero[1]{\partial^{#1}_0}
\newcommand\Dermorph[1]{\partial^{#1}}
\newcommand\Dirac[1]{\delta_{#1}}
\newcommand\Diracm{\delta}

\newcommand\Lintop[1]{\lambda({#1})}
\newcommand\Neigh[2]{\operatorname{\cV}_{#1}(#2)}
\newcommand\Matrix[1]{\mathsf{M}({#1})}

\newcommand\Ev{\operatorname{\mathsf{Ev}}}
\newcommand\Evlin{\operatorname{\mathsf{ev}}}
\newcommand\REL{\operatorname{\mathbf{Rel}}}
\newcommand\RELK{\operatorname{\mathbf{Rel}_\oc}}
\newcommand\RELPOC{\operatorname{\mathbf{RelC}}}

\newcommand\Diag[1]{\Delta^{#1}}
\newcommand\Codiag[1]{\operatorname{a}^{#1}}
\newcommand\Norm[1]{\|{#1}\|}
\newcommand\Normsp[2]{\|{#1}\|_{#2}}

\newcommand\Tpower[2]{{#1}^{\otimes{#2}}}

\newcommand\Termty[3]{{#1}\vdash{#2}:{#3}}
\newcommand\Polyty[3]{{#1}\vdash_!{#2}:{#3}}

\newcommand\Ruleskip{\quad\quad\quad\quad}

\newcommand\Sterms{\Delta}
\newcommand\Polysymb{!}
\newcommand\Pterms{\Delta^{\Polysymb}}
\newcommand\SPterms{\Delta^{(\Polysymb)}}
\newcommand\Nsterms{\Delta_0}
\newcommand\Npterms{\Delta_0^{\Polysymb}}
\newcommand\Nspterms{\Delta^{(\Polysymb)}_0}

\newcommand\Relspan[1]{\overline{#1}}
\newcommand\Rel[1]{\mathrel{#1}}

\newcommand\Redmult{\bar{\beta}^{\mathrm{b}}_\Delta}
\newcommand\Red{\beta_\Delta}
\newcommand\Redst[1]{\mathop{\mathsf{Red}}}

\newcommand\Redredex{\beta^1_\Delta}
\newcommand\Redredexm{\beta^{\mathrm{b}}_\Delta}

\newcommand\Multn{\mathsf{m}}
\newcommand\Shape[1]{\mathcal{T}(#1)}
\newcommand\Tay[1]{{#1}^*}

\newcommand\Deg[1]{\textsf{deg}_{#1}}
\newcommand\Linearize[2]{\mathcal L^{#1}_{#2}}
\newcommand\Fmodrel{\Fmod}

\newcommand\Redeq{=_\Delta}
\newcommand\Kriv{\mathsf K}
\newcommand\Dom[1]{\operatorname{\mathsf{D}}(#1)}
\newcommand\Cons{::}
\newcommand\Addtofun[3]{{#1}_{{#2}\mapsto{#3}}}
\newcommand\Tofclos{\mathsf T}
\newcommand\Tofstate{\mathsf T}
\newcommand\BT{\operatorname{\mathsf{BT}}}
\newcommand\NF{\operatorname{\mathsf{NF}}}

\newcommand\Msubst[3]{\partial_{#3}(#1,#2)}
\newcommand\Msubstbig[3]{\partial_{#3}\bigl(#1,#2\bigr)}
\newcommand\MsubstBig[3]{\partial_{#3}\Bigl(#1,#2\Bigr)}
\newcommand\Symgrp[1]{\mathfrak S_{#1}}

\newcommand\Tcoh{\mathrel{\frownsmile}}
\newcommand\Tcohs{\mathrel{\frown}}
\newcommand\Size{\mathop{\textsf{size}}}

\newcommand\Varset{\mathcal{V}}
\newcommand\NC[1]{\langle{#1}\rangle}
\newcommand\Tuple[1]{\langle{#1}\rangle}
\newcommand\Cotuple[1]{\left[{#1}\right]}
\newcommand\Treeterms{\mathsf{T}}
\newcommand\Tree[2]{\mathop{\mathrm{tree}_{#1}}({#2})}
\newcommand\Treeeq[1]{P_{#1}}
\newcommand\Fvar[1]{\mathop{\textsf{fv}}({#1})}
\newcommand\Avar[1]{\Fvar{#1}}

\newcommand\Msetofsubst[1]{\bar F}
\newcommand\Invgrp[2]{\mathsf G({#1},{#2})}
\newcommand\Invgrpsubst[3]{\mathsf G({#1},{#2},{#3})}
\newcommand\Eqofpart[1]{\mathrel{\mathord{\sim}_{#1}}}
\newcommand\Order[1]{\left|{#1}\right|}
\newcommand\Inv[1]{{#1}^{-1}}
\newcommand\Linext[1]{\widetilde{#1}}
\newcommand\Linauto{\operatorname{\mathsf{Iso}}}

\newcommand{\Ens}[1]{\{#1\}}

\newcommand\Dummy{\_}

\newcommand\Varrename[2]{{#2}^{#1}}

\newcommand\Pcoh[1]{\mathsf P{#1}}
\newcommand\Pcohp[1]{\Pcoh{(#1)}}
\newcommand\Pcohc[1]{\overline{\mathsf P}{#1}}
\newcommand\Pcohcp[1]{\Pcohc{(#1)}}
\newcommand\Pcohsig[1]{\mathsf B_1{#1}}
\newcommand\Base[1]{\overline{#1}}
\newcommand\Matapp[2]{{#1}\Compl{#2}}

\newcommand\PCOH{\mathbf{Pcoh}}
\newcommand\PCOHKL{\mathbf{Pcoh}_!}
\newcommand\PCOHEMB{\mathbf{Pcoh}_{\mathord\subseteq}}

\newcommand\Leftu{\lambda}
\newcommand\Rightu{\rho}
\newcommand\Assoc{\alpha}
\newcommand\Sym{\gamma}
\newcommand\Msetempty{\Mset{\,}}

\newcommand\Tenstg[3]{\Tens{(\Tens{#1}{#2})}{#3}}
\newcommand\Tenstd[3]{\Tens{#1}{(\Tens{#2}{#3})}}

\newcommand\Banach[1]{\mathsf e{#1}}
\newcommand\Absval[1]{\left|{#1}\right|}

\newcommand\Msetsum[1]{\Sigma{#1}}

\newcommand\Retri\zeta
\newcommand\Retrp\rho

\newcommand\Subpcs{\subseteq}
\newcommand\Cuppcs{\mathop{\cup}}

\newcommand\Impl[2]{{#1}\Rightarrow{#2}}
\newcommand\Implp[2]{({#1}\Rightarrow{#2})}
\newcommand\Basemax[1]{\mathsf c_{#1}}

\newcommand\Isosucc{\mathsf s}

\newcommand\Leqway{\ll}

\newcommand\Funofmat[1]{\mathsf{fun}(#1)}
\newcommand\Funkofmat[1]{\mathsf{Fun}(#1)}

\newcommand\Pcoin[3]{\mathsf{ran}({#1},{#2},{#3})}

\newcommand\Tsem[1]{\llbracket{#1}\rrbracket}
\newcommand\Tsemrel[1]{\llbracket{#1}\rrbracket^{\REL}}
\newcommand\Tsemcoh[1]{\llbracket{#1}\rrbracket^{\NUCS}}
\newcommand\Tseme[1]{\llbracket{#1}\rrbracket^{\mathord\oc}}
\newcommand\Psem[2]{\llbracket{#1}\rrbracket_{#2}}
\newcommand\Psemrel[2]{\llbracket{#1}\rrbracket^{\REL}_{#2}}
\newcommand\Psemcoh[2]{\llbracket{#1}\rrbracket^{\NUCS}_{#2}}

\newcommand\Tnat\iota
\newcommand\Fix[1]{\operatorname{\mathsf{fix}}(#1)}
\newcommand\If[3]{\operatorname{\mathsf{if}}(#1,#2,#3)}
\newcommand\Pred[1]{\operatorname{\mathsf{pred}}(#1)}
\newcommand\Succ[1]{\operatorname{\mathsf{succ}}(#1)}
\newcommand\Num[1]{\underline{#1}}
\newcommand\Loop\Omega
\newcommand\Loopt[1]{\Omega^{#1}}
\newcommand\Ran[1]{\mathsf{ran}(#1)}
\newcommand\Dice[1]{\operatorname{\mathsf{coin}}(#1)}
\newcommand\Tseq[3]{{#1}\vdash{#2}:{#3}}
\newcommand\Tseqst[2]{{#1}\vdash{#2}}
\newcommand\Eseq[5]{{#1}\vdash{#2}:{((#3\to #4)\to #5)}}
\newcommand\Timpl\Impl
\newcommand\Timplp\Implp
\newcommand\Simpl\Impl

\newcommand\Pcfred{\beta_{\mathrm{PCF}}}

\newcommand\PCFP{\mathsf{pPCF}}
\newcommand\PCF{\textsf{PCF}}
\newcommand\PCFPZ{\mathsf{pPCF}^-}

\newcommand\Curry[1]{\operatorname{\mathsf{Cur}}({#1})}
\newcommand\Fixpcoh[1]{\operatorname{\mathsf{Fix}}_{#1}}

\newcommand\Pnat{\mathbf N}

\newcommand\Evcontext[1]{\mathsf C_{#1}}
\newcommand\Evhead[1]{\mathsf H_{#1}}
\newcommand\Reduct[1]{\operatorname{\mathsf{red}}({#1})}
\newcommand\Ready[1]{\operatorname{\mathsf{ready}}({#1})}
\newcommand\Undefready[1]{\operatorname{\mathsf{ready}}({#1})\mathord\uparrow}
\newcommand\Defready[1]{\operatorname{\mathsf{ready}}({#1})\mathord\downarrow}

\newcommand\Contextempty{\emptyset}

\newcommand\States{\operatorname{\mathsf S}}

\newcommand\Stateloop[1]{\perp^{#1}}
\newcommand\Matred{\operatorname{\mathsf{Red}}}

\newcommand\Vecrow[1]{\operatorname{\mathsf r}_{#1}}
\newcommand\Veccol[1]{\operatorname{\mathsf c}_{#1}}

\newcommand\Plotkin[1]{\cR^{#1}}

\newcommand\Redone[1]{\stackrel{#1}\rightarrow}
\newcommand\Redonetr[1]{\stackrel{#1}{\rightarrow^*}}
\newcommand\Redoned{\rightarrow_{\mathsf d}}

\newcommand\Transition{\leadsto}
\newcommand\Probatransp{\operatorname{\mathsf p}}
\newcommand\Lg{\operatorname{\mathsf{lg}}}
\newcommand\Emptyseq{1}

\newcommand\Cbsdirect[1]{{#1}^+}
\newcommand\Cbsdual[1]{{#1}^-}
\newcommand\Cbsres[3]{\left\langle#2,#3\right\rangle_{#1}}
\newcommand\Cbsdirectn[2]{\left\|{#2}\right\|_{#1}^+}
\newcommand\Cbsdualn[2]{\left\|{#2}\right\|_{#1}^-}
\newcommand\Cbsdirectp[1]{\mathsf C_{#1}^+}
\newcommand\Cbsdualp[1]{\mathsf C_{#1}^-}

\newcommand\Cbsofpcs[1]{\mathsf{cbs}({#1})}

\newcommand\Init[2][]{\mathop{\mathord\downarrow}_{#1}{#2}}
\newcommand\Final[2][]{\mathop{\mathord\uparrow}_{#1}{#2}}
\newcommand\Poop[1]{{#1}^{\mathrm{op}}}
\newcommand\Dual[1][]{\mathrel{\mathord\perp}_{#1}}

\newcommand\PP{PP}

\newcommand\Scottsymb{\mathsf S}

\newcommand\Weak[1]{\operatorname{\mathsf{w}}_{#1}}
\newcommand\Contr[1]{\operatorname{\mathsf{contr}}_{#1}}
\newcommand\WeakS[1]{\Weak{#1}^\Scottsymb}
\newcommand\ContrS[1]{\Contr{#1}^\Scottsymb}

\newcommand\Contrc[1]{\operatorname{\mathsf{contr}}_{#1}}

\newcommand\Der[1]{\operatorname{\mathsf{der}}_{#1}}
\newcommand\Coder[1]{\operatorname{\mathord\partial}_{#1}}
\newcommand\Digg[1]{\operatorname{\mathsf{dig}}_{#1}}
\newcommand\DerS[1]{\Der{#1}^{\Scottsymb}}
\newcommand\DiggS[1]{\Digg{#1}^{\Scottsymb}}

\newcommand\Lfun{\operatorname{\mathsf{fun}}}
\newcommand\Fun[1]{\widehat{#1}}

\newcommand\LfunS{\Lfun^\Scottsymb}
\newcommand\FunS{\Fun^\Scottsymb}

\newcommand\Ltrace{\mathsf{tr}}
\newcommand\Trace{\mathsf{Tr}}
\newcommand\LtraceS{\mathsf{tr}^\Scottsymb}
\newcommand\TraceS{\mathsf{Tr}^\Scottsymb}

\newcommand\Id{\operatorname{\mathsf{Id}}}
\newcommand\IdS{\Id^\Scottsymb}

\newcommand\Proj[1]{\mathsf{pr}_{#1}}
\newcommand\ProjS[1]{\Proj{#1}^\Scottsymb}
\newcommand\EvS{\Ev^\Scottsymb}
\newcommand\EvlinS{\Evlin^\Scottsymb}

\newcommand\Excl[1]{\oc{#1}}
\newcommand\Excll[1]{\oc\oc{#1}}
\newcommand\Exclp[1]{\oc({#1})}
\newcommand\Int[1]{\wn{#1}}
\newcommand\Intp[1]{\wn({#1})}
\newcommand\ExclS{\Excl}
\newcommand\IntS{\Int}

\newcommand\Prom[1]{#1^!}
\newcommand\Promm[1]{{#1}^{!!}}
\newcommand\PromS{\Prom}
\newcommand\PrommS{\Promm}

\newcommand\SCOTTLIN{\operatorname{\mathbf{ScottL}}}
\newcommand\SCOTTPOC{\operatorname{\mathbf{ScottC}}}
\newcommand\SCOTTLINKL{\operatorname{\mathbf{ScottL}}_\oc}
\newcommand\SCOTTLINK{\operatorname{\mathbf{ScottL}}_\oc}
\newcommand\PPLIN{\operatorname{\mathbf{PpL}}}
\newcommand\PPLINK{\operatorname{\mathbf{PpL}}_\oc}
\newcommand\PPPOC{\operatorname{\mathbf{PpC}}}
\newcommand\RELLIN{\operatorname{\mathbf{RelL}}}
\newcommand\PERLIN{\operatorname{\mathbf{PerL}}}
\newcommand\PERPOC{\operatorname{\mathbf{PerC}}}
\newcommand\PERLINK{\operatorname{\mathbf{PerL}_\oc}}

\newcommand\Mix{\operatorname{\mathsf{mix}}}

\newcommand\Realize[1]{\Vdash_{#1}}
\newcommand\Collapseeq[1]{\sim_{#1}}

\newcommand\Relsem[1]{\left[{#1}\right]}
\newcommand\Perofpp[1]{\epsilon_0({#1})}
\newcommand\Perofppf{\epsilon_0}
\newcommand\Scottofpp[1]{\sigma_0({#1})}
\newcommand\Scottofppf{\sigma_0}

\newcommand\Perofpprel[1]{\mathrel B_{#1}}
\newcommand\Perofppr[1]{\Collapseeq{\Perofpp{#1}}}
\newcommand\Perofppro[1]{\Collapseeq{\Orth{\Perofpp{#1}}}}

\newcommand\Kapp[2]{#1(#2)}

\newcommand\Left{\mathsf l}
\newcommand\Up{\mathsf u}
\newcommand\Upleft{\mathsf{ul}}
\newcommand\Upright{\mathsf{ur}}

\newcommand\Collapsecat{\operatorname{\mathsf e}}
\newcommand\Hetercat{\operatorname{\mathsf e}}
\newcommand\Hetercatmod{\operatorname{\mathsf e_{\mathsf{mod}}}}
\newcommand\Heterobjnext[1]{\ulcorner{#1}\urcorner}
\newcommand\Heterobjext[1]{\llcorner{#1}\lrcorner}
\newcommand\Heternext[1]{\ulcorner{#1}\urcorner}
\newcommand\Heterext[1]{\llcorner{#1}\lrcorner}
\newcommand\Collapseobj[1]{\ulcorner{#1}\urcorner}
\newcommand\Collapsenext{\rho}

\newcommand\EXT{\sigma}
\newcommand\COLL{\epsilon}
\newcommand\COLLL{\epsilon_0}

\newcommand\Subper{\sqsubseteq}
\newcommand\Unionper{\bigsqcup}
\newcommand\Purerel{\cD_{\mathsf r}}
\newcommand\Pureper{\cD_{\mathsf e}}
\newcommand\Pureperstep{\Phi_{\mathsf e}}

\newcommand\Subpo{\sqsubseteq}
\newcommand\Unionpo{\bigsqcup}
\newcommand\Purepo{\cD_{\mathsf s}}
\newcommand\Purepostep{\Phi_{\mathsf s}}

\newcommand\Subpp{\sqsubseteq}
\newcommand\Unionpp{\bigsqcup}
\newcommand\Purepp{\cD_{\mathsf{pp}}}
\newcommand\Pureheter{\cD_{\mathsf{h}}}
\newcommand\Pureppstep{\Phi_{\mathsf{pp}}}

\newcommand\Appmod{\operatorname{\mathsf{app}}}
\newcommand\Lammod{\operatorname{\mathsf{lam}}}

\newcommand\Kleisli[1]{{#1}_{\Excl{}}}

\newcommand\Initw[2]{\Init[\Web{#1}]{#2}}
\newcommand\Finalw[2]{\Final[\Web{#1}]{#2}}

\newcommand\Heterofpp{\mathsf h}
\newcommand\Collofper{\mathsf q}

\newcommand\Hetermod{Z}

\newcommand\Relincl\eta
\newcommand\Relrestr\rho

\newcommand\Seely{\mathsf m}
\newcommand\Seelyz{\Seely^0}
\newcommand\Seelyt{\Seely^2}
\newcommand\Monoidal{\mu}

\newcommand\Compl{\,}
\newcommand\Curlin{\operatorname{\mathsf{cur}}}
\newcommand\Cur{\mathsf{Cur}}

\newcommand\Op[1]{{#1}^{\mathsf{op}}}
\newcommand\Kl[1]{{#1}_\oc}
\newcommand\Em[1]{{#1}^\oc}

\newcommand\NUM[1]{\underline{#1}}
\newcommand\SUCC{\mathsf{succ}}
\newcommand\PRED{\mathsf{pred}}
\newcommand\IF{\mathsf{if}}
\newcommand\IFV[3]{\IF(#1,(#2)\,#3)}
\newcommand\Ifv[4]{\IF(#1,#2,#3\cdot #4)}
\newcommand\FIX{\mathsf{fix}}
\newcommand\FIXV[2]{\mathsf{fix}\,#1\cdot#2}
\newcommand\FIXVT[3]{\mathsf{fix}\,#1^{#2}\,#3}
\newcommand\Fixt[3]{\mathsf{fix}\,#1^{#2}\,#3}
\newcommand\APP[2]{(#1)#2}
\newcommand\ABST[3]{\lambda #1^{#2}\, #3}
\newcommand\MUT[3]{\mu #1^{#2}\, #3}
\newcommand\MU[2]{\mu #1\, #2}
\newcommand\NAMED[2]{[#1]#2}

\newcommand\Abstpref[2]{\lambda #1^{#2}\,}

\newcommand\NAT{\iota}
\newcommand\IMPL[2]{#1\Rightarrow #2}

\newcommand\SEQ[4]{#1\vdash #2:#3\mid #4}
\newcommand\SEQI[3]{#1\vdash #2:#3}
\newcommand\SEQN[3]{#1\vdash #2\mid #3}
\newcommand\SEQP[3]{#1\vdash #2\mid #3}
\newcommand\SEQS[4]{#1\mid #2:#3\vdash #4}

\newcommand\HOLE{[\ ]}

\newcommand\Redpcf{\mathord{\to}}
\newcommand\Redpcfi{\mathord{\to_{\mathsf{PCF}}}}
\newcommand\Redclean{\relstack\to{\mathsf{cl}}}

\newcommand\Freenames{\mathsf{FN}}

\newcommand\STARG[1]{\mathsf{arg}(#1)}
\newcommand\STIF[2]{\mathsf{ifargs(#1,#2)}}
\newcommand\STSUCC{\mathsf{succ}}
\newcommand\STPRED{\mathsf{pred}}
\newcommand\EMPTY{\epsilon}
\newcommand\CONS[2]{#1\cdot #2}
\newcommand\PUSH[2]{#1\cdot #2}
\newcommand\STATE[3]{#1\star #2\mid #3}
\newcommand\TAPP[2]{#1\star#2}
\newcommand\PROC[2]{#1\star#2}

\newcommand\Div[1]{#1\uparrow}
\newcommand\Conv[2]{#1\downarrow#2}
\newcommand\Convproba[1]{\mathrel{\downarrow^{#1}}}

\newcommand\Pcftrad[1]{{#1}^\bullet}
\newcommand\Transcl[1]{{#1}^*}

\newcommand\Conn[1]{\gamma_{#1}}
\newcommand\Sing[1]{\mathsf C(#1)}
\newcommand\RELW{\operatorname{\mathbf{RelW}}}

\newcommand\Natrelw{\mathsf{N}}

\newcommand\Least{\bot}

\newcommand\Posnat{\overline{\mathsf N}}
\newcommand\Botnat{\Bot}
\newcommand\Coalg[1]{h_{#1}}

\newcommand\Ifrel{\mathsf{if}}

\newcommand\Eval[2]{\langle#1,#2\rangle}

\newcommand\Let[3]{\mathsf{let}(#1,#2,#3)}

\newcommand\Add{\mathsf{add}}
\newcommand\Exp{\mathsf{exp}_2}
\newcommand\Cmp{\mathsf{cmp}}

\newcommand\Unift{\mathsf{unif}_2}
\newcommand\Unif{\mathsf{unif}}

\newcommand\Closed[1]{\Lambda^{#1}_0}
\newcommand\Open[2]{\Lambda^{#2}_{#1}}

\newcommand\Redmat[1]{\mathsf{Red}(#1)}
\newcommand\Redmats{\mathsf{Red}}
\newcommand\Redmato[2]{\mathsf{Red}(#1,#2)}

\newcommand\Mexpset[2]{\mathsf L(#1,#2)}

\newcommand\Expmonisoz{\Seely^0}
\newcommand\Expmonisob[2]{\Seely^2_{#1,#2}}
\newcommand\Expmonisobn{\Seely^2}

\newcommand\Injms[2]{#1\cdot#2}

\newcommand\Natobj{\mathsf{N}}
\newcommand\Natalg{h_\Natobj}

\newcommand\Snum[1]{\overline{#1}}
\newcommand\Sif{\overline{\mathsf{if}}}
\newcommand\Ssuc{\overline{\mathsf{suc}}}
\newcommand\Spred{\overline{\mathsf{pred}}}

\newcommand\Sfix{\mathsf{Y}}

\newcommand\Vect[1]{\vec{#1}}

\newcommand\Rts[1]{\cR^{#1}}

\newcommand\Probw[1]{\mathsf p(#1)}
\newcommand\Spath[2]{\mathsf{R}(#1,#2)}

\newcommand\Obseq{\sim}
\newcommand\Obsleq{\lesssim}

\newcommand\Ptest[1]{{#1}^+}
\newcommand\Ntest[1]{{#1}^-}
\newcommand\Plen[1]{|#1|^+}
\newcommand\Nlen[1]{|#1|^-}

\newcommand\Shift{\mathsf{shift}}
\newcommand\Shvar[2]{\App{\Shift_{#1}}{#2}}

\newcommand\Probe{\mathsf{prob}}
\newcommand\Pprod{\mathsf{prod}}
\newcommand\Pchoose{\mathsf{choose}}

\newcommand\Argsep{\,}

\newcommand\Simplex[1]{\Delta_{#1}}

\newcommand\Embi{\eta^+}
\newcommand\Embp{\eta^-}

\newcommand\Shvec[2]{{#1}\left\{{#2}\right\}}
\newcommand\ShvecBig[2]{{#1}\Big\{{#2}\Big\}}

\newcommand\Msetu[1]{\mathsf o(#1)}
\newcommand\Msetb[2]{\mathsf o(#1,#2)}

\newcommand\Mfinex[2]{\cM^-(#1,#2)}
\newcommand\Mfinr[2]{\Mfin{#1,#2}}

\newcommand\Canb[1]{\mathsf e_{#1}}

\newcommand\Prem[2]{\pi(#1,#2)}
\newcommand\Nrem[2]{\mu(#1,#2)}

\newcommand\Restrms[3]{\mathsf W(#1,#2,#3)}
\newcommand\Shleft[2]{\mathsf S(#1,#2)}

\newcommand\Listarg[3]{\,#1_{#2}\cdots#1_{#3}}

\newcommand\Rseg[2]{[#1,#2]}

\newcommand\Iftrans[2]{#1^\bullet_{#2}}

\newcommand\Bnfeq{\mathrel{\mathord:\mathord=}}
\newcommand\Bnfor{\,\,\mathord|\,\,}

\newcommand\PPCF{\mathsf{pPCF}}

\newcommand\Hempty{\ \ }
\newcommand\Hole[2]{[\Hempty]^{#1\vdash #2}}
\newcommand\Thole[3]{#1^{#2\vdash #3}}
\newcommand\Thsubst[2]{#1[#2]}

\newcommand\Tseqh[5]{\Tseq{#1}{\Thole{#2}{#3}{#4}}{#5}}

\newcommand\Locpcs[2]{{#1}_{#2}}

\newcommand\Deriv[1]{{#1}'}

\newcommand\Klfun[1]{\Fun{#1}}

\newcommand\Distsp[3]{\mathsf d_{#1}(#2,#3)}
\newcommand\Distspsymb[1]{\mathsf d_{#1}}
\newcommand\Distobs[2]{\mathsf d_{\mathsf{obs}}(#1,#2)}
\newcommand\Distobspsymb[1]{\mathsf d_{\mathsf{obs},#1}}
\newcommand\Distobsp[3]{\Distobspsymb{#1}(#2,#3)}

\newcommand\Fibloc[1]{\mathsf T#1}

\newcommand\Cuball[1]{\mathcal B#1}
\newcommand\Cuballp[1]{\mathcal B(#1)}
\newcommand\Cuballr[2]{\Cuball{\Crel{#1}{#2}}}

\newcommand\Cloc[2]{{#1}_{#2}}

\newcommand\Rem[1]{\widetilde{#1}}

\newcommand\Probared[2]{\mathbb{P}(#1\downarrow #2)}

\newcommand\Tamed[2]{#1\cdot #2}

\newcommand\Mark[2]{{#1}^{#2}}

\newcommand\Starg[1]{\mathsf{arg}(#1)}
\newcommand\Stsucc{\mathsf{succ}}
\newcommand\Stpred{\mathsf{pred}}
\newcommand\Stif[2]{\mathsf{if}(#1,#2)}
\newcommand\Stoneif[1]{\mathsf{seq}(#1)}
\newcommand\Stlet[2]{\mathsf{let}(#1,#2)}
\newcommand\State[2]{\langle#1,#2\rangle}
\newcommand\Stinit[2]{\mathsf I\langle#1,#2\rangle}
\newcommand\Stfin[2]{\mathsf F\langle#1,#2\rangle}
\newcommand\Stcons{\cdot}
\newcommand\Stempty{\epsilon}
\newcommand\Oempty{\Seqempty}

\newcommand\Stred[1]{\to^{#1}}
\newcommand\Stredd{\to}

\newcommand\Tone{\One}
\newcommand\Labels{\cL}

\newcommand\Onec{()}
\newcommand\Oneif[2]{#1\cdot#2}

\newcommand\Finseq[1]{{#1}^{<\omega}}

\newcommand\Ocons[2]{\Tuple{#1} #2}
\newcommand\Osingle[1]{\langle#1\rangle}
\newcommand\Ofinal[1]{\mathsf{O}(#1)}

\newcommand\Reltr[1]{\mathrel{\mathord{#1}^*}}
\newcommand\Kminus{\overset{.} -}

\newcommand\Lpcfst{\mathsf S_\Labels}
\newcommand\Evallst{\mathsf{Ev}_\Labels}

\newcommand\Pweight[2]{\mathsf w(#1,#2)}
\newcommand\Pcompl[2]{\mathsf c(#1,#2)}
\newcommand\Inistate[1]{\State{#1}{\Stempty}}

\newcommand\Nhead[2]{\mathsf H(#1,#2)}

\newcommand\Cantbopen[1]{\mathord\uparrow#1}
\newcommand\Len[1]{\mathsf{len}(#1)}

\newcommand\Pcfst{\mathsf S}
\newcommand\Evalst{\mathsf{Ev}}
\newcommand\Evdom[1]{\cD(#1)}

\newcommand{\Labext}{\mathsf{lab}}
\newcommand\PPCFlab{\PPCF_\Labext}
\newcommand\PCFlab{\PCF_\Labext}
\newcommand\Clabext{\mathsf{lc}}
\newcommand\PPCFlc{\PPCF_\Clabext}

\newcommand\Cantor{\cC}
\newcommand\Cantorfin{\Cantor_0}
\newcommand\Ospace{\Cantorfin}
\newcommand\Evalstlab{\Evalst_\Labext}
\newcommand\Evalstlc{\Evalst_\Clabext}
\newcommand\Evalstlcsh{\Evalst_\Clabext^\eta}
\newcommand\Evalstlcshinv{\Evalst_\Clabext^{-\eta}}
\newcommand\Pcfstlab{\Pcfst_\Labext}
\newcommand\Pcfstlc{\Pcfst_\Clabext}
\newcommand\Striplab[1]{\underline{#1}}
\newcommand\Striplc[1]{\underline{#1}}
\newcommand\Evdomlab[1]{\cD_\Labext(#1)}
\newcommand\Evdomlc[1]{\cD_\Clabext(#1)}
\newcommand\Complstlab{\mathsf c_\Labext}

\newcommand\Dicelab[2]{\operatorname{\mathsf{lcoin}}(#1,#2)}

\newcommand\Labof[1]{\mathsf{lab}(#1)}
\newcommand\Lcof[1]{\mathsf{lc}_{#1}}

\newcommand\Pcfspy[2]{\mathsf{sp}_{#1}(#2)}

\newcommand\Evalstf[1]{\mathsf{Ev}(#1)}
\newcommand\Evalstlabf[1]{\Evalstlab(#1)}
\newcommand\Evalstlabfp[2]{\Evalstlab(#1)_{#2}}

\newcommand\Eventconvz[1]{{#1}\downarrow\Num 0}

\newcommand\Expect[1]{\mathbb E(#1)}

\newcommand\Spy[1]{\mathsf{sp}_{#1}}

\newcommand\Tamedc[2]{{#1}^{\langle#2\rangle}}
\newcommand\Tdistobs[3]{\mathsf d^{\langle#1\rangle}_{\mathsf{obs}}(#2,#3)}
\newcommand\Tdistobsf[1]{\mathsf d^{\langle#1\rangle}}

\newcommand\Ldistobs[2]{\mathsf d_{\mathsf{lin}}(#1,#2)}

\newcommand\Undef{\mathord\uparrow}
\newcommand\Eset[1]{\{#1\}}

\newcommand\Pneg[2]{\nu_{#1}(#2)}

\newcommand\Zeroc[1]{0_{#1}}
\newcommand\Underc[1]{\underline{#1}}
\newcommand\Sseq[4]{#1\vdash #2:#3:#4}

\newcommand\Sfun{\mathbf S}
\newcommand\Scfun{\mathbf S_{\Into}}
\newcommand\Scflip{\widetilde{\mathsf c}}
\newcommand\Stuple[1]{\Tuple{#1}_{\Sfun}}
\newcommand\Sctuple[1]{\Tuple{#1}_{\Scfun}}
\newcommand\Smont{\mathsf L}
\newcommand\Scmont{\widetilde{\Smont}}
\newcommand\Sstr{\phi}
\newcommand\Sstrs{\Sstr^\Sym}
\newcommand\Sstrc{\Sstr^{\multimap}}
\newcommand\Sdiff{\partial}

\newcommand\Sdfun{\mathsf D}

\newcommand\Coh[3]{#2\coh_{#1}#3}
\newcommand\Scoh[3]{#2\scoh_{#1}#3}
\newcommand\Incoh[3]{#2\incoh_{#1}#3}
\newcommand\Sincoh[3]{#2\sincoh_{#1}#3}
\newcommand\Neu[3]{#2\equiv_{#1}#3}
\newcommand\Matappa[2]{{#1}\cdot{#2}}
\newcommand\Matappap[2]{\left({#1}\cdot{#2}\right)}

\newcommand\Saxcom{(\textbf{S-com})}
\newcommand\Saxzero{(\textbf{S-zero})}
\newcommand\Saxass{(\textbf{S-assoc})}
\newcommand\Saxwit{(\textbf{S-witness})}
\newcommand\Saxdist{(\textbf{S$\ITens$-dist})}
\newcommand\Saxfun{(\textbf{S$\ITens$-fun})}
\newcommand\Saxprod{(\textbf{S$\IWith$-pres})}

\newcommand\Csaxepi{(\textbf{CS-epi})}
\newcommand\Csaxsum{(\textbf{CS-sum})}
\newcommand\Ccsaxepi{(\textbf{CCS-epi})}
\newcommand\Ccsaxsum{(\textbf{CCS-sum})}
\newcommand\Csaxflip{(\textbf{CS-flip})}

\newcommand\Sprodisoz{\mathsf p^0}
\newcommand\Sprodisot{\mathsf p^2}

\newcommand\Sone{\One}
\newcommand\Sonelem{\ast}

\newcommand\Win[1]{\Inj{#1}^{\with}}
\newcommand\Wdiag{\Delta^\with}
\newcommand\Into{\mathsf I}

\newcommand\Diffst{\mathsf d}
\newcommand\Obj[1]{\mathsf{Obj(#1)}}

\newcommand\Fdom{\mathsf{dom}}

\newcommand\Diffofst[1]{{#1}^{+}}
\newcommand\Diffofste[1]{{\overline{#1}}^{+}}
\newcommand\Stofdiff[1]{{#1}^{-}}

\newcommand\Sdiffst{\widetilde\partial}

\newcommand\Diffrac[2]{\frac{d #1}{d #2}}
\newcommand\Intcc[2]{[#1,#2]}

\newcommand\Sfunadd{\sigma}

\newcommand\Symbdiff[1]{\mathsf d_{#1}}

\newcommand\Forec[2]{\mathsf{rec}(#1,#2)}
\newcommand\Trec[3]{\mathsf{rec}\,(#1:#2)\,#3}

\newcommand\Tinto{\boldsymbol\Delta_\Into}

\newcommand\Adj{\mathrel{\dashv}}

\newcommand\Ftunit{\eta}
\newcommand\Ftcounit{\epsilon}

\newcommand\Daxchain{($\partial$\textbf{-chain})}
\newcommand\Daxlocal{($\partial$\textbf{-local})}
\newcommand\Daxclocal{(\textbf{C}$\partial$\textbf{-local})}
\newcommand\Daxclin{(\textbf{C}$\partial$\textbf{-lin})}
\newcommand\Daxlin{($\partial$\textbf{-lin})}
\newcommand\Daxwith{($\partial$\textbf{-}$\mathord{\IWith}$)}
\newcommand\Daxcwith{(\textbf{C}$\partial$\textbf{-}$\mathord{\IWith}$)}
\newcommand\Daxschwarz{($\partial$\textbf{-Schwarz})}
\newcommand\Daxcschwarz{(\textbf{C}$\partial$\textbf{-Schwarz})}
\newcommand\Daxcchain{(\textbf{C}$\partial$\textbf{-chain})}
\newcommand\Tysub[2]{#1 - #2}
\newcommand\Tykind[1]{\overline{#1}}
\newcommand\Typrod[3]{\Pi\,(#1:#2)#3}

\newcommand\Cseq[1]{#1\vdash}
\newcommand\Tyseq[2]{#1\vdash #2}
\newcommand\Kindofty[1]{\underline{#1}}

\newcommand\Emptycont{(\,)}
\newcommand\Treesep{\quad\quad}

\newcommand\Sdfunit{\zeta}
\newcommand\Sdfmult{\tau}
\newcommand\Sdfstr{\psi}

\newcommand\Kllin{\mathsf{Lin}_{\mathord\oc}}

\newcommand\Tdiff[1]{\mathsf D{#1}}
\newcommand\Tdiffm[2]{\mathsf D^{#1}{#2}}
\newcommand\Tdnat[1]{\mathsf D^{#1}\Tnat}
\newcommand\Lproj[2]{\Sproj{#1}(#2)}
\newcommand\Lprojs[1]{\Sproj{#1}}
\newcommand\Linj[2]{\Sin{#1}(#2)}
\newcommand\Lsum[1]{\tau(#1)}
\newcommand\Lsums{\tau}
\newcommand\Ldlet[3]{\partial(#1,#2,#3)}
\newcommand\Lzero{0}
\newcommand\Lplus[2]{#1+#2}
\newcommand\Ldiff[1]{\Tdiff{#1}}
\newcommand\Ldiffp[1]{\Tdiff({#1})}
\newcommand\Lfix[1]{\mathsf Y#1}

\newcommand\Wseq[3]{#1\vdash_{w}#2:#3}

\renewcommand\Red{\to}

\newcommand\Tbundle[1]{\mathsf T#1}

\newcommand\Tonat{\Rightarrow}

\newcommand\ACOH{\mathbf{APcoh}}

\section*{Introduction}
\addcontentsline{toc}{section}{\protect\numberline{}Introduction}%


This article is a long version of a paper, with the same title and
by the same authors, accepted at the \emph{ACM/IEEE Symposium on Logic
  in Computer Science} 2023. In particular, all the proofs which are
missing in the conference version are provided in the present article.

\medbreak

Linear Logic (LL) and its models~\cite{Girard87} strongly suggest that
differentiation of proofs should be a natural operation
extracting their best ``local'' linear approximation.
Remember that for any \(E,F\) Banach spaces, \(f:E\to F\)
is differentiable at \(x \in E \) if there is a
neighborhood \(U\) of \(0\) in \(E\) and a linear and continuous
function \(\phi:E\to F\) such that, for all \(u\in U\)
\begin{align}
  \label{eq:derivative-def}
  f(x+u)=f(x)+\phi(u)+o(\|u\|)\,.
\end{align}
When \(\phi\) exists, it is unique and is denoted as
\(f'(x)\).
When \(f'(x)\) exists for all \(x\in E\),
the function \(f':E\to\mathcal L(E,F)\), where
\(\mathcal L(E,F)\) is the Banach space of linear and continuous
functions \(E\to F\), is called the \emph{differential} of \(f\).
This function can itself admit a differential and so on. When all
these iterated differentials exist one says that \(f\) is
\emph{smooth} and the \(n\)th derivative of \(f\) is a function
\(f^{(n)}:E\to\mathcal L_n(E,F)\) where \(\mathcal L_n(E,F)\) is the
space of \(n\)-linear symmetric functions \(E^n\to F\).
It can even happen that \(f\) is locally (or even globally) expressed
using its iterated derivatives by means of the \emph{Taylor Formula}
\(f(x+u)=\sum_{n=0}^\infty\frac1{n!}f^{(n)}(x)(u,\dots,u)\); when this
holds locally at any point \(x\), \(f\) is said to be \emph{analytic}.

Based on categorical models of LL where morphisms are 
analytic functions, the differential \(\lambda\)-calculus and
differential LL provide a logical and syntactical account of
differentiation. A program of type \(A\Rightarrow B\) can be turned
into a program of type \(A\Rightarrow(A\multimap B)\). This provides 
a new approach of finite approximations of functions by a
syntactical version of the Taylor Formula which has shown relevance in
the study of the \(\lambda\)-calculus and of LL.
Differentiation is deeply connected with \emph{addition}, as it can
already be seen in its definition~\cref{eq:derivative-def}.
This connection also appears when writing the differential
of \(f:\Real^n\to\Real\) as a sum of \emph{partial derivatives}:
  \[ f'(x_1,\dots,x_n)\cdot(u_1,\dots,u_n)=
  \sum_{i=1}^n\frac{\partial f(x_1,\dots,x_n)}{\partial x_i}u_i \]
and, of course, in the Taylor formula itself.
For this reason, until recently, all categorical models of the
differential \(\lambda\)-calculus and of 
differential LL~\cite{Blute06, Blute09} were using
categories where hom-sets have a structure of commutative monoid
and both formalisms feature a
formal and unrestricted addition operation on terms or proofs of the
same type.
The only available operational interpretation of such a sum being
erratic choice, these formalisms are inherently non-deterministic.

Recently, \Ehrhard{} observed~\cite{Ehrhard22a} that, in a setting where all coefficients
are non-negative, differentiation survives to strong restrictions on
the use of addition.
Consider for instance a function \([0,1]\to[0,1]\) which is smooth on
\([0,1)\) and all of whose iterated derivatives are everywhere \(\geq 0\)%
\footnote{This actually implies that \(f\) is analytic.}.
If \(x,u\in[0,1]\) are such that \(x+u\in[0,1]\) then
\(f(x)+f'(x)u\leq f(x+u)\in[0,1]\) (this makes sense even if
\(f'(1)=\infty\), which can happen: take \(f(x)=1-\sqrt{1-x}\)).
So if \(S\) is the set of all such pairs \((x,u)\) that we call
\emph{summable}, we can consider the function
\(\D(f):(x,u)\mapsto(f(x),f'(x)u)\) as a map \(S\to S\).
This basic observation is generalized in~\cite{Ehrhard21} to a wide
range of categorical models \(\cL\) of LL including coherence spaces,
probabilistic coherence spaces \emph{etc.}~where hom-sets have only a
\emph{partially defined} addition.
%
%
In these \emph{summable categories}, \(S\) becomes an endofunctor
\(\cL\to\cL\) equipped with an additional structure which allows to
define summability and (partial) sums in a very general way and turns
out to induce a monad.
Differentiation is then axiomatized as a distributive law between this
monad (similar to the tangent bundle monad of a tangent
category~\cite{Rosicky84}) and the resource comonad \(\oc\_\) of the
LL structure of the category%
\footnote{Which by the way needs not be a fully-fledged LL model.}
\(\cL\).
Indeed, this distributive law allows one to extend \(S\) to $\kleisliExp$ the Kleisli
category of \(\oc\_\) and this extension \(\D:\Kl\cL\to\Kl\cL\) turns
out to be a monad which has all the required properties of
differentiation.

It is well known that $\kleisliExp$ is a cartesian closed category,
and it can be interesting to drift away
from the LL structure of $\categoryLL$
by only looking at the structure of its Kleisli 
category. 
This is what happened with differentiation.
It was first axiomatized in a typical LL setting with additive
categories and differential categories~\cite{Blute06}.
It was then carried to the setting of cartesian categories with
left-additive categories and cartesian differential categories
(CDC)~\cite{Blute09}.
Unsurprisingly, the Kleisli categories of the former provide instances
of the latter, but cartesian differential categories cover a wider
range of models.
As mentionned in~\cite{Ehrhard21}, differential categories can be seen as 
a special instance of summable categories equipped with differentiation
(we will call those \emph{coherent differential categories}) in which 
addition is unrestricted.
Naturally, we can wonder if there is a notion of \emph{cartesian
  coherent differential categories}, that arise as the Kleisli
categories of coherent differential categories, and that generalize
CDC to a partial setting.

We provide a positive answer to this question.
We define coherent differentiation in an \emph{arbitrary category}, whose
morphisms are intuitively considered as smooth.
So we start from a category \(\cC\) equipped with a map%
\footnote{It will become a functor and even a monad later.} %
\(\D:\Ob\cC\to\Ob\cC\) given together with morphisms %
\(\Dproj_{0,X},\Dproj_{1,X},\Dsum_X\in\cC(\D(X),X)\) (for each
\(X\in\Obj\cC\)).
The intuition is that \(\D(X)\) is the object of summable pairs of
elements of \(X\), that \(\Dproj_i\) are the obvious projections and
that \(\Dsum\) computes the sums. We assume \(\Dproj_0,\Dproj_1\) to
be jointly monic and this is sufficient to say when
\(f_0,f_1\in\cC(X,Y)\) are summable: this is when there is a
necessarily unique \(h\in\cC(X,\D(Y))\) such that
\(\Dproj_i\circ h=f_i\) and when this holds we set
\(f_0+f_1=\Dsum\Comp h\).
Under suitable assumptions this very light structure suffices to equip
hom-sets of \(\cC\) with a structure of partial commutative monoid
which is compatible with composition on the left%
\footnote{And not on the right in general since, intuitively, the
  morphisms of \(\cC\) are not assumed to be linear.}. %

This structure is a convenient setting for differentiation: it
suffices to furthermore equip \(\D\) with a functorial action on
morphisms wrt.~which some morphisms (definable in terms of
\(\Dproj_0,\Dproj_1,\Dsum\)) are natural.
This is the notion of \emph{coherent differential category} whose
axioms are in one-to-one correspondence with those of a CDC.
Just as in tangent categories~\cite{Rosicky84,Cockett14}, our functor
\(\D\) can be equipped with a monad structure.
Contrarily to the additive framework of CDC
our differentiation functor \(\D\) is not defined in terms of the
cartesian product so it is important to understand how it interacts
with the cartesian product when available: this is formalized by the
concept of \emph{cartesian coherent differential category} (CCDC).
This compatibility can be expressed in terms of a strength with which
\(\D\) can be equipped, turning it into a commutative monad.
This induces a satisfactory theory of partial derivatives.
We provide a concrete example of such a category based on
probabilistic coherence spaces and illustrate our formalism by
interpreting a simple term language equipped with a notion of
differentiation in a CCDC.


\section{Left summability structure}

We introduce in this section the notion of \emph{left summability
  structure} in order to generalize the notion of summability
structure introduced in \cite{Ehrhard21} to a setting where morphisms
are not necessarily additive.

\subsection{Left pre-summability structures}

Let $\category$ be a category with objects $\objects$ and hom-set 
$\category(X, Y)$ for any $X, Y \in \objects$. We assume that any
hom-set $\category(X, Y)$ contains a distinguished morphism $0^{X,Y}$
(usually $X$ and $Y$ are kept implicit) such that for any
$f \in \category(Z, X)$, $0^{X,Y} \comp f = 0^{Z, Y}$.

\begin{definition} \label{def:pre-presummability-structure} A
  \emph{summable pairing structure} on a category \(\cC\) is a tuple
  $(\D, \Dproj_0, \Dproj_1, \Dsum)$ where: \begin{itemize}
  \item $\D: \objects \arrow \objects$ is a map (a functional class) on objects;
  \item
    $(\Dproj_{0, X})_{X \in \objects}, (\Dproj_{1, X})_{X \in
      \objects}$ and $(\Dsum_X)_{X \in \objects}$ are families of
    morphisms in $\category(\D X, X)$. 
    The object $X$ will usually be kept implicit;
  \item $\Dproj_0$ and $\Dproj_1$ are jointly monic: for any
    $f, g \in \category(Y, \D X)$, if
    $\Dproj_0 \comp f = \Dproj_0 \comp g$ and
    $\Dproj_1 \comp f = \Dproj_1 \comp g$ then $f = g$.
\end{itemize}
\end{definition}

We assume in what follows that $\category$ is equipped with a
summable pairing structure $(\D, \Dproj_0, \Dproj_1, \Dsum)$.


\begin{definition}
  Two morphisms $f_0, f_1 \in \category(X,
  Y)$ are said to be \emph{summable} if there exists $h \in
  \category(X, \D Y)$ such that $\Dproj_i \comp h = f_i$.
  The joint monicity of the $\Dproj_i$'s ensures that when
  $h$ exists, it is unique.
  We set $\Dpair{f_0}{f_1} \defEq
  h$, and we call it the \emph{witness} of the sum.
  By definition, $\Dproj_i \comp \Dpair{f_0}{f_1} = f_i$.
  Then we set $f_0 + f_1 \defEq \Dsum \comp \Dpair{f_1}{f_2}$.
\end{definition}

\begin{remark}
  A more standard approach to notations would be to write $\Dproj_1$ and 
  $\Dproj_2$ instead of $\Dproj_0$ and $\Dproj_1$.
  The reason we proceed that way is that \Cref{eq:derivative-def} will be 
  formalized in our setting with the use of a pair 
  $\Dpair{f(x)}{f'(x) \cdot u}$.
  That is, the left element of this pair is of order $0$, and 
  the right element is of order $1$.
\end{remark}

\begin{notation}\label{notation:sum-defined}
  We write $f_0 \summable f_1$ for the property that $f_0$ and $f_1$
  are summable.
  We say that an algebraic expression containing binary sums is
  \emph{well defined} if each pair of morphisms involved in these sums
  are summable.
  For example, $(f_0 + f_1) + f_2$ is well defined if
  $f_0 \summable f_1$ and $(f_0 + f_1) \summable f_2$.
\end{notation}

\begin{proposition}
  The morphism \label{prop:proj-sum} $\Dproj_0$ and $\Dproj_1$ are summable
  with witness $\Dpair{\Dproj_0}{\Dproj_1} = \id$ and sum
  $\Dproj_0 + \Dproj_1 = \Dsum$.
\end{proposition}

\begin{proof}
  $\Dproj_i \comp \id = \Dproj_i$ so by definition,
  $\Dproj_0 \summable \Dproj_1$ with witness $\id$ and sum
  $\Dsum \comp \id = \Dsum$.
\end{proof}

\begin{proposition}[Left compatibility of sum] %
  \label{prop:sum-left-compatible} %
  For any
  $f_0, f_1 \in \category(Y, Z)$ and $g \in \category(X, Y)$, if
  $f_0 \summable f_1$, then $(f_0 \comp g) \summable (f_1 \comp g)$ with
  witness
  $\Dpair{f_0 \comp g}{f_1 \comp g} =\Dpair{f_0}{f_1} \comp
  g$.
  Moreover, $(f_0 \comp g)+(f_1 \comp g) = (f_0 + f_1) \comp g$.
\end{proposition}

\begin{proof}
  Let $\witness = \Dpair{f_0}{f_1} \comp g$.
  Then $\Dproj_i \comp \witness = f_i \comp g$ so $\witness$ is a
  witness for the summability of $f_0 \comp g$ and $f_1 \comp
  g$.
  And
  $f_0 \comp g + f_1 \comp g \defEq \Dsum \comp \witness = (f_0 + f_1)
  \comp g$.
\end{proof}

An important class of morphisms is that of additive morphisms, for
which addition is compatible with composition on the right.

\begin{definition} \label{def:additive} %
  A morphism $h \in \category(Y, Z)$ is \emph{additive} if
  $h \comp 0 = 0$ and if for any $f_0, f_1 \in \category(X, Y)$, if
  $f_0 \summable f_1$ then $h \comp f_0 \summable h \comp f_1$ and
  $h \comp (f_0 + f_1) = h \comp f_0 + h \comp f_1$.
  Note that \(\id\) is additive and that the composition of two
  additive morphisms is an additive morphism.
\end{definition}

\begin{proposition} \label{prop:additive} %
  A morphism $h$ such that $h \comp 0 = 0$ is additive if and only if 
  $h \comp \Dproj_0 \summable h \comp \Dproj_1$ with sum
  $h \comp \Dsum$.
\end{proposition}

\begin{proof}
  For the forward implication, recall that
  $\Dproj_0 \summable \Dproj_1$ with sum $\Dsum$. Thus by additivity of $h$,
  $h \comp \Dproj_0 \summable h \comp \Dproj_1$ with sum
  $h \comp \Dsum$.
  For the reverse implication, assume that $f_0 \summable f_1$.
  Since $h \comp \Dproj_0 \summable h \comp \Dproj_1$,
  \Cref{prop:sum-left-compatible} ensures that
  $h \comp f_0 = h \comp \Dproj_0 \comp \Dpair{f_0}{f_1}$ and
  $h \comp f_1 = h \comp \Dproj_1 \comp \Dpair{f_0}{f_1}$ are
  summable, with sum
  $(h \comp \Dproj_0 + h \comp \Dproj_1) \comp \Dpair{f_0}{f_1} = h
  \comp \Dsum \comp \Dpair{f_0}{f_1} = h \comp (f_0 + f_1)$.
\end{proof}

\begin{definition} \label{def:presummability-structure} %
  The summable pairing structure $(\D, \Dproj_0, \Dproj_1, \Dsum)$ is
  a \emph{left pre-summability structure} if $\Dproj_0, \Dproj_1$ and
  $\Dsum$ are additive.
\end{definition}

The additivity of the projections implies that the sum behaves
well with respect to the operation
$\Dpair{\_}{\_}$ itself.

\begin{proposition} \label{prop:pair-sum} %
  Assume that $\Dproj_0$ and $\Dproj_1$ are additive. Then 
  for any $f_0, f_1, g_0, g_1 \in \category(X, Y)$,
  if $f_0 \summable f_1$, $g_0 \summable g_1$ and
  $\Dpair{f_0}{f_1} \summable \Dpair{g_0}{g_1}$, then
  $f_0 \summable g_0$, $f_1 \summable g_1$,
  $(f_0 + g_0) \summable (f_1 + g_1)$ and
  $\Dpair{f_0}{f_1} + \Dpair{g_0}{g_1} = \Dpair{f_0 + g_0}{f_1 +
    g_1}$.
\end{proposition}

\begin{proof}
  By additivity of $\Dproj_i$, $\Dproj_i \comp \Dpair{f_0}{f_1} = f_i$
  and $\Dproj_i \comp \Dpair{g_0}{g_1} = g_i$ are summable with sum
  $f_i + g_i = \Dproj_i \comp (\Dpair{f_0}{f_1} + \Dpair{g_0}{g_1})$.
  Since \(\Dproj_0\summable\Dproj_1\) this entails by \Cref{prop:sum-left-compatible}
  that $f_0 + g_0$,
  $f_1 + g_1$ are summable with witness
  $\Dpair{f_0}{f_1} + \Dpair{g_0}{g_1}$.
\end{proof}

The additivity of $\Dsum$ implies that whenever
$\Dpair{f_0}{f_1} \summable \Dpair{g_0}{f_1}$, one has
$\Dsum \comp \Dpair{f_0}{f_1} \summable \Dsum \comp \Dpair{g_0}{g_1}$
and
\[\Dsum \comp (\Dpair{f_0}{f_1} + \Dpair{g_0}{g_1}) = (\Dsum \comp
\Dpair{f_0}{f_1}) + (\Dsum \comp \Dpair{g_0}{g_1})\]
Assuming the additivity of the projections, the additivity of
\(\Dsum\) implies that whenever
$\Dpair{\Dpair{f_0}{f_1}}{\Dpair{g_0}{g_1}}$ exists, the two sums
below are well defined (see \Cref{notation:sum-defined}) and
\begin{equation} \label{eq:medial}
(f_0 + g_0) + (f_1 + g_1) = (f_0 + f_1) + (g_0 + g_1)\,.
\end{equation}


\begin{proposition} \label{prop:zero-additive} %
  The morphisms $0$ and $0$ are summable with witness $0$ and sum $0$.
  In particular, $0$ is additive.
\end{proposition}

\begin{proof}
  On the one hand, $\Dproj_i \comp 0 = 0$ by additivity of $\Dproj_i$,
  so $0 \summable 0$ with witness $0$.
  On the other hand, $\Dsum \comp 0 = 0$ by additivity of $\Dsum$ so
  $0+0 = 0$.
  In particular, $0$ is additive thanks to \Cref{prop:additive}
  because $0 \comp \Dproj_0 = 0$ and $0 \comp \Dproj_1 = 0$ are
  summable with witness $0$ and sum $0 = 0 \comp \Dsum$.
\end{proof}

\subsection{Left summability structures}

We consider a category $\category$ equipped with a left
pre-summability structure $(\D, \Dproj_0, \Dproj_1, \Dsum)$.
The goal of this section is to make $(\category(X, Y), +, 0)$
a \emph{partial} commutative monoid.
Similar structures appear in~\cite{Arbib80} or more recently
in~\cite{Hines13}, in a setting where sums can be infinitary.
Our partial monoids have only finite sums%
\footnote{Although the extension of the finite sum to an infinitary
  operations will have to be considered when dealing with
  fixpoints.}.
More crucially, the categorical notion of summability defined above is
essential for us whereas it is not categorically formalized in these
works.

\begin{definition} \labeltext{($\D$-com)}{ax:D-com} %
  The left pre-summability structure is \emph{commutative} if for any object
  $X$, $\Dproj_1, \Dproj_0 \in \category(\D X, X)$ are summable with sum
  $\Dsum$.
  Then we set
  $\Dsym=\Dpair{\Dproj_1}{\Dproj_0} \in \category(\D X, \D X)$ so that
  \(\Dproj_i\comp\Dsym=\Dproj_{1-i}\).
  This property is called \ref{ax:D-com}.
\end{definition}

\begin{proposition}[Commutativity] %
  The left pre-summability structure is commutative if and only if for
  any $f_0, f_1 \in \category(X, Y)$, if $f_0 \summable f_1$ then
  $f_1 \summable f_0$ and $f_0 + f_1 = f_1 + f_0$.
\end{proposition}

\begin{proof}
  For the direct implication, assume that $f_0 \summable f_1$.
  Then
  $\Dproj_i \comp \Dsym \comp \Dpair{f_0}{f_1} = \Dproj_{1-i} \comp
  \Dpair{f_0}{f_1} = f_{1-i}$ so $f_1 \summable f_0$ with witness
  $\Dsym \comp \Dpair{f_0}{f_1}$.
  Furthermore,
  $f_1 + f_0 = \Dsum \comp \Dsym \comp \Dpair{f_0}{f_1} = \Dsum \comp
  \Dpair{f_0}{f_1} = f_0 + f_1$.
  Conversely, $\Dproj_0 \summable \Dproj_1$ so by commutativity
  $\Dproj_1 \summable \Dproj_0$ and
  $\Dproj_1 + \Dproj_0 = \Dproj_0 + \Dproj_1 = \Dsum$.
\end{proof}

\begin{definition} \label{def:injections}
  \labeltext{($\D$-zero)}{ax:D-zero} %
  The left pre-summability structure \emph{has $0$ as a neutral
  element} if for any object $X$, $\id_X \summable 0$ and
  $0 \summable \id_X$ with sums equal to $\id_X$.
  We call this property \ref{ax:D-zero}.
  We define $\Dinj_0, \Dinj_1 \in \category(X, \D X)$ as
  $\Dinj_0 \defEq \Dpair{\id_X}{0}$ and
  $\Dinj_1 \defEq \Dpair{0}{\id_X}$.
\end{definition}

\begin{proposition}[Neutrality of \(0\)] %
  The left pre-summability structure has $0$ as a neutral element if and
  only if for any morphism $f \in \category(X, Y)$, $0 \summable f$,
  $f \summable 0$ and $f + 0 = 0 + f = f$.
\end{proposition}

\begin{proof}
  By definition of $\Dinj_0$,
  $\Dproj_0 \comp \Dinj_0 \comp f = \id \comp f = f$ and
  $\Dproj_1 \comp \Dinj_0 \comp f = 0 \comp f = 0$.
  So $f \summable 0$ with witness $\Dinj_0 \comp f$ and
  $f + 0 = \Dsum \comp \Dinj_0 \comp f = \id \comp f = f$.
  We do the same for $0 + f$ with $\Dinj_1$.
  Conversely, we apply the neutrality of 0 on $\id$ to get that
  $\id \summable 0$ and $0 \summable \id$, with sum $\id$.
\end{proof}

Associativity is not that straightforward, as there are two possible
notions.
The situation is similar in the infinitary setting of~\cite{Hines13}
with the distinction between Weak Partition Associativity and
Partition Associativity.

\begin{definition}[Weak Associativity] %
  The operation $+$ is called \emph{weakly associative} if whenever
  $(f_0 + f_1) + f_2$ and $f_0 + (f_1 + f_2)$ are well defined (recall
  \Cref{notation:sum-defined}), we have
  $(f_0 + f_1) + f_2 = f_0 + (f_1 + f_2)$.
\end{definition}

\begin{definition}[Associativity] \label{def:associative} %
  The operation $+$ is called \emph{associative} if whenever
  $(f_0 + f_1) + f_2$ or $f_0 + (f_1 + f_2)$ is well defined, the
  other expression is also well defined and
  $(f_0 + f_1) + f_2 = f_0 + (f_1 + f_2)$.
\end{definition}

We need to work in a partial setting in which addition is associative:
this is required for instance in \Cref{sec:differential} to define
$\DmonadSum = \Dpair{\Dproj_0 \comp \Dproj_0}{\Dproj_0 \comp \Dproj_1
  + \Dproj_1 \comp \Dproj_0}$.
This associativity seems related to a kind of positivity of morphisms.

\begin{example} %
  Let $x, y \in\Intcc{-1}{1}$ be summable when $|x|+|y| \leq 1$, with
  \(x+y\) as sum. %
  Then $+$ is weakly associative, but is not associative.
  Indeed, take $x_0 = -\frac{1}{2}, x_1 = \frac{1}{2}, y_1 = 1$.  Then
  $(x_0 + x_1) + y_1$ is defined, but $x_0 + (x_1 + y_1)$ is not since
  $|x_1| + |y_1| = \frac{3}{2} > 1$.
  However, the same definition on $[0, 1]$ yields an associative
  operation.
\end{example}

Recall from~\Cref{eq:medial} that whenever
$\Dpair{\Dpair{f_0}{f_1}}{\Dpair{g_0}{g_1}}$ exists, the expressions
\((f_0 + g_0) + (f_1 + g_1)\) and \((f_0 + f_1) + (g_0 + g_1)\) are
well defined and equal.  Taking $g_0 = 0$ and assuming
\ref{ax:D-zero}, this means that
whenever $\Dpair{\Dpair{f_0}{f_1}}{\Dpair{0}{g_1}}$ exists,
$(f_0 + f_1) + g_1$ and $f_0 + (f_1 + g_1)$ are well defined and
equal.
Taking $f_1 = 0$ and assuming \ref{ax:D-zero}, this means that 
whenever $\Dpair{\Dpair{f_0}{0}}{\Dpair{g_0}{g_1}}$ exist,
$ f_0 + (g_0 + g_1)$ and $(f_0 + g_0) + g_1$ are well defined and
equal.
Thus associativity holds if \ref{ax:D-zero} holds and if whenever
$(f_0 + f_1) + g_1$ is defined (respectively $f_0 + (g_0 + g_1)$ is
defined), then $\Dpair{\Dpair{f_0}{f_1}}{\Dpair{0}{g_1}}$ exists
(respectively $\Dpair{\Dpair{f_0}{0}}{\Dpair{g_0}{g_1}}$ exists).
This shows that associativity
follows from the following axiom.

\begin{definition} \labeltext{($\D$-witness)}{ax:D-witness} %
  The left pre-summability structure \emph{admits witnesses} if for
  any $f, g \in \category(Y, \D X)$, if
  $\Dsum \comp f \summable \Dsum \comp g$ then $f \summable g$.
  We call this property \ref{ax:D-witness}.
\end{definition}

\begin{theorem}
  The properties \ref{ax:D-zero}, \ref{ax:D-com} and
  \ref{ax:D-witness} give to $\category(X, Y)$ the structure of a
  \emph{partial commutative monoid} for any objects $X,Y$.
  That is, for any $f, f_0, f_1, f_2 \in \category(X, Y)$:
  \begin{itemize}
  \item $f \summable 0$, $0 \summable f$ and $0 + f = f + 0 = f$;
  \item If $f_0 \summable f_1$ then $f_1 \summable f_0$ and
    $f_0 + f_1 = f_1 + f_0$;
  \item If $(f_0 + f_1) + f_2$ or $f_0 + (f_1 + f_2)$ is defined, then
    both are defined and $(f_0 + f_1) + f_2 = f_0 + (f_1 + f_2)$.
  \end{itemize}
\end{theorem}

One can define inductively from this binary sum a notion of arbitrary
finite sum.
The empty family is always summable with sum $0$. The
family $(f_i)_{i \in I}$ for $I \neq \emptyset$ is summable if
$\exists i_0 \in I$ such that $(f_i)_{i \in I/\{i_0\}}$ is summable
and if $(\sum_{i \in I/\{i_0\}} f_i) \summable f_{i_0}$.
Then we set
$\sum_{i \in I} f_i \defEq \sum_{i \in I/\{i_0\}} f_i + f_{i_0}$.
\Cref{prop:arbitrary-sum} shown in \cite{Ehrhard21} ensures that 
the choice of order for the sum is
irrelevant.

\begin{theorem} \label{prop:arbitrary-sum}
  A family $(f_i)_{i \in I}$ is summable if and only if for all
  partition%
  \footnote{Where we admit that some \(I_j\)s can be empty.} %
  $I_1, \ldots, I_n$ of $I$, we have that for all
  $j \in \interval{1}{n} := \{1, \ldots, n \}$, $(f_i)_{i \in I_j}$ is
  summable and $(\sum_{i \in I_j} f_i)_{j \in \interval{1}{n}}$ is
  summable.
  Moreover,
  $\sum_{i \in I} f_i = \sum_{j \in \interval{1}{n}} \sum_{i \in I_j}
  f_i$.
\end{theorem}

\begin{definition} 
  \label{def:left-summability-struct}
  A \emph{left summability structure} is a left pre-summability
  structure $(\D, \Dproj_0, \Dproj_1, \Dsum)$ such that 
  \ref{ax:D-zero}, \ref{ax:D-com} and \ref{ax:D-witness} hold.
\end{definition}

\subsection{Comparison with summability structures}
\label{sec:comparison-ss}
In the $\LL{}$ setting of~\cite{Ehrhard21}, \Ehrhard{} introduced a
notion of pre-summability structure $(\S, \Sproj_0, \Sproj_1, \Ssum)$
as a summable pairing structure (recall
\Cref{def:pre-presummability-structure}) where $\S$ is a functor for
which $\Sproj_0, \Sproj_1, \Ssum$ are natural transformations.

\begin{theorem} \label{prop:summability-struct-equivalence} %
  The following are equivalent
  \begin{itemize}
  \item $(\S, \Sproj_0, \Sproj_1, \Ssum)$ is a left pre-summability
    structure and every morphism is additive;
  \item $(\S, \Sproj_0, \Sproj_1, \Ssum)$ is a pre-summability
    structure \cite{Ehrhard21}.
  \end{itemize}
\end{theorem}
Remember that in~\cite{Ehrhard21}, the underlying category
\(\category\) is assumed to be enriched over the monoidal category of
pointed sets, the distinguished element of \(\category (X,Y)\) being
denoted \(0\).
In particular \(f\Comp 0=0\) always holds.

\begin{proof}
  Let $(\S, \Sproj_0, \Sproj_1, \Ssum)$ be a left pre-summability
  structure in which every morphism is additive.
  By \Cref{prop:additive}, for any $f \in \category(X, Y)$ we can define
  $\S f \defEq \Spair{f \comp \Sproj_0}{f \comp \Sproj_1}$ and the following
  equations hold:
  $\Dproj_i \circ \S f = f \circ \Dproj_i$, $\Dsum \circ \S f
  = f \circ \Dsum$.
  Furthermore, $\S$ is a functor:
  $\Sproj_i \comp \S \id = \id \comp \Sproj_i = \Sproj_i \comp \id$
  and
  $\Sproj_i \comp \S f \comp \S g = f \comp \Sproj_i \comp \S g = f
  \comp g \comp \Sproj_i = \Sproj_i \comp \S (f \comp g)$.
  Thus, by joint monicity of the $\Dproj_i$, $\S \id = \id$ and
  $\S (f \comp g) = \S f \comp \S g$.  
  Then the equations $\Sproj_i \comp \S f = f \comp \Sproj_i$ and
  $\Ssum \comp \S f = f \comp \Ssum$ introduced above correspond
  to the naturality of $\Dproj_0, \Dproj_1$ and $\Dsum$.

  Conversely, let $(\S, \Sproj_0, \Sproj_1, \Ssum)$ be a
  pre-summability structure in the sense of~\cite{Ehrhard21}.
  The naturality of $\Sproj_0$ and $\Sproj_1$ ensures that for any
  $f$, $f \comp \Sproj_0 \summable f \comp \Sproj_1$ with witness
  $\S f$.
  The naturality of $\Ssum$ ensures that the sum of those two
  morphisms is $\Ssum \comp \S f = f \comp \Ssum$.  Finally,
  $f \comp 0 = 0$ by assumption. So every morphism is additive by
  \Cref{prop:additive}.
  In particular, $\Sproj_0, \Sproj_1$ and $\Ssum$ are additive, so
  $(\S, \Sproj_0, \Sproj_1, \Ssum)$ is a left pre-summability
  structure.
\end{proof}

\begin{cor}
  The summability structures of~\cite{Ehrhard21} are the left summability
  structures where all morphisms are additive.
\end{cor}

\section{Differential}

\subsection{Differential Structure}
\label{sec:diff-structure}

\label{sec:differential}

Recall from~\Cref{eq:derivative-def} the main idea of the differential
calculus.
We generalize it to a partial additive setting: $f$ is differentiable
at $x$ if for any $u$, if $x \summable u$ then $f'(x) \cdot u$ is
defined, $f(x) \summable f'(x) \cdot u$ and, intuitively,
$f(x + u) \simeq f(x) + f'(x) \cdot u$.
Hence the differential of $f$ can be seen as a function $\D f$ that
maps a pair of two summable elements $\Dpair{x}{u}$ to a pair of two
summable elements $\D f(x, u) = \Dpair{f(x)}{f'(x) \cdot u}$.

\begin{definition} \label{def:differential-struct} %
  A \emph{pre-differential structure} is a left summability structure
  $(\D, \Dproj_0, \Dproj_1, \Dsum)$ together with, for each
  \(X,Y\in\Obj\category\), an operator
  \(\category(X,Y)\to\category(\D X,\D Y)\), also denoted as \(\D\),
  and such that $\Dproj_0 \comp \D f = f \comp \Dproj_0$.
  We define the \emph{differential} of $f$ as
  $\dcoh f \defEq \Dproj_1 \comp \D f \in \category(\D X, Y)$.
  By our assumptions $\D f = \Dpair{f \comp \Dproj_0}{\dcoh f}$.
\end{definition}
At this point we do not assume $\D$ to be a functor, this will be
the Chain Rule. Then the equation
$\Dproj_0 \comp \D f = f \comp \Dproj_0$ will be the naturality of
$\Dproj_0$.
We can also introduce three families of morphisms $\DmonadSum$,
$\Dlift$ and $\Dswap$ whose naturality will correspond to some axioms
of differentiation.
This is very similar to what happens in tangent categories
\cite{Cockett14}, the difference being the structure of the functor
$\D$ itself%
\footnote{There might be a way to combine tangent categories and
  coherent differentiation in one notion allowing to axiomatize
  objects similar to manifolds where the tangent spaces have an
  addition of vectors which is only partially defined.
  The first step should be to develop convincing concrete
  examples of such objects, which might be related to the semantics of
  Type Theory.}.

The additivity of $\Dsum$ ensures that
  $\Dsum \comp \Dproj_0 \summable \Dsum \comp \Dproj_1$.  That is,
  $(\Dproj_0 \comp\Dproj_0 + \Dproj_1 \comp\Dproj_0) \summable
  (\Dproj_0\comp \Dproj_1 + \Dproj_1\comp \Dproj_1)$.
  By associativity, this implies that
  $((\Dproj_0\comp \Dproj_0 + \Dproj_1\comp \Dproj_0) + \Dproj_0\comp
  \Dproj_1) + \Dproj_1\comp \Dproj_1$ is well defined, so
  $(\Dproj_0\comp \Dproj_0 + \Dproj_1\comp \Dproj_0) + \Dproj_0\comp
  \Dproj_1$ is well defined.
  By associativity again,
  $\Dproj_0\comp \Dproj_0 + (\Dproj_1\comp \Dproj_0 + \Dproj_0\comp
  \Dproj_1)$ is well defined, so~\Cref{def:DmonadSum} below makes sense.

\begin{definition} \label{def:DmonadSum} %
  For any object $X$, define $\DmonadSum \in \category(\D^2 X, \D X)$
  as $\DmonadSum \defEq \Dpair{\Dproj_0 \comp \Dproj_0}{\Dproj_1 \comp
    \Dproj_0 + \Dproj_0 \comp \Dproj_1}$.
\end{definition}

  By \ref{ax:D-zero},
  $(\Dproj_0 + 0) \summable (0 + \Dproj_1)$ so by \ref{ax:D-witness}
  $\Dpair{\Dproj_0}{0} \summable \Dpair{0}{\Dproj_1}$. 
  So \Cref{def:Dlift} below makes sense.

\begin{definition} \label{def:Dlift}
 For any object $X$, define
$\Dlift \in \category(\D X, \D^2 X)$ as
$\Dlift \defEq \Dpair{\Dpair{\Dproj_0}{0}}{\Dpair{0}{\Dproj_1}}$.
\end{definition}

By \Cref{prop:sum-left-compatible} (left compatibility)
  $\Dproj_0 \comp (\Dproj_0 + \Dproj_1) \summable \Dproj_1 \comp
  (\Dproj_0 + \Dproj_1)$. By additivity of $\Dproj_0$ and
  $\Dproj_1$, it means that
  $(\Dproj_0 \comp \Dproj_0 + \Dproj_0 \comp \Dproj_1) \summable
  (\Dproj_1 \comp \Dproj_0 + \Dproj_1 \comp \Dproj_1)$.
  So by \ref{ax:D-witness}, 
  $\Dpair{\Dproj_0 \comp \Dproj_0}{\Dproj_0 \comp \Dproj_1} \summable 
  \Dpair{\Dproj_1 \comp \Dproj_0}{\Dproj_1 \comp \Dproj_1}$ and
  \Cref{def:Dswap} below makes sense.

\begin{definition} \label{def:Dswap}
For any object $X$, we can define $\Dswap \in \category(\D^2 X, \D^2 X)$ as
$\Dswap \defEq\Dpair{\Dpair{\Dproj_0 \comp \Dproj_0}{\Dproj_0 \comp \Dproj_1}}
      {\Dpair{\Dproj_1 \comp \Dproj_0}{\Dproj_1 \comp \Dproj_1}}$.
\end{definition}

It is probably easier to understand those morphisms by how they
operate on witnesses. 
This corresponds to \Cref{prop:family-on-pairs} below. 
The proof is a straightforward computation using the joint monicity of $\Dproj_0$ and
$\Dproj_1$.

\begin{proposition} \label{prop:family-on-pairs}
For any $x, u, v, w \in \category(U, X)$ such that
$\Dpair{\Dpair{x}{u}}{\Dpair{v}{w}}$ is defined, 
\begin{align*}
  \DmonadSum \comp \Dpair{\Dpair{x}{u}}{\Dpair{v}{w}} 
  &= \Dpair{x}{u+v} \\
  \Dswap \comp \Dpair{\Dpair{x}{u}}{\Dpair{v}{w}} 
  &= \Dpair{\Dpair{x}{v}}{\Dpair{u}{w}} \\
  \Dlift \comp \Dpair{x}{u} 
  &= \Dpair{\Dpair{x}{0}}{\Dpair{0}{u}}
\end{align*}
\end{proposition}

\begin{definition} %
  \label{def:CDC}
  \labeltext{(Dproj-lin)}{ax:Dproj-lin} %
  \labeltext{(D-chain)}{ax:D-chain} %
  \labeltext{(Dsum-lin)}{ax:Dsum-lin} %
  \labeltext{(D-add)}{ax:D-add} %
  \labeltext{(D-Schwarz)}{ax:D-schwarz} %
  \labeltext{(D-lin)}{ax:D-lin} %
  \labeltext{(D-struct)}{ax:D-struct} %
  A \emph{differential structure} is a pre-differential structure
  $(\D, \Dproj_0, \Dproj_1, \Dsum)$ where the following axioms
  hold, using the associated notation \(\dcoh f\) introduced
  in~\Cref{def:differential-struct}:
  \begin{enumerate}
  \item \ref{ax:Dproj-lin}
    $\dcoh \Dproj_0 = \Dproj_0 \comp \Dproj_1$,
    $\dcoh \Dproj_1  = \Dproj_1 \comp \Dproj_1$;
  \item \ref{ax:Dsum-lin} $\dcoh \Dsum  = \Dsum \comp \Dproj_1$, 
  $\dcoh 0 = 0$;
  \item \ref{ax:D-chain} $\D$ is a functor (Chain Rule);
  \item \ref{ax:D-add} $\Dinj_0, \DmonadSum$ are natural
    transformations 
    (additivity of the derivative);
  \item \ref{ax:D-lin} $\Dlift$ is a natural transformation (linearity
    of the derivatives);
  \item \ref{ax:D-schwarz} $\Dswap$ is a natural transformation
    (Schwarz Rule).
  \end{enumerate}
  A \emph{coherent differential category} is a category $\category$
  equipped with a differential structure.
\end{definition}

The axiom \ref{ax:Dproj-lin} corresponds to an important structural
property of $\D$ with regard to $\Dpair{\_}{\_}$.
The axiom \ref{ax:Dsum-lin} corresponds to the additivity of the
derivative operator, that is, $(f+g)' = f' + g'$.
The axiom \ref{ax:D-chain} corresponds to the Chain Rule of
the differential calculus.
The axiom \ref{ax:D-add} says that $u \mapsto f'(x) \cdot u$ is
additive.
The axiom \ref{ax:D-lin} says that
$u \mapsto f'(x) \cdot u$ is not only additive, but also equal to its
own derivative in 0. It is shown in Prop.~4.2 of~\cite{Cockett14} 
(in the left-additive setting of cartesian differential categories) 
that it implies that $u \mapsto f'(x) \cdot u$ is equal to its 
own derivative in any points.
The same reasoning can be generalized to our setting, 
but it would require too much technical development to be developed in this paper.
Finally, the axiom \ref{ax:D-schwarz} corresponds to the Schwarz Rule, 
that is, the second derivative $f''(x)$
(a bilinear map) is symmetric.
An account of these axioms as properties of $\dcoh$ can be found
in \Cref{sec:equivalence-lemmas} and might help the reader understand 
the ideas mentioned above.

\subsection{Linearity}

For the rest of this section, $\category$ is only assumed to be
equipped with a pre-differential structure.
Any use of an axiom of coherent differential categories will be made
explicit.

\begin{definition}[$\D$-linearity] \label{def:linear} %
  A morphism $f \in \category(X,Y)$ is \emph{$\D$-linear} if the
  following diagrams commute.
  \begin{center}
    \begin{tikzcd}
      \D X \arrow[r, "\D f"] \arrow[d, "\Dproj_1"'] & \D Y \arrow[d, "\Dproj_1"] \\
      X \arrow[r, "f"'] & Y
    \end{tikzcd}\quad
    \begin{tikzcd}
      \D X \arrow[r, "\D f"] \arrow[d, "\Dsum"'] & \D Y \arrow[d, "\Dsum"] \\
      X \arrow[r, "f"'] & Y
    \end{tikzcd}\quad
    \begin{tikzcd}
      X \arrow[rd, "0"] \arrow[d, "0"'] &  \\
      X \arrow[r, "f"'] & Y
    \end{tikzcd}
  \end{center}
\end{definition}

\begin{remark} \label{rem:linear} %
  The first diagram can also be written as
  $\dcoh(f) = f \comp \Dproj_1$ and means that
  $\D f = \Dpair{f \comp \Dproj_0}{f \comp \Dproj_1}$.
\end{remark}

\begin{proposition} \label{prop:linear-and-additive}
A morphism $f$ is $\D$-linear if and only if it is
additive and $\dcoh f = f \comp \Dproj_1$ 
(that is, $\D f = \Dpair{f \comp \Dproj_0}{f \comp \Dproj_1}$).
\end{proposition}

\begin{proof}
  Assume that $f$ is $\D$-linear. Then $f \comp 0 = 0$ and, 
  by \Cref{rem:linear}, $f \comp \Dproj_0 \summable f \comp \Dproj_1$ of
  witness $\D f$.
  Thus
  $f \comp \Dproj_0 + f \comp \Dproj_1 \defEq \Dsum \comp \D f = f
  \comp \Dsum$ by assumption.
  So $f$ is additive by~\Cref{prop:additive}, and
  $\dcoh f = f \circ \Dproj_1$ by assumption.
  Conversely, only the second diagram is not part of the assumptions.
  \begin{align*}
    \Dsum \comp \D f 
    &= (\Dproj_0 + \Dproj_1) \comp \D f \\
    &= \Dproj_0 \comp \D f + \Dproj_1 \comp \D f \quad \text{by \Cref{prop:sum-left-compatible}} \\
    &= f \comp \Dproj_0 + f \comp \Dproj_1 \quad \text{ by assumption}  \\
    &= f \comp (\Dproj_0 + \Dproj_1) = f \comp \Dsum \quad \text{ by additivity of $f$}
  \end{align*}
  Thus $f$ is $\D$-linear.
\end{proof}
Thus $\D$-linear morphisms are in particular additive.
As we will see, our notion of additive and $\D$-linear morphisms
ultimately coincides with that of~\cite{Blute09}, so this
distinction between additivity and linearity is as relevant as it is
in their setting.
%

\begin{cor} \label{prop:constructors-linear} %
  \ref{ax:Dproj-lin} is equivalent to the linearity of
  $\Dproj_0$ and $\Dproj_1$.
  \ref{ax:Dsum-lin} is equivalent to the linearity of $\Dsum$ and $0$.
\end{cor}

Thus $\D$-linear morphisms are special instances of additive
ones.
Our notion of additive and $\D$-linear morphisms
ultimately coincides with the one of~\cite{Blute09} thanks to
\Cref{prop:linear-equation} below, so this
distinction between additivity and linearity is as relevant as it is
in their setting.
\begin{proposition} \label{prop:linear-equation}
  Assuming \ref{ax:Dproj-lin}, \ref{ax:D-chain} and \ref{ax:D-add},
  any morphism $h \in \category(X, Y)$ such that
  $\dcoh h = h \comp \Dproj_1$ is additive, hence $\D$-linear.
\end{proposition}

\begin{proof} 
  The proof relies on 
  \Cref{prop:derivative-additive-zero,prop:derivative-additive-sum}
  of \Cref{sec:equivalence-lemmas}.
  If $h = \dcoh h \comp \Dproj_1$,
  then for any $g \in \category(Z, X)$, $h \comp g 
  = h \comp \Dproj_1 \comp \Dpair{0}{g} 
  = \dcoh h \comp \Dpair{0}{g}$.
  Thus, $h \comp 0 = \dcoh h \comp \Dpair{0}{0} = 0$ 
  by \Cref{prop:derivative-additive-zero}, and
$h \comp (f_0 + f_1) = 
\dcoh h \comp \Dpair{0}{f_0 + f_1} = \dcoh h \comp
\Dpair{0}{f_0} + \dcoh h \comp \Dpair{0}{f_1} 
= h \comp f_0 + h \comp f_1$ by
\Cref{prop:derivative-additive-sum} again. Thus, $h$ is additive, so
$h$ is $\D$-linear by \Cref{prop:linear-and-additive}.
\end{proof}

\label{sec:Dstruct-lin}

Thanks to \ref{ax:D-chain}, \ref{ax:Dproj-lin} and \ref{ax:Dsum-lin},
we can show that linear morphisms are closed under composition,
witnesses and sum.

\begin{proposition} \label{prop:composition-linear} %
  Assuming \ref{ax:D-chain}, $\D$-linear morphisms are closed under composition and 
  inverses.
\end{proposition}

\begin{proof} Easy verification using the functoriality of $\D$.
\end{proof}

\begin{proposition}[$\D$-linearity and pairing]
  \label{prop:pairing-linear}
  Assume \ref{ax:D-chain} and \ref{ax:Dproj-lin}.
  Assume that $h_0, h_1 \in \category(X, Y)$ are summable and both
  $\D$-linear.
  Then $\Dpair{h_0}{h_1}$ is $\D$-linear.
\end{proposition}

\begin{proof}
  Let us do the diagram involving $\Dsum$, the other two being very
  similar. By joint monicity of the $\Dproj_i$'s, it suffices to solve
  the diagram chase below for $i=0,1$.
  \begin{center}
    \begin{tikzcd}[column sep = large, row sep = 2em]
      \D X \arrow[r, "\D \Dpair{h_0}{h_1}"] \arrow[dd, "\Dsum"'] \arrow[rd, "\D h_i"'] \arrow[rdd, "(c)", phantom] \arrow[rd, "(a)" description, phantom, bend left = 25] & \D^2 Y \arrow[r, "\Dsum"] \arrow[d, "\D \Dproj_i"'] \arrow[rd, "(b)", phantom] & \D Y \arrow[d, "\Dproj_i"]  \\
      & \D Y \arrow[r, "\Dsum"]                                                        & Y                           \\
      X \arrow[rr, "\Dpair{h_0}{h_1}"'] \arrow[rru, "h_i"] & {} & \D Y
      \arrow[u, "\Dproj_i"']
    \end{tikzcd}
\end{center}
(a)~is a consequence of \ref{ax:D-chain}, (b)~is a consequence of
\ref{ax:Dproj-lin} and (c)~is the \(\D\)-linearity of $h_i$.
\end{proof}

\begin{proposition}
  Assuming \ref{ax:D-chain} and \ref{ax:Dproj-lin}, $\Dsum$ is
  $\D$-linear if and only if for all $h_0, h_1 \in \category(X, Y)$ summable
  and both $\D$-linear, $h_0 + h_1$ is $\D$-linear.
\end{proposition}

\begin{proof}
  Assume that $h_0, h_1$ are \(\D\)-linear.
  By \Cref{prop:pairing-linear}, $\Dpair{h_0}{h_1}$ is \(\D\)-linear
  so $h_0 + h_1 = \Dsum \comp \Dpair{h_0}{h_1}$ is \(\D\)-linear
  (\(\D\)-linearity is closed under composition).
  Conversely, $\Dsum = \Dproj_0 + \Dproj_1$ and $\Dproj_0$, $\Dproj_1$
  are \(\D\)-linear so $\Dsum$ is \(\D\)-linear.
\end{proof}

\begin{cor} \label{prop:everyone-linear} %
  Assuming \ref{ax:Dproj-lin}, \ref{ax:Dsum-lin} and \ref{ax:D-chain},
  $\Dinj_0, \Dinj_1, \Dswap, \Dlift, \DmonadSum$ are all $\D$-linear.
\end{cor}
\begin{proof}
  All these morphisms are obtained through pairing, sums and composition
  of $\D$-linear maps.
\end{proof}
On a side note, by \Cref{rem:linear} the $\D$-linearity of $\Dproj_i$
means that
$\D \Dproj_i = \Dpair{\Dproj_i \comp \Dproj_0}{\Dproj_i \comp
  \Dproj_1}$.
In particular, it implies that
$\Dswap = \Dpair{\D \Dproj_0}{\D \Dproj_1}$.
This is very useful because the differential of a pair can then be
obtained from the pair of the differentials.

\begin{proposition} \label{prop:pair-derivative} %
  Assume \ref{ax:Dproj-lin}, \ref{ax:D-chain}.
  Let $f_0, f_1 \in \category(X, Y)$ such that $f_0 \summable
  f_1$.
  Then $\D f_0 \summable \D f_1$ and
  $\Dpair{\D f_0}{\D f_1} = \Dswap \comp \D \Dpair{f_0}{f_1}$.
\end{proposition}

\begin{proof} 
  $\Dproj_i \comp \Dswap \comp \D \Dpair{f_0}{f_1} = \D \Dproj_i \comp
  \D \Dpair{f_0}{f_1} = \D f_i$.
\end{proof}

\subsection{The Differentiation Monad}

\begin{proposition}
  Assuming \ref{ax:D-chain}, \ref{ax:Dproj-lin} and \ref{ax:Dsum-lin},
  the following diagrams commute.
  \begin{center}
    \begin{tikzcd}
      \D X \arrow[r, "\D \Dinj_0"] \arrow[rd, "\id_{\D X}"'] & \D^2 X \arrow[d, "\DmonadSum" description] & \D X \arrow[l, "\Dinj_0"'] \arrow[ld, "\id_{\D X}"] \\
      & \D X &
    \end{tikzcd}\quad
    \begin{tikzcd}
      \D^3 X \arrow[r, "\D \DmonadSum_{X}"] \arrow[d, "\DmonadSum_{\D X}"'] & \D^2 X \arrow[d, "\DmonadSum_X"] \\
      \D^2 X \arrow[r, "\DmonadSum_X"'] & \D X
    \end{tikzcd}
  \end{center}
\end{proposition}

\begin{proof}
  By \Cref{prop:everyone-linear}, $\Dinj_0$ is $\D$-linear.
  Thus by \Cref{rem:linear},
  $\D \Dinj_0 = \Dpair{\Dinj_0 \comp \Dproj_0}{\Dinj_0 \comp \Dproj_1}
  = \Dpair{\Dpair{\Dproj_0}{0}}{\Dpair{\Dproj_1}{0}}$.
  Hence
  $\DmonadSum \comp \D \Dinj_0 = \Dpair{\Dproj_0}{0 + \Dproj_1} =
  \Dpair{\Dproj_0}{\Dproj_1} = \id_{\D X}$ by
  \Cref{prop:family-on-pairs}.
  Next
  $\Dinj_0^{\D X} = \Dpair{\Dpair{\Dproj_0}{\Dproj_1}}{\Dpair{0}{0}}$
  since $\Dpair{\Dproj_0}{\Dproj_1} = \id$ and
  $\Dpair{0^{X,X}}{0^{X,X}} = 0^{\D X, \D X}$.
  By \Cref{prop:family-on-pairs} again,
  $\DmonadSum \comp \Dinj_0 = \Dpair{\Dproj_0}{\Dproj_1 + 0} =
  \Dpair{\Dproj_0}{\Dproj_1} = \id_{\D X}$ so the triangles commute.

  The square is a direct computation.
  We use simple juxtaposition for the composition of projections for the sake 
  of readability.
  The bottom path can be reduced using left compatibility of
  addition (\Cref{prop:sum-left-compatible}) and additivity of the
  projections:
  \begin{align*}
    \DmonadSum \comp \DmonadSum
    &= \Dpair{\Dproj_0 \comp \Dproj_0 \comp \DmonadSum}{ \Dproj_1 \comp \Dproj_0 \comp \DmonadSum + \Dproj_0 \comp \Dproj_1 \comp \DmonadSum} \\
    &= \Dpair{\Dproj_0\comp \Dproj_0\comp \Dproj_0}{\Dproj_1\comp \Dproj_0\comp \Dproj_0 + \Dproj_0 \comp (\Dproj_1\comp  \Dproj_0 + \Dproj_0\comp \Dproj_1)} \\
    &=  \Dpair{\Dproj_0\comp \Dproj_0\comp \Dproj_0}{\Dproj_1\comp \Dproj_0\comp  \Dproj_0 + (\Dproj_0\comp \Dproj_1\comp \Dproj_0 + \Dproj_0\comp \Dproj_0\comp \Dproj_1)}  \,.
  \end{align*}
  The upper path can be reduced by $\D$-linearity of $\DmonadSum$ and
  left compatibility of sum (\Cref{prop:sum-left-compatible}):
  \begin{align*}
    \DmonadSum \comp \D \DmonadSum 
    &= \Dpair{\Dproj_0\comp \Dproj_0 \comp \D \DmonadSum}{ \Dproj_1\comp \Dproj_0 \comp \D \DmonadSum + \Dproj_0 \Dproj_1 \comp \D \DmonadSum}\\
    &= \Dpair{\Dproj_0 \comp \DmonadSum \comp \Dproj_0}{ \Dproj_1 \comp \DmonadSum \comp \Dproj_0 + \Dproj_0 \comp  \DmonadSum \comp \Dproj_1} \\
    &= \Dpair{\Dproj_0\comp \Dproj_0\comp \Dproj_0}{(\Dproj_1\comp \Dproj_0 + \Dproj_0\comp \Dproj_1) \comp \Dproj_0 + \Dproj_0\comp \Dproj_0\comp \Dproj_1)} \\
    &=  \Dpair{\Dproj_0\comp \Dproj_0\comp \Dproj_0}{(\Dproj_1\comp \Dproj_0\comp \Dproj_0 + \Dproj_0\comp \Dproj_1\comp \Dproj_0) + \Dproj_0\comp \Dproj_0\comp \Dproj_1)} \,.
  \end{align*}
  We conclude that those two morphisms are equal, using the
  associativity of the partial sum.
\end{proof}

\begin{cor}
  \ref{ax:Dproj-lin}, \ref{ax:Dsum-lin}, \ref{ax:D-chain} and
  \ref{ax:D-add} imply that $(\D,\Dinj_0,\DmonadSum)$ is a monad.
\end{cor}

\section{Interpreting the axioms as properties of the derivative}

\label{sec:equivalence-lemmas}

In this section, $\category$ is only assumed to be a category equipped
with a pre-differential structure~(\Cref{def:differential-struct}).
We show that the various axioms of a coherent
differential category correspond to standard rules of the differential
calculus, written as properties about $\dcoh(f)$.
The results of this section are only necessary for \Cref{sec:CDC} but
they also provide some intuitions on the axioms of coherent
differentiation.

All the proofs are similar, and consist in using the joint monicity
of $\Dproj_0$ and $\Dproj_1$ to reduce the axioms to a set of
equations, then show that only one of those equations is non trivial.
In what follows, ``linear'' always means \(\D\)-linear.

\begin{proposition} \label{prop:D-chain} %
  $\D$ is a functor if and only if $\dcoh (\id) = \Dproj_1$ and
  $\dcoh (g \comp f) = \dcoh(g) \comp \Dpair{f \comp
    \Dproj_0}{\dcoh(f)}$.
\end{proposition}

\begin{proof}
  $\D$ is a functor if and only if $\D \id_X = \id_{\D X}$ and for any $g, f$,
  $\D (g \comp f) = \D g \comp \D f$.
  By joint monicity of the $\Dproj_i$, $\D \id = \id$ if and only if
  $\Dproj_i \circ \D \id = \Dproj_i \circ \id = \Dproj_i$. 
  But $\Dproj_0 \circ \D \id = \id \circ \Dproj_0 = \Dproj_0$ by assumptions
  on Pre-Differential Structures. 
  So $\D \id = \id$ if and only if $\Dproj_1 \circ \D \id = \Dproj_1$, that is,
  if and only if $\dcoh(\id) = \Dproj_1$.
  
  Similarly,
  $\Dproj_0 \comp \D g \comp \D f = g \comp \Dproj_0 \comp \D f = g
  \comp f \comp \Dproj_0 = \Dproj_0 \comp \D (g \comp f)$ by assumption on
  pre-differential-structures. So by joint
  monicity of the $\Dproj_i$, $\D (g \comp f) = \D g \comp \D f$ if
  and only if
  $\Dproj_1 \comp \D (g \comp f) = \Dproj_1 \comp \D g \comp \D f$.
  By definition of $\dcoh$, this
  corresponds exactly to the equation
  $\dcoh (g \comp f) = \dcoh(g) \comp \D f = \dcoh(g) \comp \Dpair{f
    \comp \Dproj_0}{\dcoh(f)}$
\end{proof}

\begin{proposition} \label{prop:D-sum-com}
  Assuming \ref{ax:Dproj-lin}, $\Dsum$ is linear if and only if
  $\D \Dsum = \D \Dproj_0 + \D \Dproj_1$.
  Assuming \ref{ax:Dproj-lin} and \ref{ax:D-chain}, $\Dsum$ is linear
  if and only if for any $f_0, f_1$ that are summable,
  $\D(f_0+f_1) = \D f_0 + \D f_1$ (recall that $\D f_0 \summable \D f_1$ by
  \Cref{prop:pair-derivative}).
\end{proposition}

\begin{proof}
  By linearity of $\Dproj_i$,
  $\D \Dproj_i = \Dpair{\Dproj_i \comp \Dproj_0}{\Dproj_i \comp
    \Dproj_1}$ so by \Cref{prop:pair-sum},
  $\D \Dproj_0 + \D \Dproj_1 = \Dpair{\Dproj_0 \comp \Dproj_0 +
    \Dproj_1 \comp \Dproj_0} {\Dproj_0 \comp \Dproj_1 + \Dproj_1 \comp
    \Dproj_1} = \Dpair{(\Dproj_0 + \Dproj_1) \comp \Dproj_0}{(\Dproj_0
    + \Dproj_1) \comp \Dproj_1} = \Dpair{\Dsum \comp \Dproj_0}{\Dsum
    \comp \Dproj_1}$.
  But $\Dsum$ is linear if and only if
  $\D \Dsum = \Dpair{\Dsum \comp \Dproj_0}{\Dsum \comp \Dproj_1}$ by
  \Cref{prop:linear-and-additive}, that is, if and only if
  $\D \Dsum = \D \Dproj_0 + \D \Dproj_1$.

  For the second part of the lemma, notice that the right statement
  for $f_0 = \Dproj_0$ and $f_1 = \Dproj_1$ is exactly
  $\D \Dsum = \D \Dproj_0 + \D \Dproj_1$, so the converse direction
  holds. For the forward direction, notice that
  \begin{align*}
    \D (f_0 + f_1)
    &= \D (\Dsum \comp \Dpair{f_0}{f_1}) \\
    &= \D \Dsum \comp \D \Dpair{f_0}{f_1}  \text{\quad by \ref{ax:D-chain}} \\
    &= (\D \Dproj_0 + \D \Dproj_1) \comp \D \Dpair{f_0}{f_1}  \text{\quad by assumptions}  \\
    &= \D \Dproj_0  \comp \D \Dpair{f_0}{f_1}
      + \D \Dproj_1  \comp \D \Dpair{f_0}{f_1}    \\
    &= \D f_0 + \D f_1  \text{\quad by \ref{ax:D-chain}}
\end{align*}
\end{proof}

\begin{cor} \label{prop:Dsum-lin} %
  Assuming \ref{ax:Dproj-lin} and \ref{ax:D-chain}, $\Dsum$ is linear
  if and only if for any $f_0, f_1$ that are summable,
  $\dcoh (f_0+f_1) = \dcoh (f_0) + \dcoh (f_1)$
\end{cor}

\begin{proof}
  The linearity of $\Dsum$ is equivalent to
  $\D (f_0 + f_1) = \D f_0 + \D f_1$ for any $f_0, f_1$ summable.
  By \Cref{prop:pair-sum}, this is equivalent to
  $\Dpair{(f_0 + f_1) \comp \Dproj_0}{\dcoh (f_0 + f_1)} = \Dpair{f_0
    \comp \Dproj_0 + f_1 \comp \Dproj_0}{\dcoh (f_0) + \dcoh (f_1)}$.
  The left compatibility of addition (\Cref{prop:sum-left-compatible})
  ensures that the first coordinates are always equal.
  So $\Dsum$ is linear if and only if for all $f_0 \summable f_1$,
  $\dcoh (f_0+f_1) = \dcoh (f_0) + \dcoh (f_1)$.
\end{proof}

\begin{proposition} \label{prop:derivative-additive-zero} %
  The following assertions are equivalent:
  \begin{itemize}
  \item[(1)] $\Dinj_0$ is natural;
  \item[(2)] For any $f \in \category (X, Y)$,
    $ \dcoh f \comp \Dinj_0 = 0$;
  \item[(3)] For any $f \in \category (X, Y)$, any object $U$ and
    $x \in \category(U, X)$, $\dcoh f \comp \Dpair{x}{0} = 0$.
  \end{itemize}
\end{proposition}

\begin{proof}
  (1) $\Equiv$ (2).
  By joint monicity of the $\Dproj_i$, for any
  $f \in \category(X, Y)$, $\D f \comp \Dinj_0 = \Dinj_0 \comp f$ if
  and only if
  $\Dproj_0 \comp \D f \comp \Dinj_0 = \Dproj_0 \comp \Dinj_0 \comp f
  = f$ and
  $\Dproj_1 \comp \D f \comp \Dinj_0 = \Dproj_1 \comp \Dinj_0 \comp f
  = 0$.
  The first condition always hold by naturality of $\Dproj_0$ and
  definition of $\Dinj_0$.
  So $\Dinj_0$ is natural if and only if the second identity holds.
  This equation is precisely~(2).

  (2) $\Equiv$ (3).
  The forward direction is directly obtained by composing the identity
  of $(2)$ by $x$ on the right.
  The reverse is directly obtained by applying the equation of~(3) to
  $x = \id_X$.
\end{proof}

\begin{proposition} \label{prop:derivative-additive-sum} %
  Assuming \ref{ax:Dproj-lin} and \ref{ax:D-chain}, the following
  assertions are equivalent:
  \begin{itemize}
  \item[(1)] $\DmonadSum$ is natural;
  \item[(2)] for any $f \in \category (X, Y)$,
    $\dcoh f \comp \D \Dproj_0 \summable \dcoh f \comp \Dproj_0 $ and
    $ \dcoh f \comp \DmonadSum = \dcoh f \comp \D \Dproj_0 + \dcoh f
    \comp \Dproj_0 $;
  \item[(3)] for any $f \in \category (X, Y)$, any object $U$ and any
    $x, u, v \in \category(U, X)$ that are summable,
    $\dcoh f \comp \Dpair{x}{u} \summable \dcoh f \comp \Dpair{x}{v}$ and
    \[
      \dcoh f \comp \Dpair{x}{u+v}
      = \dcoh f \comp \Dpair{x}{u} + \dcoh f \comp \Dpair{x}{v}\,.
    \]
  \end{itemize}
\end{proposition}

\begin{proof}
  (1) $\Equiv$ (2).
  By joint monicity of the $\Dproj_i$, for any
  $f \in \category(X, Y)$,
  $\D f \comp \DmonadSum = \DmonadSum \comp \D^2 f$ if and only if
  $\Dproj_0 \comp \D f \comp \DmonadSum = \Dproj_0 \comp \DmonadSum
  \comp \D^2 f$ and
  $\Dproj_1 \comp \D f \comp \DmonadSum = \Dproj_1 \comp \DmonadSum
  \comp \D^2 f$.
  The equation  $\Dproj_0 \comp \D f \comp \DmonadSum = 
  \Dproj_0 \comp \DmonadSum \comp \D^2 f$ always holds. Indeed
  \begin{align*}
    \Dproj_0 \comp \D f \comp \DmonadSum
    &= f \comp \Dproj_0 \comp \DmonadSum
       \text{\quad by naturality of $\Dproj_0$}   \\
    &= f \comp \Dproj_0 \comp \Dproj_0
       \text{\quad by definition of $\DmonadSum$}  \\
    \Dproj_0 \comp \DmonadSum \comp \D^2 f
    &= \Dproj_0 \comp \Dproj_0 \comp \D^2 f
        \text{\quad by definition of $\DmonadSum$}  \\
    &= f \comp \Dproj_0 \comp \Dproj_0
       \text{\quad by naturality of $\Dproj_0$} \, .
 \end{align*}
 The left hand side of the equation 
 $\Dproj_1 \comp \D f \comp \DmonadSum = 
 \Dproj_1 \comp \DmonadSum \comp \D^2 f$ is
 $\dcoh(f) \circ \DmonadSum$ by definition.
 The right hand side rewrites as follows.
 \begin{align*} 
    \Dproj_1 \comp \DmonadSum \comp \D^2 f
   &= (\Dproj_0 \comp \Dproj_1 + \Dproj_1 \comp \Dproj_0) \comp \D^2 f
   &  \\
   &= \Dproj_0 \comp \Dproj_1 \comp \D^2 f
     + \Dproj_1 \comp \Dproj_0 \comp \D^2 f
    \text{\quad by \Cref{prop:sum-left-compatible}}  \\
   &= \Dproj_1 \comp \D \Dproj_0 \comp \D^2 f
     + \Dproj_1 \comp \Dproj_0 \comp \D^2 f
    \text{\quad by $\D$-linearity of $\Dproj_0$} \\
   &= \Dproj_1 \comp \D (\Dproj_0 \comp \D f)
     + \Dproj_1 \comp \Dproj_0 \comp \D^2 f
    \text{\quad by \ref{ax:D-chain}}  \\
   &= \Dproj_1 \comp \D (f \comp \Dproj_0)
     + \Dproj_1 \comp \D f \comp \Dproj_0
     \text{\quad by naturality of $\Dproj_0$}  \\
   &= \Dproj_1 \comp \D f \comp \D \Dproj_0
     + \Dproj_1 \comp \D f \comp \Dproj_0
    \text{\quad by \ref{ax:D-chain}} \\
   &= \dcoh f \comp \D \Dproj_0 + \dcoh f \circ \Dproj_0
    \text{\quad by definition}
 \end{align*}
 So this second equation under consideration is equivalent to the
 equation of~(2).

 (2) $\Equiv$ (3).
 Recall that $\D \Dproj_0 = \Dpair{\Dproj_0 \comp \Dproj_0}{\Dproj_0 \comp \Dproj_1}$
 by linearity of $\Dproj_0$.
 Then the forward direction is directly obtained by composing the equation
 of~(2) with $\Dpair{\Dpair{x}{v}}{\Dpair{u}{0}}$ on the right.
 The converse is directly obtained by applying the equation of~(3) to
 $x = \Dproj_0 \comp \Dproj_0$, $u = \Dproj_1 \comp \Dproj_0$ and
 $v = \Dproj_0 \comp \Dproj_1$.
 %
\end{proof}

\begin{remark} \label{rem:dcoh-dcoh} %
  Notice that
  $\dcoh(\dcoh(f)) = \Dproj_1 \comp \D (\Dproj_1 \comp \D f) =
  \Dproj_1 \comp \D \Dproj_1 \comp \D^2 f = \Dproj_1 \comp \Dproj_1
  \comp \D^2 f$ assuming \ref{ax:D-chain} and
  \ref{ax:Dproj-lin}.
  Thus, $\dcoh(\dcoh(f))$ is nothing more than the rightmost coordinate of
  $\D^2 f$. This will be useful for what follows in this part.
\end{remark}

\begin{proposition} \label{prop:D-lin} %
  Assuming \ref{ax:Dproj-lin}, \ref{ax:D-chain} and the naturality of
  $\Dinj_0$, the following assertions are
  equivalent:
  \begin{enumerate}
  \item $\Dlift$ is natural;
  \item for all morphism $f \in \category(X, Y)$,
    $\dcoh(\dcoh (f)) \comp \Dlift = \dcoh (f)$;
  \item for all morphism $f \in \category(X, Y)$, for all morphisms
    $x, u \in \category(U, X)$ summable,
    \[
      \dcoh(\dcoh(f)) \comp \Dpair{\Dpair{x}{0}}{\Dpair{0}{u}} =
      \dcoh(f) \comp \Dpair{x}{u}\,.
    \]
  \end{enumerate}
\end{proposition}

\begin{proof}
  By joint monicity of the $\Dproj_i$, $\Dlift$ is natural if and only
  if for all $f$ and for all
  $i, j \in \{0, 1\}, \Dproj_i \comp \Dproj_j \comp \D^2 f \comp
  \Dlift = \Dproj_i \comp \Dproj_j \comp \Dlift \comp \D f$.
  By \Cref{rem:dcoh-dcoh} (and because
  $\Dproj_1 \comp \Dproj_1 \comp \Dlift = \Dproj_1$), the equation for
  $i = j = 1$ corresponds exactly to the equation
  $\dcoh(\dcoh (f)) \comp \Dlift = \dcoh (f)$.
  Thus, it suffices to show that
  $\Dproj_i \comp \Dproj_j \comp \D^2 f \comp \Dlift = \Dproj_i \comp
  \Dproj_j \comp \Dlift \comp \D f$ always holds when
  $(i, j) \neq (1, 1)$ to conclude that $(1)$ is equivalent to $(2)$.
  \begin{itemize}
  \item Case $i = 0, j = 0$:
    $\Dproj_0 \comp \Dproj_0 \comp \Dlift \comp \D f = \Dproj_0 \comp
    \D f = f \comp \Dproj_0$ and
    $\Dproj_0 \comp \Dproj_0 \comp \D^2 f \comp \Dlift = f \comp
    \Dproj_0 \comp \Dproj_0 \comp \Dlift = f \comp \Dproj_0$;
  \item Case $i = 1, j = 0$:
    $\Dproj_0 \comp \Dproj_1 \comp \Dlift \comp \D f = 0 \comp \D f =
    0$ and
    $\Dproj_1 \comp \Dproj_0 \comp \D^2 f \comp \Dlift = \Dproj_1
    \comp \D f \comp \Dproj_0 \comp \Dlift = \Dproj_1 \comp \D f \comp
    \Dinj_0 \comp \Dproj_0 = \Dproj_1 \comp \Dinj_0 \comp f \comp
    \Dproj_0 = 0$ thanks to the naturality of $\Dinj_0$;
  \item Case $i = 0, j = 1$:
    $\Dproj_0 \comp \Dproj_1 \comp \Dlift \comp \D f = 0 \comp \D f =
    0$ and
    $\Dproj_0 \comp \Dproj_1 \comp \D^2 f \comp \Dlift = \Dproj_1
    \comp \D \Dproj_0 \comp \D^2 f \comp \Dlift = \Dproj_1 \comp \D f
    \comp \D \Dproj_0 \comp \Dlift = \Dproj_1 \comp \D f \comp \Dinj_0
    \comp \Dproj_0 = \Dproj_1 \comp \Dinj_0 \comp f \comp \Dproj_0 =
    0$ thanks to the naturality of $\Dinj_0$.
\end{itemize}

Next~(2) is a particular case of~(3) for $x = \Dproj_0 $ and
$u = \Dproj_1$.  
Conversely, assuming~(2) we have that
$\dcoh(\dcoh(f)) \comp \Dpair{\Dpair{x}{0}}{\Dpair{0}{u}} =
\dcoh(\dcoh(f)) \comp \Dlift \comp \Dpair{x}{u} = \dcoh(f) \comp
\Dpair{x}{u}$. 
\end{proof}

\begin{proposition} \label{prop:D-schwarz} %
  Assuming \ref{ax:Dproj-lin} and \ref{ax:D-chain}, the following
  assertions are equivalent: 
  \begin{enumerate}
  \item $\Dswap$ is natural;
  \item for all morphism $f \in \category(X, Y)$,
    $\dcoh(\dcoh (f)) \comp \Dswap = \dcoh(\dcoh (f))$;
  \item for all morphism $f \in \category(X, Y)$ and
    $x, u, v, w \in \category(U, X)$ that are summable,
    \begin{align*}
       \dcoh(\dcoh (f)) \comp \Dpair{\Dpair{x}{u}}{\Dpair{v}{w}} 
      = 	\dcoh(\dcoh (f)) \comp \Dpair{\Dpair{x}{v}}{\Dpair{u}{w}}
    \end{align*}
  \end{enumerate}
\end{proposition}

\begin{proof}
  By joint monicity of the $\Dproj_i$, $\Dswap$ is natural if and only
  if for all $f$ and for all
  $i, j \in \{0, 1\}, \Dproj_i \comp \Dproj_j \comp \D^2 f \comp
  \Dswap = \Dproj_i \comp \Dproj_j \comp \Dswap \comp \D^2 f$.
  But
  $\Dproj_i \comp \Dproj_j \comp \Dswap \comp \D^2 f = \Dproj_j \comp
  \Dproj_i \comp \D^2 f$.
  Then, by \Cref{rem:dcoh-dcoh}, the equation for $i = j = 1$
  corresponds exactly to the equation
  $\dcoh(\dcoh (f)) \comp \Dswap = \dcoh(\dcoh(f))$.
  Thus, it suffices to show that
  $\Dproj_i \comp \Dproj_j \comp \D^2 f \comp \Dswap = \Dproj_j \comp
  \Dproj_i \comp \D^2 f$ when $(i, j) \neq (1, 1)$ to
  conclude that~(1) is equivalent to~(2).
  \begin{itemize}
  \item $i = 0, j = 0$: The equation holds by reflexivity of equality.
  \item $i=1, j=0$:
    $\Dproj_0 \comp \Dproj_1 \comp \D^2 f = \Dproj_1 \comp \D \Dproj_0 \comp
    \D^2 f = \Dproj_1 \comp \D f \comp \D \Dproj_0$ and
    $\Dproj_1 \comp \Dproj_0 \comp \D^2 f \comp \Dswap = \Dproj_1
    \comp \D f \comp \Dproj_0 \comp \Dswap = \Dproj_1 \comp \D f \comp
    \D \Dproj_0$ so both sides are equal.
  \item $i=0, j=1$: $\Dproj_0 \comp \Dproj_1 \comp \D^2 f \comp \Dswap
  = \Dproj_1 \comp \Dproj_0 \comp \D^2 f$ if and only if
  $\Dproj_0 \comp \Dproj_1 \comp \D^2 f 
  = \Dproj_1 \comp \Dproj_0 \comp \D^2 f \comp \Dswap$ because $\Dswap$ is 
  involutive. But this equation holds, as seen above.
  \end{itemize}

  Next, (2) is a particular case of~(3) for
  $x = \Dproj_0 \comp \Dproj_0$, $u = \Dproj_0 \comp \Dproj_1$,
  $v = \Dproj_1 \comp \Dproj_0$ and $w = \Dproj_1 \comp
  \Dproj_1$.
  Conversely, if (3) holds then 
  $\dcoh(\dcoh (f)) \comp \Dpair{\Dpair{x}{u}}{\Dpair{v}{w}} =
  \dcoh(\dcoh(f)) \comp \Dswap \comp
  \Dpair{\Dpair{x}{v}}{\Dpair{u}{w}} = \dcoh(\dcoh(f)) \comp
  \Dpair{\Dpair{x}{v}}{\Dpair{u}{w}}$. 
\end{proof}

\section{Compatibility with the cartesian product}

We assume in this section that $\category$ is cartesian and is
equipped with a left summability structure
$(\D, \Dproj_0, \Dproj_1, \Dsum)$.

\begin{notation} 
  \label{notation:product}
  We use $\with$ for the cartesian product, following the notations of
  $\LL{}$.
  For any objects $Y_0, Y_1$, the projection will be written as
  $\prodProj_i \in \category(Y_0 \with Y_1, Y_i)$ and the pairing of
  $f_0 \in \category(X, Y_0)$ and $f_1 \in \category(X, Y_1)$ as
  $\prodPair{f_0}{f_1}$.
  Finally, the terminal object will be written $\top$.
  Note that the uniqueness of the pairing in the universal property of
  the cartesian product can be understood as the joint monicity of the
  $\prodProj_i$.
\end{notation}

\subsection{Cartesian product and summability structure}
\label{sec:cartesian-summability}

\begin{definition} \label{def:prod-compatible} The summability
  structure $(\D, \Dproj_0, \Dproj_1, \Dsum)$ is \emph{compatible with
    the cartesian product} if $\prodPair{0}{0} = 0$ and, for all
  $f_0, g_0 \in \category(X, Y_0)$ and
  $f_1, g_1 \in \category(X, Y_1)$:
  \begin{itemize}
  \item $\prodPair{f_0}{f_1} \summable \prodPair{g_0}{g_1}$ if and
    only if $f_0 \summable g_0$ and $f_1 \summable g_1$
  \item and then
  \(
    \prodPair{f_0}{f_1} + \prodPair{g_0}{g_1} =
    \prodPair{f_0+f_1}{g_0+g_1}
  \).
  \end{itemize}

\end{definition}

That is, sums are computed componentwise.
Let us break down this definition in more details.

\begin{proposition} \label{prop:projections-additive}
The following are equivalent
\begin{itemize}
\item $\prodProj_0, \prodProj_1$ are additive;
\item $\prodPair{0}{0} = 0$ and for all
  $f_0, g_0 \in \category(X, Y_0)$ and
  $f_1, g_1 \in \category(X, Y_1)$, if
  $\prodPair{f_0}{f_1} \summable \prodPair{g_0}{g_1}$ then
  $f_0 \summable g_0$, $f_1 \summable g_1$ and
  $\prodPair{f_0}{f_1} + \prodPair{g_0}{g_1} =
  \prodPair{f_0+f_1}{g_0+g_1}$.
\end{itemize}
\end{proposition}

\begin{proof}
  Assume that $\prodProj_0, \prodProj_1$ are additive. Then
  $\prodProj_i \comp 0 = 0 = \prodProj_i \comp \prodPair{0}{0}$. Thus
  by joint monicity, $0 = \prodPair{0}{0}$.
  Furthermore, assume that
  $\prodPair{f_0}{f_1} \summable \prodPair{g_0}{g_1}$.
  Then by additivity of $\prodProj_i$,
  $\prodProj_i \comp \prodPair{f_0}{f_1} = f_i$ and
  $\prodProj_i \comp \prodPair{g_0}{g_1} = g_i$ are summable and
  $f_i + g_i = \prodProj_i \comp (\prodPair{f_0}{f_1} +
  \prodPair{g_0}{g_1})$.
  So the joint monicity of the $\prodProj_i$ implies that
  $\prodPair{f_0}{f_1} + \prodPair{g_0}{g_1} =
  \prodPair{f_0+f_1}{g_0+g_1}$.

  Conversely, since $\prodPair{0}{0} = 0$ we have
  $\prodProj_i \comp 0 = \prodProj_i \comp \prodPair{0}{0} = 0$.
  Let $f, g \in \category(X, Y_0 \with Y_1)$ be summable.
  One can write
  $f = \prodPair{\prodProj_0 \comp f}{\prodProj_1 \comp f}$ and
  $g = \prodPair{\prodProj_0 \comp g}{\prodProj_1 \comp g}$.
  Since $f \summable g$ we have
  $\prodProj_i \comp f \summable \prodProj_i \comp g$ and
  $f + g = \prodPair{\prodProj_0 \comp f + \prodProj_0 \comp
    g}{\prodProj_1 \comp f + \prodProj_1 \comp g}$.
  Applying $\prodProj_i$ on this equation yields that
  $\prodProj_i \comp (f+g) = \prodProj_i \comp f + \prodProj_i \comp
  g$ so $\prodProj_i$ is additive.
\end{proof}

\begin{cor} \label{prop:with-sum} %
  If $\prodProj_0$ and $\prodProj_1$ are additive, then $0 \with 0 = 0$ and
  for all $f_0, g_0 \in \category(X_0, Y_0)$ and
  $f_1, g_1 \in \category(X_1, Y_1)$, if
  $f_0 \with f_1 \summable g_0 \with g_1$ then $f_0 \summable g_0$,
  $f_1 \summable g_1$ and
  $f_0 \with f_1 + g_0 \with g_1 = (f_0 + g_0) \with (f_1 + g_1)$.
\end{cor}

\begin{proof}
  We simply use the fact that
  $f \with g = \prodPair{f \comp \prodProj_0}{g \comp \prodProj_1}$
  and~\Cref{prop:projections-additive} together with the left
  compatibility of sum with regard to composition
  (\Cref{prop:sum-left-compatible}).
\end{proof}
We now assume that the projections $\prodProj_0$ and $\prodProj_1$ are additive.
This allows us to define a morphism
$\prodSwap \in \category(\D (X_0 \with X_1), \D X_0 \with \D X_1)$ for
any objects $X_0, X_1$ as
$\prodSwap \defEq \prodPair{\Dpair{\prodProj_0 \comp
    \Dproj_0}{\prodProj_0 \comp \Dproj_1}} {\Dpair{\prodProj_1 \comp
    \Dproj_0}{\prodProj_1 \comp \Dproj_1}}$.
In other words,
$\Dproj_i \comp \prodProj_j \comp \prodSwap = \prodProj_j \comp
\Dproj_i$%
, that is
\[
  \prodSwap \comp \Dpair{\prodPair{f_0}{f_1} }{\prodPair{g_0}{g_1}} =
  \prodPair{\Dpair{f_0}{g_0}}{\Dpair{f_1}{g_1}}\,.
\]
This is very reminiscent of the flip $\Dswap$ (it swaps the two middle
coordinates), except that there are no summability conditions
associated with the $\prodPair{\_}{\_}$ pairing.

\begin{theorem} \label{prop:prodSwap-inverse} %
  The following assertions are equivalent
  \begin{enumerate}
  \item $\prodSwap$ is an isomorphism;
  \item $\Dproj_0 \with \Dproj_0 \summable \Dproj_1 \with \Dproj_1$;
  \item for any $f_0, g_0 \in \category(X, Y_0)$,
    $ f_1, g_1 \in \category(X, Y_1)$, if $f_0 \summable g_0$ and
    $f_1 \summable g_1$ then $f_0 \with f_1 \summable g_0 \with g_1$;
  \item for any $f_0, g_0 \in \category(X, Y_0)$,
    $ f_1, g_1 \in \category(X, Y_1)$, if $f_0 \summable g_0$ and
    $f_1 \summable g_1$ then
    $\prodPair{f_0}{f_1} \summable \prodPair{g_0}{g_1}$
  \end{enumerate}
  and then
  \(\Dpair{\Dproj_0 \with \Dproj_0}{\Dproj_1 \with
    \Dproj_1}=\prodSwap^{-1}\).
\end{theorem}

\begin{proof}
  $(1) \Rightarrow (2)$: Assume that $\prodSwap$ is an isomorphism with
  inverse $\witness$.
  Then
  $\Dproj_i \comp \prodProj_j = \Dproj_i \comp \prodProj_j \comp
  \prodSwap \comp \witness = \prodProj_j \comp \Dproj_i \comp
  \witness$.
  But
  $\Dproj_i \comp \prodProj_j = 
  \prodProj_j \comp (\Dproj_i \with \Dproj_i)$ by naturality of $\prodProj_j$ so
  $\prodProj_j \comp \Dproj_i \comp \witness = \prodProj_j \comp
  (\Dproj_i \with \Dproj_i)$.
  By joint monicity of the $\prodProj_j$'s we have
  $\Dproj_i \comp \witness = (\Dproj_i \with \Dproj_i)$.
  That is
  $\witness = \Dpair{\Dproj_0 \with \Dproj_0}{\Dproj_1 \with
    \Dproj_1}$.

  $(2) \Rightarrow (1)$: Assume that
  $\Dproj_0 \with \Dproj_0 \summable \Dproj_1 \with \Dproj_1$, of
  witness $\witness$. Then,
  $\prodProj_j \comp \Dproj_i \comp \witness = \prodProj_j \comp
  (\Dproj_i \with \Dproj_i) = \Dproj_i \comp \prodProj_j$.
  Hence
  \begin{align*}
    \prodProj_j \comp \Dproj_i \comp \witness \comp \prodSwap = 
    \Dproj_i \comp \prodProj_j \comp \prodSwap = \prodProj_j \comp \Dproj_i \\
    \Dproj_i \comp \prodProj_j \comp \prodSwap \comp \witness
    = \prodProj_j \comp \Dproj_i \comp \witness = \Dproj_i \comp \prodProj_j
  \end{align*}
  By joint monicity of the $\prodProj_j$'s and of the $\Dproj_i$'s we
  get $\witness \comp \prodSwap = \id_{\D(X_0 \with X_1)}$ and
  $\prodSwap \comp \witness = \id_{\D X_0 \with \D X_1}$.

  $(2) \Rightarrow (3)$: We have
  $\Dpair{f_0}{g_0}\in \category(X, \D Y_0)$ and
  $\Dpair{f_1}{g_1} \in \category(X, \D Y_1)$.
  Let
  \(
    \witness = \Dpair{\Dproj_0 \with \Dproj_0}{\Dproj_1 \with
      \Dproj_1} \comp (\Dpair{f_0}{g_0} \with \Dpair{f_1}{g_1})
  \).
  We have $\Dproj_0 \comp\witness = f_0 \with f_1$ and
  $\Dproj_1\comp \witness = g_0 \with g_1$ so that
  $f_0 \with f_1 \summable g_0 \with g_1$.
 
  $(3) \Rightarrow (2)$: $(2)$ is a particular case of case $(3)$.

  $(3) \Rightarrow (4)$: Assume that $f_0 \summable g_0$ and
  $f_1 \summable g_1$.
  Then by assumption, $f_0 \with f_1 \summable g_0 \with g_1$. Let
  $\witness = \Dpair{f_0 \with f_1}{g_0 \with g_1} \comp
  \prodPair{\id}{\id}$.
  Then $\Dproj_0\comp \witness = \prodPair{f_0}{f_1}$
  and $\Dproj_1\comp \witness = \prodPair{g_0}{g_1}$ so that
  $\prodPair{f_0}{f_1} \summable \prodPair{g_0}{g_1} $.

  $(4) \Rightarrow (3)$: Assume that $f_0 \summable g_0$ and
  $f_1 \summable g_1$.
  Then $f_0 \comp \prodProj_0 \summable g_0 \comp \prodProj_0$ and
  $f_1 \comp \prodProj_1 \summable g_1\comp \prodProj_1$ by left
  compatibility wrt.~composition (\Cref{prop:sum-left-compatible}).
  Hence, by assumption,
  $\prodPair{f_0\comp \prodProj_0}{ f_1\comp \prodProj_1} \summable
  \prodPair{g_0\comp \prodProj_0}{g_1\comp \prodProj_1}$.
  That is $f_0 \with f_1 \summable g_0 \with g_1$.
\end{proof}

\begin{cor}
  \label{cor:summability-cartesian-compat}
  A summability structure is compatible with the cartesian product if
  and only if $\prodProj_0, \prodProj_1$ are additive and $\prodSwap$
  is an isomorphism.
\end{cor}

\subsection{Cartesian product and differential structure}

We now assume that $\category$ is a cartesian category with a
pre-differential structure $(\D, \Dproj_0, \Dproj_1, \Dsum)$.

\begin{definition}
  The (pre-)differential structure $(\D, \Dproj_0, \Dproj_1, \Dsum)$
  is \emph{compatible with the cartesian product} if the underlying
  summability structure is compatible with the cartesian product, and
  if $\prodProj_0, \prodProj_1$ are $\D$-linear.
  A \emph{cartesian coherent differential category} (CCDC) is a
  coherent differential category whose cartesian product is
  compatible with the differential structure.
\end{definition}

We assume that $\category$ is a CCDC.
By $\D$-linearity of $\prodProj_0$ and $\prodProj_1$, all
constructions involving only the cartesian product are $\D$-linear.

\begin{proposition} \label{prop:prod-pairing-linear} %
  If $h_0 \in \category(X, Y_0)$ and $h_1 \in \category(X, Y_1)$ are
  $\D$-linear, then $\prodPair{h_0}{h_1}$ is $\D$-linear.
  If $f_0 \in \category(X_0, Y_0)$ and $f_1 \in \category(X_1, Y_1)$
  are $\D$-linear, then $f_0 \with f_1$ is $\D$-linear.
\end{proposition}

\begin{proof}
  For the first statement we proceed as for
  \Cref{prop:pairing-linear} except that the paring as a summable
  pair is replaced by the pairing of the cartesian product.
  The second statement follows from the first one, because
  $f_0 \with f_1 = \prodPair{f_0 \comp \prodProj_0}{f_1 \comp
    \prodProj_1}$, the projections are $\D$-linear, and
  \(\D\)-linearity is closed under composition.
\end{proof}

 
For any objects $X_0, X_1$, there is a natural transformation
$\prodPair{\D \prodProj_0}{\D \prodProj_1} \in \category(\D(X_0\with X_1),
\D X_0 \with \D X_1)$.
By \(\D\)-linearity of $\prodProj_0$ and $\prodProj_1$ this
natural transformation is equal to
$\prodPair{\Dpair{\prodProj_0 \comp \Dproj_0}{\prodProj_0 \comp
    \Dproj_1}} {\Dpair{\prodProj_1 \comp \Dproj_0}{\prodProj_1 \comp
    \Dproj_1}} = \prodSwap$.
Whence a result similar to~\Cref{prop:pair-derivative}.

\begin{proposition} \label{prop:prodpair-derivative} %
  For any $f_0 \in \category(X, Y_0)$ and $f_1 \in \category(X, Y_1)$,
  $\prodPair{\D f_0}{\D f_1} = \prodSwap \comp \D \prodPair{f_0}{f_1}$
\end{proposition}

\begin{proof}
  $\prodProj_i \comp \prodSwap \comp \D \prodPair{f_0}{f_1} = \D
  \prodProj_i \comp \D \prodPair{f_0}{f_1} = \D f_i$.
\end{proof}

\subsection{Partial derivatives}

Using $\prodSwap^{-1}$, we define two natural transformations
\begin{align*}
  \strengthL &= (\prodSwap)^{-1} \comp (\id_{\D X_0} \with \Dinj_0) 
               \in \category (\D X_0 \with X_1, \D (X_0 \with X_1))\\
  \strengthR &= (\prodSwap)^{-1} \comp (\Dinj_0 \with \id_{\D X_1})
               \in \category (X_0 \with \D X_1, \D (X_0 \with X_1))
\end{align*}
Note that $\prodSwap$, $(\prodSwap)^{-1}$, $\strengthL$ and
$\strengthR$ are all $\D$-linear, thanks to
\Cref{prop:prod-pairing-linear,prop:everyone-linear,prop:composition-linear}.

\begin{proposition} \label{prop:strength} %
  $\strength{0} = \Dpair{\Dproj_0 \with \id_{X_1}}{\Dproj_1 \with 0}$
  and
  $\strength{1} = \Dpair{\id_{X_0} \with \Dproj_0}{0 \with \Dproj_1}$
\end{proposition}

\begin{proof}
  By \Cref{prop:prodSwap-inverse},
  $(\prodSwap)^{-1} = \Dpair{\Dproj_0 \with \Dproj_0}{\Dproj_1 \with
    \Dproj_1}$
  and the result follows by a straightforward computation.
\end{proof}

\begin{definition}[Partial derivative] %
  If $f \in \category(X_0 \with X_1, Y)$ one can define
  $\D_0 f \defEq \D f \comp \strengthL \in \category(\D X_0 \with X_1,
  \D Y)$ and
  $\D_1 f \defEq \D f \comp \strengthR \in \category(X_0 \with \D X_1,
  \D Y)$, the \emph{partial derivatives} of $f$.
\end{definition}

\begin{proposition} \label{prop:partial-derivative-Dproj0} %
  For any $f \in \category(X_0\with X_1,Y)$,
  $\Dproj_0 \comp \D_0 f = f \comp (\Dproj_0 \with \id)$ and
  $\Dproj_0 \comp \D_1 f = f \comp (\id \with \Dproj_0)$.
\end{proposition}

\begin{proof}
  $\Dproj_0 \comp \D_0 f = \Dproj_0 \comp \D f \comp \strengthL = f
  \comp \Dproj_0 \comp \strengthL = f \comp (\Dproj_0 \with \id)$ by
  \Cref{prop:strength}.
  The proof for $\strengthR$ is similar.
\end{proof}

\begin{proposition} \label{prop:commutative-monad-more} %
  The following diagram commutes.
  \begin{center}
    \begin{tikzcd}
      \D (X_0 \with \D X_1)  \arrow[d, "\D \strengthR"'] & \D X_0 \with \D X_1 \arrow[l, "\strengthL"'] \arrow[r, "\strengthR"] & \D(\D X_0 \with X_1) \arrow[d, "\D \strengthL"] \\
      \D^2(X_0 \with X_1) \arrow[rr,swap, "\Dswap"'] & & \D^2(X_0 \with
      X_1)
    \end{tikzcd}
  \end{center}
\end{proposition}

\begin{proof} 
  We use \Cref{prop:strength} to compute
  $\D \strengthR \comp \strengthL$ and
  $\D \strengthL \comp \strengthR$.
  Since $\strengthL$ is $\D$-linear,
  $\D \strengthL = \Dpair{\strengthL \comp \Dproj_0}{\strengthL \comp
    \Dproj_1}$ by \Cref{rem:linear}.
  Thus
  \begin{align*}
    \D \strengthL \comp \strengthR &= 
                                     \D \strengthL \comp \Dpair{\id_{X_0} \with \Dproj_0}{0 \with \Dproj_1} \\
                                   &= \Dpair{\strengthL \comp (\id_{X_0} \with \Dproj_0)}{\strengthL \comp (0 \with \Dproj_1)} \\
                                   &= \Dpair{\Dpair{\Dproj_0 \with \Dproj_0}{\Dproj_1 \with 0}}
                                     {\Dpair{0 \with \Dproj_1}{0 \with 0}} 
  \end{align*}
  Similarly,
  $\D \strengthR \comp \strengthL = \Dpair{\Dpair{\Dproj_0 \with
      \Dproj_0}{0 \with \Dproj_1}}{\Dpair{\Dproj_1 \with 0}{0 \with
      0}}$.
  The commutation results from~\Cref{prop:family-on-pairs}.
\end{proof}

\begin{proposition} \label{prop:commutative-monad} %
  The following diagram commutes
  \begin{center}
  \begin{tikzcd}
    \D (X_0 \with \D X_1) \arrow[d, "\D \strengthR"']
    & \D X_0 \with \D X_1 \arrow[l, "\strengthL"']
    \arrow[r, "\strengthR"] \arrow[d, "\prodSwap^{-1}"]
    & \D(\D X_0 \with X_1) \arrow[d, "\D \strengthL"] \\
    \D^2(X_0 \with X_1) \arrow[r, "\DmonadSum"]
    & \D (X_0 \with X_1)
    & \D^2(X_0 \with X_1) \arrow[l, "\DmonadSum"']
  \end{tikzcd}
  \end{center}
\end{proposition}

\begin{proof}
  Thanks to the computation of $\D \strengthL \comp \strengthR$ in the
  proof of \Cref{prop:commutative-monad-more}, we know that
  $\DmonadSum \comp \D \strengthL \comp \strengthR = \Dpair{\Dproj_0
    \with \Dproj_0} {\Dproj_1 \with 0 + 0 \with \Dproj_1} =
  \Dpair{\Dproj_0 \with \Dproj_0}{\Dproj_1 \with \Dproj_1}$ by
  \Cref{prop:with-sum}.
  So
  $\DmonadSum \comp \D \strengthL \comp \strengthR = (\prodSwap)^{-1}$
  by \Cref{prop:prodSwap-inverse}.
  A similar computation yields the result for
  $\DmonadSum \comp \D \strengthR \comp \strengthL$.
\end{proof}


\begin{remark}
  We can check that the natural morphisms $\strengthL, \strengthR$ are 
  \emph{strenghts}~\cite{Kock72, Moggi91} 
  for the monad
  $(\D, \Dinj_0, \DmonadSum)$.
  Then the diagram of \Cref{prop:commutative-monad} means that this
  strong monad is a \emph{commutative monad}.
  The diagrams can be checked by hand, but are also a consequence of
  very generic properties about strong monads on cartesian categories.

  As mentioned in~\cite{Aguiar18} in paragraph 2.3, any monad
  $(\monad, \monadUnit, \monadSum)$ on a cartesian category can be
  endowed with the structure of \emph{a colax symmetric monoidal
    monad}%
  \footnote{Also called oplax symmetric monoidal monad, or symmetric
    comonoidal monad, or Hopf monad, see~\cite{Moerdijk02}} %
  taking
  \begin{itemize}
  \item
    $\smfProdOne$ is the unique element of $\category(\monad
    \top, \top)$
  \item
    $\smfProdTwo_{X_1, X_2} \defEq \prodPair{\monad
      \prodProj_1}{\monad \prodProj_2} \in \category(\monad (X_1 \with
    X_2), \monad X_1 \with \monad X_2)$
  \end{itemize}
  %
  If $\smfProdTwo$ and $\smfProdOne$ are isos, $\monad$ becomes a
  \emph{(strong) symmetric monoidal monad}.
  This is what happens here for
  $\monad = \D$, because $\smfProdTwo=\prodSwap$ and
  we can show that $\smfProdOne$ is an isomorphism with inverse
  $\Dinj_0$ using the join monicity of the $\Dproj_i$.
  %
  %
  But symmetric monoidal monad are the same as commutative monads as
  shown in \cite{Kock70, Kock72}, and it turns out that the strengths
  induced from the symmetric monoidal structure are exactly
  $\strengthL$ and $\strengthR$.
\end{remark}

The axioms \ref{ax:D-schwarz} and \ref{ax:D-add} carry to the
setting of partial derivatives very naturally thanks to
\Cref{prop:commutative-monad-more,prop:commutative-monad}
respectively, giving the full fledged Schwarz and Leibniz rules.
The fact that the Leibniz rule is a consequence of the additivity of
the derivative is not surprising, as it is also the case in the usual
differential calculus:
$f'(x, y) \cdot (u, v) = f'(x, y) \cdot (u, 0) + f'(x, y) \cdot (0, v)
= \frac{\partial f}{\partial x}(x, y) \cdot u + \frac{\partial
  f}{\partial y}(x, y) \cdot v$.

\begin{proposition}[Leibniz rule]
 \label{prop:leibniz} $\D f \comp \prodSwap^{-1} 
= \DmonadSum \comp \D_0 \D_1 f = \DmonadSum \comp \D_1 \D_0 f$
\end{proposition}

\begin{proof}
  Let us prove that
  $\D f \comp \prodSwap^{-1} = \DmonadSum \comp \D_0 \D_1 f$.
  \begin{align*}
    \DmonadSum \comp \D_0 \D_1 f 
    &= \DmonadSum \comp \D (\D f \comp \strengthR) \comp \strengthL  \quad \text{by definition} \\
    &=  \DmonadSum \comp \D^2 f \comp \D \strengthR \comp \strengthL \quad
                                                                       \text{by \ref{ax:D-chain}} \\
    &= \D f \comp \DmonadSum \comp \D \strengthR \comp \strengthL \quad
                                                                    \text{by \ref{ax:D-add}}  \\
    &= \D f \comp \prodSwap^{-1} \quad \text{by \Cref{prop:commutative-monad}} \tag*{\qedhere}
  \end{align*}
  The proof of
  $\D f \comp \prodSwap^{-1} = \DmonadSum \comp \D_1 \D_0 f$ is
  similar.
\end{proof}

\begin{proposition}[Schwarz rule]
$\D_0 \D_1 f = \Dswap \comp \D_1 \D_0 f$
\end{proposition}

\begin{proof}
  Very similar to that of \Cref{prop:leibniz}, except that it uses the
  naturality of $\Dswap$ of \ref{ax:D-schwarz} instead of the
  naturality of $\DmonadSum$.
\end{proof}

\subsection{Generalization to arbitrary finite products}

\begin{notation} Recall that the existence of arbitrary finite products is equivalent
to the existence of a binary product and a terminal object. In order to stay consistent
with the current notations, we write the finite products starting from $0$: 
$X_0 \with \cdots \with X_n$. We allow empty products, with the convention
that taking $n = -1$ yields a product $X_0 \with \cdots \with X_{-1} := \top$.
\end{notation}

The constructions above can be extended to arbitrary finite products.
On can indeed define
a (symmetric monoidal) natural transformation
$\prodSwap[n] \in \category(\D (X_0 \with \cdots \with X_n), 
\D X_0 \with \cdots \with \D X_{n})$
inductively by 
\((\prodSwap[-1])_X \defEq \prodFinal_{\D \top} \in \category(\D \top, \top)\),
\((\prodSwap[0])_X \defEq \id_{\D X} \in \category(\D X, \D X)\) and
\(\prodSwap[n+1] \defEq \prodSwap \comp
\prodPair{\prodSwap[n]}{\id_{\D X_{n+1}}}\).
By associativity of the cartesian product, this definition does not
depend on the actual parenthesizing of $X_0 \with \cdots \with X_n$.

\begin{notation}
  Let $X_0, Y_0, \ldots, X_n, Y_n\in\Obj\category$,
  $i \in \interval{0}{n}$ and
  $f_k \in \category(X_k, Y_k)$ for each $k\neq i$.
  Let $g \in \category(X_i, Y_i)$.
  Define $\singleApp{i}{f}{g} \defEq
  f_0 \with \cdots \with f_{i-1} \with g \with f_{i+1} \with \cdots
  \with f_n$ 
  in which we use $f_i$ everywhere except at position $i$ where we use $g$.
\end{notation}

Similarly to the binary case, one can then define a strength
$\strength{i} \in \category(X_0 \with \cdots \with \D X_i \with \cdots
\with X_n, \D (X_0 \with \cdots \with X_n))$ as
\[
  \strength{i}
  \defEq (\prodSwap[n])^{-1} \comp \singleApp{i}{(\Dinj_0)}{\id_{\D X_i}}\,.
\]
\begin{proposition}
  $\prodSwap[n]$ is an isomorphism and
  $(\prodSwap[n])^{-1} = \Dpair{\Dproj_0 \with \cdots \with
    \Dproj_0}{\Dproj_1 \with \cdots \with \Dproj_1}$.
  Hence,
  \label{prop:strength-n}
  $\strength{i} =
  \Dpair{\singleApp{i}{\id}{\Dproj_0}}{\singleApp{i}{0}{\Dproj_1}}$.
\end{proposition}

\begin{proof} 
  The equation on $\prodSwap[n]$ is obtained by unfolding the
  inductive definition and using~\Cref{prop:prodSwap-inverse}.
  The equations on the $\strength{i}$'s follow from this, as
  in~\Cref{prop:strength}.
\end{proof}

\begin{definition}
  For any $f \in \category(X_0 \with \cdots \with X_n, Y)$ one can
  define the \emph{i-th partial derivative of $f$} as
  $\D_i f \defEq \D f \comp \strength{i} \in \category(X_0 \with
  \cdots \with \D X_i \with \cdots \with X_n, \D Y)$.
\end{definition}

\begin{proposition} \label{prop:partial-derivative-Dproj0-n}
  $\Dproj_0 \comp \D_i f = f \comp \singleApp{i}{\id}{\Dproj_0}$.
\end{proposition}

\begin{proof}
  Same as \Cref{prop:partial-derivative-Dproj0}.
\end{proof}

\begin{definition}
  \label{def:natural-trans-iterate}
  For any $X\in\Obj\category$ and $n \geq 0$, we can define
  $\DmonadSum_X^k \in \category(\D^{n+1} X, \D X)$
  as the composition of \(k\) copies of \(\DmonadSum\):
  $\DmonadSum_X^0 = \id_{\D X}$ and
  $\DmonadSum_X^{k+1} = \DmonadSum_X^k \comp \DmonadSum_{\D^k X}$.
  We define similarly $\Dproj_i^k \in \category(\D^k X, X)$.
\end{definition}

Note that
$\DmonadSum^{k} = \Dpair{\Dproj_0^{k+1}}{\sum_{j = 0}^{k} \Dproj_0^{j}
  \comp \Dproj_1 \comp \Dproj_0^{k-j}}$.
In other words, the right component of $\DmonadSum^k$ sums over all of
the possible combinations of $k$ left projections and one right
projection.
One can prove a generalization of \Cref{prop:commutative-monad} for
$n \geq 1$,
\[
  (\prodSwap[n])^{-1}
  = \DmonadSum^{n} \comp \D^n \strength{\alpha(n)} \comp \cdots
  \comp \D \strength{\alpha(1)} \comp \D\strength{\alpha(0)}
\]
for any $\alpha$ permutation of $\llbracket 0, n \rrbracket$.
As in \Cref{prop:leibniz}, this generalizes the Leibniz Rule to the
$n$-ary case.

\begin{proposition}[Leibniz, generalized]
 \label{prop:leibniz-n} For any $n \geq 1$ and for any permutation $\alpha$ of 
 $\llbracket 0, n \rrbracket$,
 \[
   \D f \comp (\prodSwap[n])^{-1} = \DmonadSum^{n} \comp \D_{\alpha(n)} \ldots
   \D_{\alpha(0)} f\,.
 \]
\end{proposition}

\subsection{Multilinear morphism}

We generalize to multivariate functions the notion of additivity and 
$\D$-linearity.

\begin{definition}
  A morphism $f \in \category(Y_0 \with \cdots \with Y_n, Z)$ 
  is additive in its $i^{th}$ argument (for $i \in \interval{0}{n}$) if 
  $f \comp \singleApp{i}{\id}{0} = 0$ and
  if for all $h_0, h_1 \in \category(X, Y_i)$ such that $h_0 \summable h_1$, then
  $f \comp \singleApp{i}{\id}{h_0} \summable f \comp \singleApp{i}{\id}{h_1}$
  and
  \[ f \comp \singleApp{i}{\id}{h_0} + f \comp \singleApp{i}{\id}{h_1} = 
    f \comp \singleApp{i}{\id}{h_0 + h_1} \]
\end{definition}

\begin{proposition} A morphism $f \in \category(Y_0 \with \cdots \with Y_n, Z)$ 
such that $f \comp \singleApp{i}{\id}{0} = 0$
is additive in its $i^{th}$ argument if and only if 
$f \comp \singleApp{i}{\id}{\Dproj_0} \summable f \comp \singleApp{i}{\id}{\Dproj_1}$
 with sum $f \comp \singleApp{i}{\id}{\Dsum}$.
\end{proposition}

\begin{proof} 
  The proof is the same as \Cref{prop:additive}, using the fact that
  for any $k \in \{0,1\}$, 
  $f \comp \singleApp{i}{\id}{h_k} 
  = f \comp \singleApp{i}{\id}{\Dproj_k} \comp 
  \singleApp{i}{\id}{\Dpair{h_0}{h_1}}$.
\end{proof}

\begin{definition} A morphism
  $f \in \category(X_0 \with \cdots \with X_n, Y)$ is
 \emph{linear in its $i^{th}$ argument} if  it is
 additive in this argument and if
 $\Dproj_1 \comp \D_i f = f \comp \singleApp{i}{\id}{\Dproj_1}$.
\end{definition}

As in \Cref{prop:linear-equation}, \ref{ax:D-add} ensures
that the equation 
$\Dproj_1 \comp \D_i f = f \comp \singleApp{i}{\id}{\Dproj_1}$
is a sufficient condition for linearity in the $i^{th}$ argument.

\begin{proposition} \label{prop:additive-i}
Assume that $\Dproj_1 \comp \D_i f = f \comp \singleApp{i}{\id}{\Dproj_1}$.
Then $f$ is additive in its $i^{th}$ argument, hence linear in that argument.
\end{proposition}

\begin{proof} The equation allows rewriting $f \comp \singleApp{i}{\id}{h}$ 
  as follows.
  \begin{align*}
    f \comp \singleApp{i}{\id}{h} 
    &= f \comp \singleApp{i}{\id}{\Dproj_1} \comp \singleApp{i}{\id}{\Dpair{0}{h}} \\
    &= \D_i f \comp \singleApp{i}{\id}{\Dpair{0}{h}}
    \quad \text{by assumption} \\
    &= \D f \comp \strength{i} \comp \singleApp{i}{\id}{\Dpair{0}{h}} \\
    &= \D f \comp \Dpair{\singleApp{i}{\id}{0}}{\singleApp{i}{0}{h}}
      \quad \text{by \Cref{prop:strength-n}}
  \end{align*}

  In particular, $f \comp \singleApp{i}{\id}{0} =
  \D f \comp \Dpair{\singleApp{i}{\id}{0}}{\singleApp{i}{0}{0}}$.
  But $\singleApp{i}{0}{0} = 0$ by \Cref{prop:with-sum}. 
  So by \ref{ax:D-add} and \Cref{prop:derivative-additive-zero},
  $f \comp \singleApp{i}{\id}{0} = 0$. Similarly, if $h_0 \summable h_1$,
  \begin{align*}
    &f \comp \singleApp{i}{\id}{h_0 + h_1} \\
    &=\D f \comp \Dpair{\singleApp{i}{\id}{0}}{\singleApp{i}{0}{h_0 + h_1}} \\
    &= \D f \comp \Dpair{\singleApp{i}{\id}{0}}{\singleApp{i}{0}{h_0} +
    \singleApp{i}{0}{h_1}} 
    \quad \text{by \Cref{prop:with-sum}}\\
    &= \D f \comp \Dpair{\singleApp{i}{\id}{0}}{\singleApp{i}{0}{h_0}}
      + \D f \comp \Dpair{\singleApp{i}{\id}{0}}{\singleApp{i}{0}{h_1}} 
      \quad \text{by \ref{ax:D-add} and \Cref{prop:derivative-additive-sum}} \\
    &= f \comp \singleApp{i}{\id}{h_0}  + f \comp \singleApp{i}{\id}{h_1} \, .
  \end{align*}
\end{proof}

\begin{definition} \label{def:multilinear} A morphism
  $f \in \category(X_0 \with \cdots \with X_n, Y)$
  is \emph{multilinear} (and more precisely, \emph{$(n+1)$-linear}) if
  it is linear in all of its argument.
  Note that the \(1\)-linear morphisms are exactly the \(\D\)-linear
  ones.

\end{definition}


As a sanity check of the notion, we can use the result below together
with the Leibniz rule to show a result similar to the fact that in
differential calculus, if $\Phi$ is a bilinear map, then
$\Phi'(x, y) \cdot (u, v) = \Phi(x, v) + \Phi(u, y)$.

\begin{lemma} \label{lemma:Dproj0-partial-commute}
For any $f \in \category(X_0 \with \cdots \with X_n, Y)$ and 
$i, j \in \interval{0}{n}$ such that $i \neq j$,
\[ \D \Dproj_0 \comp \D_i \D_j f = \D_i f \comp \singleApp{j}{\id}{\Dproj_0}  \]
\end{lemma}

\begin{proof} This is a direct computation
\begin{align*}
\D \Dproj_0 \comp \D_i \D_j f &= \D \Dproj_0 \comp \D (\D_j f) \comp \strength{i} \\
&= \D (\Dproj_0 \comp \D_j f) \comp \strength{i} \quad \text{by \ref{ax:D-chain}} \\
&= \D (f \comp \singleApp{j}{\id}{\Dproj_0}) \comp \strength{i} \quad
\text{by \Cref{prop:partial-derivative-Dproj0-n}} \\
&= \D f \comp \D \singleApp{j}{\id}{\Dproj_0} \comp \strength{i} \quad \text{by \ref{ax:D-chain}} \\
&= \D f \comp \strength{i} \comp \singleApp{j}{\id}{\Dproj_0} \quad \text{since $\strength{i}$ natural
and $i \neq j$} \\
&= \D_i f \comp \singleApp{j}{\id}{\Dproj_0}
\end{align*}
\end{proof}

\begin{theorem}
  For any $(n+1)$-linear morphism
  $f \in \category(X_0 \with \cdots \with X_n, Y)$
\begin{align*}
  \Dproj_0 \comp \D f \comp (\prodSwap)^{-1}
  &= f \comp (\Dproj_0 \with \cdots \with \Dproj_0) \\
  \Dproj_1 \comp \D f \comp (\prodSwap)^{-1}
  &=f \comp (\Dproj_1 \with \Dproj_0 \with \cdots \with \Dproj_0) +
                                                    \cdots 
   + \: f \comp (\Dproj_0 \with \cdots \with \Dproj_0 \with \Dproj_1)
\end{align*}
\end{theorem}

\begin{proof}
  We will write the proof for $n=1$. The general case relies on the
  same arguments.
  The first equation is just a direct consequence of the naturality of
  $\Dproj_0$ and \Cref{prop:strength}.
  For the second equation, Leibniz (\Cref{prop:leibniz}) ensures that
  $\Dproj_1 \comp \D f \comp \prodSwap^{-1} = \Dproj_1 \comp
  \DmonadSum \comp \D_0 \D_1 f = \Dproj_1 \comp \Dproj_0 \comp \D_0
  \D_1 f + \Dproj_0 \comp \Dproj_1 \comp \D_0 \D_1 f$.
  We can compute those two summands separately.
\begin{align*}
 \Dproj_1 \comp{} \Dproj_0 \comp \D_0 \D_1 f  
&= \Dproj_1 \comp \D_1 f \comp (\Dproj_0 \with \id) 
\quad \text{by \Cref{prop:partial-derivative-Dproj0}}\\
&= f \comp (\id \with \Dproj_1) \comp (\Dproj_0 \with \Dproj_1) 
\quad  \text{by bilinarity of $f$}  \\
&= f \comp (\Dproj_0 \with \Dproj_1)
\end{align*}
\begin{align*}
\Dproj_0 \comp{} \Dproj_1 \comp \D_0 \D_1 f
&= 
\Dproj_1 \comp \D \Dproj_0 \comp \D_0 \D_1 f 
\quad \text{by linearity of $\Dproj_0$} \\
&= \Dproj_1 \comp \D_0 f \comp (\id \with \Dproj_0) 
\quad \text{by \Cref{lemma:Dproj0-partial-commute}}  \\
&= f \comp (\Dproj_1 \with \id) \comp (\id \with \Dproj_0) 
\quad \text{by bilinarity of $f$} \\
&= f \comp (\Dproj_1 \with \Dproj_0)
\end{align*}
Which concludes the proof.
\end{proof}

We can expand on the ideas of the proof
\Cref{lemma:Dproj0-partial-commute} to show the following result.
This result is crucial, as it explains how to project on a series of
partial derivatives.
%
%
\begin{proposition} \label{prop:DDproj-commute-partial}
Let $n \geq 0$, $f \in \category(X_0 \with \cdots \with X_n)$,
$d \geq 0$ and $i, i_1, \ldots, i_d \in \interval{0}{n}$. Then,
\[ \D^d \Dproj_0 \comp \D_{i_d} \ldots \D_{i_1} \D_i f =
\D_{i_d} \ldots \D_{i_1} f \comp \singleApp{i}{\id}{\D^{h_d(i)} \Dproj_0} \]
where $h_d(i) = \# \{k \in \interval{1}{d} \mid i_k = i\}$.
Furthermore, if $f$ is $(n+1)$-linear, then 
\[ \D^d \Dproj_1 \comp \D_{i_d} \ldots \D_{i_1} \D_i f =
\D_{i_d} \ldots \D_{i_1} f \comp \singleApp{i}{\id}{\D^{h_d(i)} \Dproj_1} \]
\end{proposition}

\begin{proof}
  By induction on $d$.
  The case $d = 0$ is \Cref{prop:partial-derivative-Dproj0-n} for
  $\Dproj_0$, and the definition of $n$-linearity for $\Dproj_1$.
  We deal with the inductive step for $\Dproj_0$.
  The inductive step for $\Dproj_1$ is dealt with similarly.
\begin{align*}
   \D^{d+1} \Dproj_0 \comp \D_{i_{d+1}} \ldots \D_{i_1} \D_i f
  &= \D(\D^d \Dproj_0) \comp \D(\D_{i_d} \ldots \D_{i_1} \D_i f) \comp \strength{i_{d+1}}
    \quad \text{by definition}  \\
  &= \D (\D^d \Dproj_0 \comp \D_{i_d} \ldots \D_{i_1} \D_i f) \comp \strength{i_{d+1}}
    \quad \text{by \ref{ax:D-chain}} \\
  &= \D (\D_{i_d} \ldots \D_{i_1} \D_i f \comp \singleApp{i}{\id}{\D^{h_d(i)} \Dproj_0}) \comp
    \strength{i_{d+1}}
    \quad \text{by inductive hypothesis} \\
  &=  \D \D_{i_d} \ldots \D_{i_1} \D_i f \comp \D \singleApp{i}{\id}{\D^{h_d(i)} \Dproj_0} \comp
    \strength{i_{d+1}}
    \quad \text{by \ref{ax:D-chain}} 
\end{align*}
The next step is to use the naturality of $\strength{i_{d+1}}$:
\[ \D (f_0 \with \cdots \with f_n) \comp \strength{i_{d+1}} =
\singleApp{i_{d+1}}{f}{\D f_{i_{d+1}}} \]
If $i_{d+1} = i$, then 
\[ \D \singleApp{i}{\id}{\D^{h_d(i)} \Dproj_0} \comp \strength{i_{d+1}}
= \strength{i_{d+1}} \comp \singleApp{i}{\id}{\D^{h_d(i) + 1} \Dproj_0} \]
If $i_{d+1} \neq i$ then
\[ \D \singleApp{i}{\id}{\D^{h_d(i)} \Dproj_0} \comp \strength{i_{d+1}}
= \strength{i_{d+1}} \comp \singleApp{i}{\id}{\D^{h_d(i)} \Dproj_0} \]
In both case, 
\[ \D \singleApp{i}{\id}{\D^{h_d(i)} \Dproj_0} \comp \strength{i_{d+1}}
= \strength{i_{d+1}} \comp \singleApp{i}{\id}{\D^{h_{d+1}(i)} \Dproj_0}\]
 Consequently:
\begin{align*}
\D^{d+1} \Dproj_0 \comp \D_{i_{d+1}} \ldots \D_{i_1} \D_i f  
&= \D \D_{i_d} \ldots \D_{i_1} \D_i f 
\comp \strength{i_{d+1}} \comp \singleApp{i}{\id}{\D^{h_{d+1}(i)} \Dproj_0} \ \\
&=  \D_{i_{d+1}} \D_{i_d} \ldots \D_{i_1} \D_i f \comp \singleApp{i}{\id}{\D^{h_{d+1}(i)} \Dproj_0}
\end{align*}
which concludes the proof.
\end{proof}

This property instantiated in $d = 1$ gives back something similar to
\Cref{lemma:Dproj0-partial-commute}.

\begin{cor} \label{cor:Dproj-commute-partial}
If $f \in \category(X_0 \with \cdots \with X_n)$ is $(n+1)$-linear, then
for any $i, j \in \interval{0}{n}$ such that $i \neq j$ and for any $k \in \{0,1\}$,
\[ \D \Dproj_k \comp \D_i \D_j f = \D_i f \comp \singleApp{j}{\id}{\Dproj_k} \]
\[ \D \Dproj_k \comp \D_i \D_i f = \D_i f \comp \singleApp{i}{\id}{\D \Dproj_k} \]
\end{cor}

We can use this corollary to show that the partial derivative of a $(n+1)$-linear
morphism is also $(n+1)$-linear.

\begin{theorem}  \label{prop:partial-preserve-multilinearity}
If $f \in \category(X_0 \with \cdots \with X_n)$ is $(n+1)$-linear, 
then for any $i \in \interval{0}{n}$, $\D_i f$ is $(n+1)$-linear.
\end{theorem}

\begin{proof} Let $j \in \interval{0}{n}$. The goal is to prove that 
$\Dproj_1 \comp \D_j \D_i f = \D_i f \comp \singleApp{j}{\id}{\Dproj_1}$. By joint monicity
of the $\Dproj_k$, it suffices to prove that 
$\Dproj_k \comp \Dproj_1 \comp \D_j \D_i f = \Dproj_k \comp \D_i f \comp 
\singleApp{j}{\id}{\Dproj_1}$ for any $k \in \{0,1\}$. 
If $i \neq j$, 
\begin{align*}
\Dproj_k \comp \Dproj_1 \comp \D_j \D_i f 
&= \Dproj_1 \comp \D \Dproj_k \comp \D_j \D_i f  \text{\quad by $\D$-linearity of $\Dproj_1$} \\
&= \Dproj_1 \comp \D_j f \comp \singleApp{i}{\id}{\Dproj_k}  \text{\quad by 
\Cref{cor:Dproj-commute-partial}} \\
&= f \comp \singleApp{j}{\id}{\Dproj_1} \comp \singleApp{i}{\id}{\Dproj_k}  \text{\quad since $f$ is 
$(n+1)$-linear}  \\
&= f \comp \singleApp{i}{\id}{\Dproj_k} \comp \singleApp{j}{\id}{\Dproj_1} \text{\quad since }
i \neq j  \\
&= \Dproj_k \comp \D_i f \comp \singleApp{j}{\id}{\Dproj_1}  \text{\quad since $f$ is $(n+1)$-linear}
\end{align*}
The case $i = j$ is very similar
\begin{align*}
\Dproj_k \comp \Dproj_1 \comp \D_i \D_i f 
&= \Dproj_1 \comp \D \Dproj_k \comp \D_i \D_i f \text{\quad by $\D$-linearity of $\Dproj_1$} \\
&= \Dproj_1 \comp \D_i f \comp \singleApp{i}{\id}{\D \Dproj_k} \text{\quad by 
\Cref{cor:Dproj-commute-partial}} \\
&= f \comp \singleApp{i}{\id}{\Dproj_1} \comp \singleApp{i}{\id}{\D \Dproj_k}  \text{\quad since $f$ is 
$(n+1)$-linear}  \\
&= f \comp \singleApp{i}{\id}{\Dproj_k} \comp \singleApp{i}{\id}{\Dproj_1} 
\text{\quad since $\Dproj_k$ is $\D$-linear}  \\
&= \Dproj_k \comp \D_i f \comp \singleApp{i}{\id}{\Dproj_1}  \text{\quad since $f$ is $(n+1)$-linear.}
\end{align*}
\end{proof}

Composition with a linear morphism preserves multilinearity.
Thus, the Leibniz rule ensures that if $f$ is multilinear then $\D f$ is also multilinear.

\begin{proposition}
\label{prop:composition-preserve-multilinearity}
 If $f \in \category(X_0 \with \cdots \with X_n, Y)$ is
$(n+1)$-linear and $h \in \category(Y, Z)$ is linear, then
$h \comp f$ is $(n+1)$-linear.
\end{proposition}

\begin{proof} This follows from a straightforward computation
$\Dproj_1 \comp \D_i (h \comp f) = \Dproj_1 \comp \D (h \comp f) \comp \strength{i}
= \Dproj_1 \comp \D h \comp \D f \comp \strength{i} = h \comp \Dproj_1 \comp \D_i f 
= h \comp f \comp \singleApp{i}{\id}{\Dproj_1}$.
\end{proof}

%
%
\begin{theorem} 
  If $f \in \category(X_0 \with \cdots \with X_n, Y)$ is
  $(n+1)$-linear, then %
  $\D f \comp (\prodSwap[n])^{-1}\in\category(\D X_0 \with \cdots
  \with \D X_n, \D Y)$ is also $(n+1)$-linear.
\end{theorem} 

\begin{proof} By Leibniz (\Cref{prop:leibniz-n}),
$\D f \comp (\prodSwap[n])^{-1} = \DmonadSum^{n} \comp \D_{\alpha(n)} \ldots
 \D_{\alpha(0)} f$. But the partial derivatives preserves multilinearity by 
 \Cref{prop:partial-preserve-multilinearity} and composition by $\DmonadSum^{n}$ 
 on the left preserves multilinearity by \Cref{prop:composition-preserve-multilinearity}. 
\end{proof}

\section{Kleisli category of the exponential of a model of $\LL{}$}

\subsection{Coherent differentiation in a linear setting}

\label{sec:CD-induces-CCDC}

Let $\categoryLL$ be a symmetric monoidal closed category that is a model
of $\LL{}$, and more precisely a
Seely category in the sense of~\cite{Mellies09}. 
We write the composition of $f \in \categoryLL(X, Y)$ with
$g \in \categoryLL(Y, Z)$ as $g \compl f$ to stress the intuition
that the morphisms of \(\cL\) are linear.
The axioms of a Seely category include the existence of a cartesian
product $\with$ and a comonad $(\oc, \der, \dig)$ on $\categoryLL$,
where $\der_X \in \categoryLL(\Excl X, X)$ and
$\dig_X \in \categoryLL(\Excl X, \Excll X)$ are natural transformations.
The \emph{Kleisli category} \(\Kl\cL\) of this comonad is the category
whose objects are the objects of $\categoryLL$ and whose
hom-sets are $\kleisliExp(X, Y) = \categoryLL(\Excl X, Y)$.
Composition is defined in this category as
$g \comp f = g \compl \Excl f \compl \dig$ and the identity at
$X$ is $\der_X$, the unit of the comonad.
It is well known that $\kleisliExp$ is a cartesian (closed)
category, with the same cartesian product $\with$ as $\categoryLL$.

The goal of this section is to show that coherent differentiation 
on $\categoryLL$ as
introduced in~\cite{Ehrhard21} in the setting of $\LL{}$ gives
$\kleisliExp$ a CCDC structure. 

  \begin{theorem}
    \label{prop:cd-induces-cdc}
    Any differential structure on a summable category $\categoryLL$
    (see~\cite{Ehrhard21}) induces a CCDC structure
    on $\kleisliExp$.
\end{theorem}

Let us detail first what the assumption means.
The category $\categoryLL$ is said to be
\emph{summable}~\cite{Ehrhard21} if it has a summability structure
$(\S, \Sproj_0, \Sproj_1, \Ssum)$ in the sense of \Ehrhard.
By \Cref{prop:summability-struct-equivalence}, this means that
$(\S, \Sproj_0, \Sproj_1, \Ssum)$ is a left summability structure in
the sense of \Cref{def:left-summability-struct} where every morphism
is additive and the functorial action of $\S$ is given by
$\S f \defEq \Spair{f \comp \Sproj_0}{f \comp \Sproj_1}$.
Then, we can define $\Sinj_i$, $\SmonadSum$, $\Slift$ and $\Sswap$ as
usual\footnote{Note that in~\cite{Ehrhard21}, $\SmonadSum$ is called
  $\tau$}.
The difference is that the additivity of every morphism ensures that
those families are natural transformations for the functor $\S$.
In particular, $(\S, \Sinj_0, \SmonadSum)$ is \emph{de facto} a monad.
The category $\categoryLL$ is said to be \emph{summable as a cartesian
  category} if $\prodPair{\S \prodProj_0}{\S \prodProj_1} = \prodSwap$
is an isomorphism%
\footnote{We can show that the condition required in~\cite{Ehrhard21}
  that $0 \in \categoryLL(\S \top, \top)$ is an isomorphism always
  hold, using the joint monicity of the $\Dproj_i$}. %
Because every morphism of $\categoryLL$ is additive, this corresponds
by \Cref{cor:summability-cartesian-compat} to the fact that the
cartesian product is compatible with the left summability structure
as in~\Cref{def:prod-compatible}.

It is well known that there is a faithful functor
$\kleisliCastExp : \categoryLL \arrow \kleisliExp$ which maps $X$ to
$X$ and $f \in \categoryLL(X, Y)$ to
$f \Compl \der_X \in \categoryLL(\Excl X, Y)$.
We can show that this functor induces a left summability structure
$(\D, \lin(\Sproj_0), \lin(\Sproj_1), \lin(\Ssum))$ on $\kleisliExp$
(where $\D X \defEq \S X$) compatible with the cartesian product
$\with$ of $\kleisliExp$.
The reason is that $\lin$ preserves monicity and additivity, thanks to
the well known fact that $\lin(h) \comp f = h \Compl f$.
Finally, the definition of $\lin$ ensures that
$\lin(\Spair{f_0}{f_1}) = \Dpair{\lin(f_0)}{\lin(f_1}$.
In particular, the families of morphism generated by the Left
Summability Structure
$(\D, \lin(\Sproj_0), \lin(\Sproj_1), \lin(\Ssum))$ in
\Cref{def:injections,def:DmonadSum,def:Dlift,def:Dswap} are
$\lin(\Sinj_i)$, $\lin(\SmonadSum)$, $\lin(\Slift)$ and $\lin(\Sswap)$
respectively.

Then a \emph{differential structure} on
 a summable category
$\categoryLL$ is a natural transformation
$\devM_X \in \categoryLL(\Excl \S X, \S \Excl X)$ satisfying some
equations called ($\devM$-chain), ($\devM$-local), ($\devM$-lin),
($\devM$-$\with$) and ($\devM$-Schwarz) (see~\cite{Ehrhard21}).
The first axiom, ($\devM$-chain), is a compatibility condition of
$\devM$ with regard to $\dig$ and $\der$, making $\devM$ a
\emph{distributive law} between the functor $\S$ and the comonad
$\Excl\_$.

\begin{definition}
  A \emph{distributive law} between a functor
  $F : \categoryLL \arrow \categoryLL$ and the comonad $\Excl\_$ on
  $\categoryLL$ is a natural transformation
  $\lambda^F \in \categoryLL(\Excl FX, F\Excl X)$ such that the two
  following diagrams commute.
  \begin{center}
    \begin{tikzcd}
      \Excl FX \arrow[r, "\lambda^F_X"] \arrow[rd, "\der_{FX}"']
      & F\Excl X \arrow[d, "F \der_X"] \\
      & FX
    \end{tikzcd}
    \begin{tikzcd}
      \Excl FX \arrow[d, "\dig_{FX}"'] \arrow[rr, "\lambda_X^F"]
      &
      & F\Excl X \arrow[d, "F\dig_X"] \\
      \Excll FX \arrow[r, "\Excl\lambda^F_X"']
      & \Excl F\Excl X \arrow[r, "\lambda^F_{\Excl X}"']
      & F\Excll X
    \end{tikzcd}
\end{center}
\end{definition}

A definition of distributive laws can be found in~\cite{Power02},
together with a proof of \Cref{prop:dl-extension,prop:dl-morphism} 
stated below
(corollary 5.11 of~\cite{Power02})%
\footnote{These observations are made in the more general setting of
  $2$-categories}.

\begin{proposition} \label{prop:dl-extension} %
  Let $F : \category \arrow \category$ be an endofunctor.
  There is a bijection between distributive laws
  $\lambda_F \in \categoryLL(\Excl F X, F \Excl X)$ and
  \emph{liftings}\footnote{The word ``extension'' is also used. 
  We use the term lifting in order to stick to the terminology 
  of~\cite{Power02}} $\lift{F}$ of $F$ on $\kleisliExp$.
  A lifting $\lift{F}$ of $F$ is a functor
  $\lift{F} : \kleisliExp \arrow \kleisliExp$ such that
  $\lift{F} X = F X$ and
  $\lift{F} (\kleisliCastExp(h)) = \kleisliCastExp(F h)$.
\end{proposition}

\begin{proof}
  Given a distributive law
  $\lambda^F \in \categoryLL(\Excl FX, F\Excl X)$, one can define an
  extension mapping $X$ to $F X$ and $f \in \kleisliExp(X, Y)$ to
  $F (f) \Compl \lambda^F_X \in \kleisliExp(FX, FY)$. We can check that
  it is a functor using the diagrams of distributive laws, and a
  lifting
   of $F$ using the naturality of $\lambda^F$.
  Conversely, any lifting $\lift{F}$ of $F$ induces a family
  $\lambda^F_{X} = \lift{F} \id_{\Excl X} \in \kleisliExp(\Excl FX,
  F\Excl X)$.
  The two diagrams of distributive law comes from the functoriality of
  $\lift{F}$ and the naturality comes from the fact that $\lift{F}$ is
  an extension of $F$.
\end{proof}

\begin{remark} \label{rem:dl-composite} %
  Let $F, G : \categoryLL \arrow \categoryLL$ be two functors, with
  respective lifting $\lift{F}$ and $\lift{G}$ associated to the
  distributive laws $\lambda^F \in \categoryLL(\Excl FX, F\Excl X)$
  and $\lambda^G \in \categoryLL(\Excl GX, G\Excl X)$.
  Then $\lift{G} \lift{F}$ is a lifting of $G F$ and the distributive
  law associated with $\lift{G} \lift{F}$ is the following natural
  transformation:
  $\lambda^{GF}_X = G (\lambda^F_X) \Compl \lambda^G_{F X} \in
  \categoryLL(\Excl GFX, GF \Excl X)$.
\end{remark}

The result below, proved in~\cite{Power02}, is rather overlooked. While
the proof is indeed quite simple, it provides a very interesting perspective 
on the idea of extending structure to a Kleisli (or similarly to an 
Eilenberg-Moore) category.

\begin{proposition} \label{prop:dl-morphism} %
  Let $F, G : \categoryLL \arrow \categoryLL$ be two endofunctors.
  Assume that $\lift{F}$ and $\lift{G}$ are lifting of $F$ and $G$
  respectively, and let $\lambda^F$ and $\lambda^G$ be their
  respective associated distributive law.  
  Let $\alpha_X \in \categoryLL(F X, G X)$ be a natural
  transformation.
  Then
  $\kleisliCastExp(\alpha_X) \in \kleisliExp(\lift{F} X, \lift{G} X)$
  is natural if and only if the following diagram commutes.
  \begin{equation}
  \label{eq:dl-morphism}
    \begin{tikzcd}
      \Excl F X \arrow[r, "\lambda^F"] \arrow[d, "\Excl \alpha"']
      & F \Excl X \arrow[d, "\alpha"] \\
      \Excl GX \arrow[r, "\lambda^G"']
      & G\Excl X
    \end{tikzcd}
  \end{equation}
\end{proposition}

\begin{proof}
  straightforward computation.
\end{proof}

In the case of differentiation, the axiom ($\devM$-chain) implies that
$\devM \in \category(\Excl \S X, \S \Excl X)$ is a distributive law
between the comonad $\oc$ and the functor $\S$.
This means that $\S$ can be lifted to an endofunctor $\D$ on
$\kleisliExp$.
Besides, there is a trivial distributive law
$\id_{\Excl X} \in \categoryLL(\Excl X,\Excl X)$ associated to the
lifting of the identity functor on $\categoryLL$ to the identity
functor on $\kleisliExp$.
Then ($\devM$-local) is an instance of~\Cref{eq:dl-morphism}
in which $F = \S$, $G = \Id$ and $\alpha = \Sproj_0$.
This means that ($\devM$-local) holds if and only if
 $\lin(\Sproj_0) \in \kleisliExp(\D X, X)$ is a natural transformation.
Thus, $(\D, \lin(\Sproj_0), \lin(\Sproj_1), \lin(\Ssum))$ is a
pre-differential structure on $\kleisliExp$ (in the sense of
\Cref{def:differential-struct}) and \ref{ax:D-chain} holds.

Moreover, since $\D$ is a lifting of $\S$, for any
$h \in \categoryLL(X, Y)$, the morphism
$\lin(h) \in \kleisliExp(X, Y)$ is $\D$-linear.
Indeed,
$\lin(\Sproj_0) \comp \D (\lin(h)) = \lin(\Sproj_0) \comp \lin(\S h) =
\lin(\Sproj_0 \Compl \S h) = \lin(h \Compl \Sproj_0) = \lin(h) \comp
\lin(\Sproj_0)$.
As a result, $\lin(\Sproj_i), \lin(\Ssum), \lin(\prodProj_i)$ are all
linear so \ref{ax:Dproj-lin}, \ref{ax:Dsum-lin} hold and the
pre-differential structure is compatible with the cartesian product.

Furthermore, ($\devM$-lin) consists of two instances of
\Cref{eq:dl-morphism}. The first one is an instance in which
$F = \S$, $G = \Id$ and $\alpha = \Sinj_0 \in \categoryLL(X, \S X)$.
The second one is an instance in which 
$F = \S^2$, $G = \S$ and
$\alpha = \SmonadSum \in \categoryLL(\S^2 X, X)$. Indeed, 
as we saw in \Cref{rem:dl-composite}, there is a distributive law
$\S (\devM_X) \Compl \devM_{\S X} \in \category(\Excl{\S^2} X, \S^2
\Excl X)$ associated to $\D^2$, the lifting of $\S^2$ to
$\kleisliExp$.
So ($\devM$-lin) holds if and only if 
$\lin(\Sinj_0) \in \kleisliExp(\D^2 X, \D X)$ and 
$\lin(\SmonadSum) \in \kleisliExp(\D^2 X, \D X)$ 
are natural transformation, that is if and only if 
\ref{ax:D-add} hold\footnote{As we saw, this gives to $\D$ the structure of a 
Monad on $\kleisliExp$. In fact, ($\devM$-chain) and ($\devM$-lin) taken together 
make $\devM$ a distributive law between the \emph{monad} $\D$ and 
the comonad $\Excl\_$. There is a striking symmetry, because it also allows 
to lift $\Excl\_$ to a comonad on $\kleisliS$ the Kleisli category of $\S$}.

Finally, ($\devM$-Schwarz) consists of an instance of 
\Cref{eq:dl-morphism} in which $F = \S^2$, $G = \S^2$ and
$\alpha = \Sswap$. So ($\devM$-Schwarz) holds if and only if 
$\lin(\Sswap) \in \kleisliExp(\D^2 X, \D^2 X)$ is natural. 
The only lacking axiom is~\ref{ax:D-lin} that corresponds to the
naturality of $\lin(\Slift)$. Thanks to~\Cref{prop:dl-morphism}, it
would hold if and only if the diagram below commutes.
\begin{equation}
\begin{tikzcd}
  \Excl \S X \arrow[d, "\Excl \Slift"'] \arrow[rr, "\devM_X"] &                                    & \S\Excl X \arrow[d, "\Slift_{\Excl X}"] \\
  \Excl{\S^2} X \arrow[r, "\devM_{\S X}"'] & \S ! \S X \arrow[r, "\S
  \devM_X"'] & \S^2 !X
\end{tikzcd}
\end{equation}
This diagram is not mentioned in~\cite{Ehrhard21} but makes perfectly
sense in the setting of coherent differentiation in \LL{} and holds in
all known \LL{} models of coherent differentiation.
The study of the consequences of this diagram is left for further
work.
This ends the proof of \cref{prop:cd-induces-cdc}.

\begin{remark} 
The only remaining axiom is ($\devM$-$\with$) that deals with the
Seely isomorphisms
$\seely{n} \in \categoryLL(! X_0 \tensor \ldots \tensor !X_n, !(X_0
\with \ldots \with X_n))$ of the Seely category $\categoryLL$.
It is possible to define in $\LL{}$ a notion of multilinearity: given any
$l \in \categoryLL(X_0 \tensor \ldots \tensor X_n, Y)$, one can define
$\mlin{l} \in \kleisliExp(X_0 \with \ldots \with X_n, Y)$ as
$\mlin{l} = l \Compl (\der \tensor \ldots \tensor \der) \Compl
(\seely{n})^{-1}$.
Then a morphism in $\kleisliExp(X_0 \with \ldots \with X_n, Y)$ is
\emph{$(n+1)$-linear} (in the sense of $\LL{}$) if it can be written as
$\mlin{h}$ for some $h$.
The axiom ($\devM$-$\with$) allows to show that any $(n+1)$-linear
morphism in the sense of \LL{} is also $(n+1)$-linear in the
sense of \Cref{def:multilinear}.
A proof of this fact can be implicitly found in Theorem 4.26
of~\cite{Ehrhard22-pcf}.
This is a crucial fact, because it shows that what really matters is
the $(n+1)$-linearity in terms of CCDC rather than the
$(n+1)$-linearity in terms of $\LL{}$.
\end{remark}

Many models of $\LL{}$ have a coherent differential structure, such as coherence
spaces, non-uniform coherence spaces and probabilistic coherence
spaces.
Thus, their Kleisli categories are all CCDCs. This provides a rich
variety of examples.
We present here the example of probabilistic coherence spaces.

\subsection{The example of probabilistic coherence spaces}

\label{sec:apcoh}

A \emph{probabilistic coherence space} (PCS)~\cite{DanosEhrhard08} is
a pair %
\(X=(\Web X,\Pcoh X)\) where \(\Web X\) is a set and %
\(\Pcoh X\subseteq\Realpto{\Web X}\) satisfies %
\(\Pcoh X=\{x\in\Realpto{\Web X}\St\forall x'\in\cP'\
\Eval{x}{x'}\defEq\sum_{a\in\Web X}x_ax'_a\leq 1\}\) for some
\(\cP'\subseteq\Realpto{\Web X}\) called a \emph{predual} of \(X\).
To avoid \(\infty\) coefficients it is also assumed that
\(\forall a\in\Web X\ 0<\sup_{x'\in\cP'}x'_a<\infty\) and then it
is easily checked that for all %
\(\forall a\in\Web X\ 0<\sup_{x\in\Pcoh X}x_a<\infty\).

A multiset of elements of a set \(I\) is a function \(m:I\to\Nat\)
such that the set \(\Supp m=\{i\in I\St m(i)\not=0\}\) is finite. The
set \(\Mfin I\) of these multisets is the free commutative monoid
generated by \(I\). We use \(\Mset{\List i1k}\) for the
\(m\in\Mfin I\) such that \(m(i)=\Card{\{j\St i_j=i\}}\), for
\(\List i1k\in I\).

Given PCSs \(X\) and \(Y\), a function \(f:\Pcoh X\to\Pcoh Y\) is
\emph{analytic}%
\footnote{There is also a purely functional characterization of these
  functions as those which are totally monotone and Scott continuous,
  see~\cite{Crubille18}} %
if there is a \emph{matrix}
\(t\in\Realpto{\Mfin{\Web X}\times\Web Y}\) such that, for all
\(x\in\Pcoh X\) and \(b\in\Web Y\), one has %
\(f(x)_b=\sum_{(m,b)\in\Mfin{\Web X}\times\Web Y}t_{m,b}x^m\) %
where \(x^m=\prod_{a\in\Web X}x_a^{m(a)}\).
Thanks to the fact that all the coefficients in \(t\) are finite, it
is not difficult to see that they can be recovered from the function
\(f\) itself by means of iterated differentiation,
see~\cite{DanosEhrhard08}.
So an analytic function has \emph{exactly one} associated matrix. 

The identity function \(\Pcoh X\to\Pcoh X\) is analytic (of matrix
\(t\) given by \(t_{m,a}=\Kronecker{m}{\Mset a}\)) and the composition
of two analytic functions is still analytic.
We use \(\ACOH\) for the category whose objects are PCSs and morphisms
are analytic functions.
For instance, if \(\Sone\) is the PCS \((\Eset\ast,\Intercc01)\) then
\(f_1,f_2:\Intercc01\to\Intercc01\) given by \(f_1(x)=1-\sqrt{1-x^2}\)
and \(f_2(x)=e^{x-1}\) are in \(\ACOH(\Sone,\Sone)\), but
\(f_3(x)=2x-x^2\) is not because of the negative coefficient.
The (pointwise) sum of two analytic functions \(\Pcoh X\to\Pcoh Y\) is
always well defined \(\Pcoh X\to\Realp^{\Web Y}\), but is not
necessarily in \(\ACOH(X,Y)\) so \(\ACOH\) is not left-additive%
\footnote{At least for this most natural addition.}.

If \(X\) is a PCS then
\(\D X
=(\{0,1\}\times\Web X,
\Pcohp{\D X}
=\{z\in\Realpto{\{0,1\}\times\Web X}
\St\Dproj_0(z)+\Dproj_1(z)\in\Pcoh X\})\),
where \(\Dproj_i(z)_a=z_{i,a}\), is a PCS.
Then \(\Dproj_0,\Dproj_1\in\ACOH(\D X,X)\) and we have also %
\(\Dsum\in\ACOH(\D X,X)\) given by
\(\Dsum(z)=\Dproj_0(z)+\Dproj_1(z)\).
In other words \(\D X\) is the PCS whose elements are the pairs
\((x,u)\in\Pcoh X^2\) such that \(x+u\in\Pcoh X\).
In that way we have equipped \(\ACOH\) with a left pre-summability
structure and the associated notion of summability is the obvious one:
\(f_0,f_1\in\ACOH(X,Y)\) are summable if their pointwise sum
\(f_0+f_1\) is in \(\ACOH(X,Y)\)
(the matrix of this sum is the sum of the matrices of \(f_0\) and
\(f_1\)).
It is easily checked that this left pre-summability structure is a
left summability structure (see~\Cref{def:left-summability-struct}).

As explained in~\Cref{sec:differential}, differentiation boils down to
extending the operation \(\D\) to morphisms in such a way that the
conditions of~\Cref{def:CDC} be satisfied.
Given \(f\in\ACOH(X,Y)\) of matrix \(t\) and \((x,u)\in\Pcohp{\D X}\)
we have %
\begin{align*}
f(x+u) &=\sum_{(m,b)\in\Mfin{\Web X}\times\Web 
Y}t_{m,b}(x+u)^m \\ 
&=\sum_{(m,b)\in\Mfin{\Web X}\times\Web Y}t_{m,b}
\sum_{p\leq m}\Binom mpx^{m-p}u^p  \\
&=f(x)+\sum_{a\in\Supp m}\Binom
m{\Mset a}x^{m-\Mset a}u_a+r(x,u) \\
&=f(x)+\sum_{a\in\Supp
  m}m(a)x^{m-\Mset a}u_a+r(x,u)
  \end{align*}
  where
\(\Binom mp=\prod_{a\in\Web X}\Binom{m(a)}{p(a)}\in\Nat\) when
\(p\leq m\) for the pointwise order.
In these expressions the remainder \(r(x,u)\) is a power series in
\(x\) and \(u\) all of whose monomials have total degree \(>1\) in \(u\)
(such as \(x_au_bu_c\) if \(a,b,c\in\Web X\)).
In particular \(\Norm{r(x,u)}\in o(\Norm u)\) where
\(\Norm x=\sup\{\Eval x{x'}\St x'\in\cP'\}\in\Intercc 01\) for any
predual of \(X\) (this norm does not depend on the choice of
\(\cP'\)).
Using~\Cref{def:differential-struct} we set
\[ \dcoh f(x,u)=\sum_{a\in\Supp m}m(a)x^{m-\Mset a}u_a.\] Since all
coefficients of \(t\) are \(\geq 0\) we have
\(f(x)+\dcoh f(x,u)\leq f(x+u)\) for the pointwise order so that %
\(\D f(x,u)=(f(x),\dcoh f(x,u))\in\Pcoh{(\D Y)}\).
In that way we have defined an analytic function
\(\D f\in\ACOH(\D X,\D Y)\) and it is easily checked that \(\ACOH\) is
a coherent differential category in the sense of~\Cref{def:CDC}.
For the two examples above we get \(\dcoh f_2(x,u)=e^{x-1}u\) and
\(\dcoh f_1(x,u)=xu/\sqrt{1-x^2}\) which seems to be undefined when
\(x=1\) but is not because then \emph{we must have} \(u=0\) and so
\(\dcoh f_1(1,0)=0\).

An analytic \(f\in\ACOH(X,Y)\) is \emph{linear} if its matrix \(t\)
satisfies that whenever \(t_{m,b}\not=0\), one has \(m=\Mset a\) for
some \(a\in\Web X\).
This notion of linearity%
\footnote{Which arises from the fact that \(\ACOH\) is the Kleisli
  category of the comonad ``\(\oc\)'' on the PCS model of LL
  of~\cite{DanosEhrhard08}.} %
coincides with both additivity~\Cref{def:additive} and
\(\D\)-linearity~\Cref{def:linear}.

The category \(\ACOH\) is cartesian, with \(\top=(\emptyset,\{0\})\)
and
\(X\with Y =(\{0\}\times\Web X\cup\{1\}\times\Web Y),
\{z\in\Realpto{\{0\}\times\Web X\cup\{1\}}\St\prodProj_0(z)\in\Pcoh
X\text{ and }\prodProj_1(z)\in\Pcoh Y\}\) which is easily seen to be a
PCS (\(\prodProj_i\) is defined exactly as \(\Dproj_i\)) such that
\(\Pcohp{X\with Y}=\Pcoh X\times\Pcoh Y\) up to a trivial
bijection.
The projections \(\prodProj_i\) are additive, and \(\prodSwap\)
(see~\Cref{sec:cartesian-summability}) is an iso: if
\(((x,u),(y,v))\in\Pcohp{\D X\with\D Y}\) then
\(((x,y),(u,v))\in\Pcohp{\D(X\with Y)}\) since
\((x,y)+(u,v)=(x+u,y+v)\) so the summability structure is compatible
with the cartesian product by~\Cref{cor:summability-cartesian-compat}.

An \(f\in\ACOH(X_0\with X_1,Y)\) is \emph{bilinear} in \(X_0,X_1\) if
it is linear (or additive) separately in both inputs, which is
equivalent to saying that its matrix \(t\) satisfies that if
\(t_{m,b}\not=0\) then \(m=\Mset{(0,a_0),(1,a_1)}\) with
\(a_i\in\Web{X_i}\) for \(i=0,1\).
Let \(\SNat=(\Nat,\{x\in\Realpto\Nat\St\sum_{n\in\Nat}x_n\leq 1\})\)
which represents the type of integers in \(\ACOH\), then the function
\(h:\ACOH(\SNat\with\SNat\with\SNat,\SNat)\) given by
\(h(u,x,y)=u_0x+(\sum_{n=1}^\infty u_n)y\) is bilinear in \(\SNat\),
\(\SNat\with\SNat\) and can be understood as an \(\mathtt{ifzero}\)
operator.
The function \(k\in\ACOH(\SNat,\SNat)\) such that \(k(x)_n=x_{n+1}\)
is linear and represents the successor operation.


\section{Link with cartesian differential categories} 

\label{sec:CDC}
 
We show in this section that CCDCs are a generalization of cartesian
differential categories~\cite{Blute09}.

\subsection{Cartesian left additive categories}

\label{sec:comparison-sum}

We rely on the presentation of~\cite{Lemay18} for left additive
categories, since this article uses a minimal set of assumptions.

\begin{definition} \label{def:left-additive} %
  A left additive category is a category such that each hom-set is a
  commutative monoid, with addition $+$ and zero $0$ commuting
  with composition on the right, that is
  $(f+g) \comp h = f \comp h + g \comp h$ and $0 \comp f = 0$.
\end{definition}

\begin{definition}
  A morphism $h$ is additive if addition is compatible with composition
  with $h$ on the left, that is
  $h \comp (f + g) = h \comp f + h \comp g$ and $h \comp 0 = 0$.
  Note that the identity is additive, and additive morphisms are
  closed under addition and composition.
\end{definition}

\begin{definition} 
  \label{def:additive-total} %
  A cartesian left additive category is a left additive category such
  that the projections are additive.
\end{definition}

Given a cartesian left additive category $\category$, one can
define a summable pairing structure
(\Cref{def:pre-presummability-structure})
$(\Dwith, \prodProj_0, \prodProj_1, \prodProj_0 + \prodProj_1)$ with
$\Dwith X = X \with X$.
Then one can check that all morphisms are summable (the witness of
$f \summable g$ is $\prodPair{f}{g}$).
Moreover the left additivity of the category ensures that the notion
of sum induced by
$(\Dwith, \prodProj_0, \prodProj_1, \prodProj_0 + \prodProj_1)$
coincides with the native structure of monoid on the hom-sets.
In particular, a morphism is additive in the sense of
\Cref{def:additive} if and only if it is additive in the sense of
\Cref{def:additive-total}.
Consequently, $\prodProj_0, \prodProj_1$ and
$\prodProj_0 + \prodProj_1$ are additive. Thus,
$(\Dwith, \prodProj_0, \prodProj_1, \prodProj_0 + \prodProj_1)$ is a
left pre-summability structure.
Finally, it is a left summability structure because 
\ref{ax:D-witness} trivially holds (everything is summable), 
and \ref{ax:D-zero}, \ref{ax:D-com} hold thanks to the fact that 
everything is summable and that $(\category(X, Y), +, 0)$ is a 
commutative monoid.

Conversely any left summability structure on $\category$ of shape
$(\Dwith, \prodProj_0, \prodProj_1, \Dsum)$ with
$\D_{\with} X = X \with X$
endows each hom-set with a commutative monoid structure and
\Cref{prop:sum-left-compatible} ensures that the category is left
additive.
Then, as above, a morphism is additive in the sense of
\Cref{def:additive} if and only if it is additive in the sense of
\Cref{def:additive-total}.
Thus $\prodProj_0, \prodProj_1$ are additive so the category is
cartesian left additive.
Moreover $\Dsum = \prodProj_0 + \prodProj_1$ by \Cref{prop:proj-sum}
so the left summability structure induced by the monoid on the hom-set
coincides with the left summability structure we started from.
We just proved \Cref{prop:summability-total} below.

\begin{theorem} \label{prop:summability-total} %
  Let $\category$ be a cartesian category.
  Define $\Dwith X = X \with X$.
  There is a bijection between the monoid structures on the hom-set
  that make $\category$ a cartesian left additive category and the
  left summability structures $(\D, \Dproj_0, \Dproj_1, \Dsum)$ on
  $\category$ such that $\D = \Dwith$, $\Dproj_0 = \prodProj_0$ and
  $\Dproj_1 = \prodProj_1$.
\end{theorem}

\begin{remark}
  Any left summability structure on $\category$ of shape
  $(\Dwith, \prodProj_0, \prodProj_1, \Dsum)$ with
  $\Dwith X = X \with X$ is \emph{de facto} compatible with the
  cartesian product.
  The additivity of $\prodProj_0$ and $\prodProj_1$ is part of the
  axioms of summability, and $\prodSwap$ is an isomorphism thanks to 
  \Cref{prop:prodSwap-inverse} and the fact that everything is summable.

\end{remark}

\subsection{Cartesian differential categories}

\label{sec:comparison-diff}

We give the axioms of a cartesian differential category following the
alternative formulation of~\cite{Cockett14} for convenience.

\begin{definition} \label{def:CartesianDC} %
  A cartesian differential category is a cartesian left additive
  category $\category$ equipped with a differential combinator $\d$
  that maps each morphism $f \in \category(X, Y)$ to a morphism
  $\d{f} \in \category(X \with X, Y)$ such that
  \begin{enumerate}
  \item $\d{\prodProj_0} = \prodProj_0 \comp \prodProj_1$,
    $\d{\prodProj_1} = \prodProj_1 \comp \prodProj_1$;
  \item $\d{0} = 0$ and $\d{(f+g)} = \d{f} + \d{g}$;
  \item $\d{\id} = \Dproj_1$ and
    $\d {(g \comp f)} = \d{g} \comp \prodPair{f \comp \Dproj_0}{\d{f}}$;
  \item $\d{f} \comp \prodPair{x}{0} = 0$ and
    $\d{f} \comp \prodPair{x}{u+v} = \d{f} \comp \prodPair{x}{u} +
    \d{f} \comp \prodPair{x}{v}$;
  \item
    $\d{\d{f}} \comp \prodPair{\prodPair{x}{0}}{\prodPair{0}{u}} =
    \d{f} \comp \prodPair{x}{u}$;
  \item
    $\d{\d{f}} \comp \prodPair{\prodPair{x}{u}}{\prodPair{v}{w}} =
    \d{\d f} \comp \prodPair{\prodPair{x}{v}}{\prodPair{u}{w}}$.
  \end{enumerate}
\end{definition}

Note that the axiom $\d{\id} = \prodProj_1$ seems to be missing from
the axioms given in~\cite{Cockett14}, although it can be found in the
original formulation in~\cite{Blute09}.
There is usually another axiom, that states that
$\d{\prodPair{f}{g}} = \prodPair{\d{f}}{\d{g}}$.
But as observed in \cite{Lemay18}, this axiom is a consequence of the
linearity of the projections and of the chain rule so we discard it.

Let $\category$ be a left additive category.
As stated in \Cref{prop:summability-total}, the structure of
monoid in the hom-set arises from a
summability structure
$(\Dwith, \prodProj_0, \prodProj_1, \prodProj_0 + \prodProj_1)$
compatible with the cartesian product.
Then, there is a bijection between pre-differential structures on top
of this summability structure and differential combinators in the sense
of \Cref{def:CartesianDC}: we can define the functorial action of
$\Dwith$ from $\d$ as
$\Dwith{f} \defEq \prodPair{f \comp \prodProj_0}{\d {f}}$, and we can
define $\d$ from $\Dwith$ as $\d{f} = \prodProj_1 \comp \Dwith{f}$.

Besides, we have shown in \Cref{sec:equivalence-lemmas} that
the axioms of coherent differentiation are equivalent to some 
equational properties on $\dcoh$.
When the underlying left summability structure is
$(\Dwith, \prodProj_0, \prodProj_1, \prodProj_0 + \prodProj_1)$, those
properties turn out to be exactly the axioms of cartesian differential
categories.
The axiom (1) corresponds to \ref{ax:Dproj-lin}.
%
By \Cref{prop:Dsum-lin},~(2) corresponds to \ref{ax:Dsum-lin}.
%
By \Cref{prop:D-chain},~(3) corresponds to \ref{ax:D-chain}.
By \Cref{prop:derivative-additive-zero,prop:derivative-additive-sum},~(4) 
corresponds to \ref{ax:D-add}.
By \Cref{prop:D-lin},~(5) corresponds to \ref{ax:D-lin}.
By \Cref{prop:D-schwarz},~(6) corresponds to \ref{ax:D-schwarz}.

Finally, the differential structures on top of the left summability 
structure $(\Dwith, \prodProj_0, \prodProj_1, \Dsum)$
 are \emph{de facto} compatible with the cartesian product,
because the linearity of $\prodProj_0$ and $\prodProj_1$ is included
in~(1).
%


\begin{theorem} \label{prop:ccdc-cdc}
  The cartesian differential categories are exactly the cartesian
  coherent differential categories in which $\D X = X \with X$,
  $\Dproj_0 = \prodProj_0$, $\Dproj_1 = \prodProj_1$.
\end{theorem}

\begin{remark}
  In \cite{Blute09}, $h$ is said to be linear if
  $\d (h) = h \comp \prodProj_1$.
  Then \Cref{prop:linear-equation} ensures that this notion of
  linearity exactly corresponds through \Cref{prop:ccdc-cdc} to our
  notion of $\D$-linearity introduced in \Cref{def:linear}.
\end{remark}

\begin{remark} %
  \label{rk:tangent-categories} %
  Every cartesian differential category is also a tangent
  category~\cite{Cockett14}, and the tangent functor induced from $\d$
  is exactly the same functor as $\D_{\with}$.
  This makes sense, as coherent differentiation and tangent categories
  are very similar: they extend cartesian differential categories by
  generalizing addition in two different ways.
\end{remark}

\section{A first order coherent differential language}

We introduce a first order language associated to these models.  Note
that a development of a whole coherent differential PCF of which our
language can be roughly considered as a fragment can already be found
in \cite{Ehrhard22-pcf}, with a semantics based on~\cite{Ehrhard21}.
Our main contribution here is that CCDCs provide the tools for a
more principled and synthetic treatment of the semantics.
This tighter connection between syntax and semantics allows for the
development of new ideas, such as a more systematic treatment of
multilinearity.

\subsection{Terms}

\begin{definition}
  Le $\varTypes$ be a set of ground type symbols, ranged over by
  $\alpha, \beta, \ldots$
  For any $\alpha \in \varTypes$ and $h \in \N$, $\D^h \alpha$ is a
  ground type.  General types are inductively defined by
  \[
    A, B, C \defEq \D^h \alpha \mid A \with B\,.
  \]
\end{definition}

For any type $A$, we define the type $\D A$ inductively on $A$ by
$\D \D^h \alpha = \D^{h+1} \alpha$ and
$\D (A \with B) = \D A \with \D B$.

\begin{definition}
  Let $\multiLinVar, \multiLinVar[1], \ldots$ be function symbols.
  Each function symbol $\multiLinVar$ is uniquely assigned a
  \emph{function type} of the form $A_0, \ldots, A_n \arrow B$ where
  $A_i$ and $B$ are types.
  Then, $n+1$ is called the arity of $\multiLinVar$, denoted as
  $\arr(\multiLinVar)$.
\end{definition}

A function symbol $\multiLinVar$ of type $A_0, \ldots, A_n \arrow B$
will be interpreted in section \Cref{sec:semantics} as a 
$(n+1)$-linear morphisms
$\sem{\multiLinVar} \in \category(\sem{A_0} \with \cdots \with
\sem{A_n}, \sem{B})$ (recall
\Cref{def:multilinear}).
Note that the types $A_i$ can themselves be products and need not be
ground types.
For example, a $2$-linear map in $\category((A \with B) \with C, D)$
can by no means be seen as a $3$-linear map in
$\category(A \with B \with C, D)$.

\begin{definition} \label{def:functions} %
  Define \emph{functions} as
  \[
    \multiLin, \multiLin[1], \ldots \defEq \multiLinVar \mid
    \Dproj_i^{A} \mid \prodProj_i^{A, B} \mid \Dinj_i^{A} \mid
    \DmonadSum_n^{A}
  \]
  where $i \in \{0, 1\}$, $n \geq 0$, $\multiLinVar$ are function
  symbols and $A, B$ are types.
  Each function $\multiLin$ has a function type:
  $\Dproj_0^A, \Dproj_1^A$ have type $\D A \arrow A$,
  $\Dinj_0^A, \Dinj_1^A$ have type $A \arrow \D A$, the
  $\DmonadSum_n^A$ have type $\D^{n+1} A \arrow \D A$ and
  $\prodProj_0^{A,B}, \prodProj_1^{A, B}$ have types
  $A \with B \arrow A$ and $A \with B \arrow B$ respectively.
  Notice that projections have arity $1$ and not $2$.
  The type attached to the constructors $\Dproj_i$, $\prodProj_i$,
  $\Dinj_i$ and $\DmonadSum_n$ will always be kept implicit in what
  follows.
\end{definition}

\begin{remark}
  Taking $n=-1$ allows to write constants.
\end{remark}

\begin{definition}
  Let $\Var$ be a set of variable symbols. The set $\terms$ of terms
  is defined inductively as follows
  \[
    t, u, \ldots \defEq \prodPair{t_0}{t_1} \mid
    \multiLin^{\word}(t_0, \ldots, t_n) \mid x
  \] 
  where $x \in \Var$, $\multiLin$ are function symbols of arity $n+1$
  and $\word \in \interval{0}{n}^\ast $, the set of finite words%
  \footnote{Such a word represents a successive application of partial
    derivatives on the multilinear symbol $f$, more on this
    in~\Cref{sec:semantics}.} %
  of elements of \(\interval0n\).
\end{definition}

\begin{remark}
  Nothing prevents us from adding to this calculus non multilinear
  function symbols, assuming that the formal derivatives for the
  function symbols are also provided.
  We focus on multilinear functions though, due to the nature of the
  basic operations of PCF.
  A coherent differential PCF would contain a base type $\nat$, two
  function symbols $\predCons$ and $\succCons$ of type
  $\nat \arrow \nat$, a family of function symbols $\ifCons^A$ of type
  $\nat, A \with A \arrow A$ (conditional) and a family of function
  symbols $\letCons^A$ of type $\nat, (\nat \arrow A) \arrow A$
  (call-by-value on the type of integers).
  An analysis of the semantics of these symbols in coherent
  differentiation in the $\LL{}$ setting of~\cite{Ehrhard22-pcf} or in the
  example of~\Cref{sec:apcoh} indeed shows that $\predCons$ and
  $\succCons$ should be interpreted as linear morphisms, and that
  $\ifCons^A$ and $\letCons^A$ should be interpreted as $2$-linear
  morphisms.
  Using the fact that variables can be used in a non-linear way as
  well as the PCF fixpoint operator, it is then possible to write terms
  whose interpretation is not multilinear.
  For instance, \(f_1\) of~\Cref{sec:apcoh} is the semantics of a
  term, see~\cite{Ehrhard22a}.
\end{remark}

\begin{notation}
  For any word $\word$, we write $\wordLength{\word}$ for its length,
  and $\wordLetter{\word}{j}$ for the number of occurrences of the
  letter $j$.
  We will write $f$ for $f^{\emptyWord}$, where $\emptyWord$ is the
  empty word.
  Notice that when $\arr (\multiLin) = 0$, a word
  $\word \in \interval{0}{0}^\ast$ can be uniquely seen as an integer
  $d = |\word|$.
  We will then write $\depth{\multiLin}{d}$ for $\multiLin^{\word}$.
\end{notation}
We introduce the typing rules in \Cref{fig:typing}.
The systematic treatment of multilinear morphisms allows for a great
factorization of the rules.
We write $\multiLin: A_0, \ldots, A_n \arrow B$ if $f$ has type
$A_0, \ldots, A_n \arrow B$.
\begin{figure}
\begin{center}
  \begin{prooftree}
    \hypo{x:A \in \Gamma} \infer1[(Var)]{\Gamma \vdash x : A}
  \end{prooftree}\quad
  \begin{prooftree}
    \hypo{\Gamma \vdash t_0 : A} \hypo{\Gamma \vdash t_1 : B}
    \infer2[(Pair)]{\Gamma \vdash \prodPair{t_0}{t_1} : A \with B}
  \end{prooftree}
\end{center}
\begin{center}    
  \begin{prooftree}
    \hypo{\multiLin: A_0, \ldots, A_n \arrow B}
    \hypo{\zeta\in\interval 0n^\ast}
    \hypo{
      (\Gamma \vdash t_i: \D^{\wordLetter{\word}{i}} A_i)_{i=0}^n}
    \infer3[(App)]{\Gamma \vdash \multiLin^{\word}(t_0, \ldots, t_n) :
      \D^{|\word|} B}
  \end{prooftree}
\end{center}
\caption{Typing rules}
\label{fig:typing}
\end{figure}
Given any term $t$, one can define a term $\diffTerm{t}{x}$ by
induction on $t$.
The inductive steps are given in \Cref{fig:differential}. 

%

\begin{figure}
\begin{align*}
  \diffTerm{y}{x}
  &=\begin{cases}
    x & \text{ if $y=x$} \\
    \Dinj_0(y) & \text{ otherwise}
  \end{cases} \\
  \diffTerm{\prodPair{t_0}{t_1}}{x}
  &=\prodPair{\diffTerm{t_0}{x}}{\diffTerm{t_1}{x}}  \\
  \diffTerm{f^{\word}(t_0, \ldots, t_n)}{x}
  &=\DmonadSum_n(f^{\word n \cdots 1 0}
    (\diffTerm{t_0}{x}, \ldots, \diffTerm{t_n}{x}))
\end{align*}
\caption{Differential of a term}
\label{fig:differential}
\end{figure}

\begin{proposition} \label{prop:typing-differential}
If $\Gamma, x : A \vdash t : B$ then 
$\Gamma, x : \D A \vdash \diffTerm{t}{x} : \D B$ 
\end{proposition}

\begin{proof}
  By induction on the typing derivation.  %
  \begin{itemize}
  \item If the last rule applied is \varRule{} then the first
    possibility is that $t = x$ and $\Gamma, x : A \vdash x : A$.  But
    then, $\diffTerm{x}{x} = x$ and
    $\Gamma, x : \D A \vdash x : \D A$. The second possibility is that
    $t = y$ with $y \neq x$ and $\Gamma \vdash y : B$. But then,
    $\diffTerm{y}{x} = \Dinj_0(y)$ and
    $\Gamma \vdash \Dinj_0(y) : \D B$. Thus,
    $\Gamma, x : \D A \vdash \Dinj_0(y) : \D B$ in both cases.
  \item If the last rule applied is \pairRule, then
    $t = \prodPair{t_0}{t_1}$, $t$ is of type $B_0 \with B_1$,
    $\Gamma, x : A \vdash t_0 : B_0$ and
    $\Gamma, x : A \vdash t_1 : B_1$. But
    $\diffTerm{t}{x} =
    \prodPair{\diffTerm{t_0}{x}}{\diffTerm{t_1}{x}}$.  By induction
    hypothesis 
    $\Gamma, x : \D A \vdash \diffTerm{t_0}{x} : \D B_0$ and
    $\Gamma, x : \D A \vdash \diffTerm{t_1}{x} : \D B_1$.  Thus, by
    applying \pairRule,
    $\Gamma, x : \D A \vdash
    \prodPair{\diffTerm{t_0}{x}}{\diffTerm{t_1}{x}} : \D B_0 \with \D
    B_1$.  But $\D B_0 \with \D B_1 = \D (B_0 \with B_1)$ so
    $\Gamma, x : \D A \vdash \diffTerm{\prodPair{t_0}{t_1}}{x} : \D
    (B_0 \with B_1)$.
  \item If the last rule applied is \appRule{} then
    $t = f^{\word}(t_0, \ldots, t_n)$, 
    $f$ has some type $A_0, \ldots, A_n \arrow B$, and
    $\Gamma, x : A \vdash t : \D^{\wordLength{\word}} B$.
    Besides, for any $i$, $\Gamma, x : A \vdash t_i : \D^{\wordLetter{\word}{i}} A_i$.
    By induction hypothesis,
    $\Gamma, x : \D A \vdash \diffTerm{t_i}{x} :
    \D^{\wordLetter{\word}{i} + 1} A_i$. But
    $\wordLetter{\word n \cdots 1 0}{i} = \wordLetter{\word}{i} + 1$
    so applying the \appRule{} rule gives a derivation for
    $\Gamma, x : \D A \vdash f^{\word n \cdots 1 0}(\diffTerm{t_0}{x},
    \ldots, \diffTerm{t_n}{x}) : \D^{\wordLength{\word}+n+1} B$.
    Applying the \appRule{} rule again for $f=\DmonadSum_n$ yields a derivation of 
    $\Gamma, x : \D A \vdash \DmonadSum_n(f^{\word n \cdots 1 0}(\diffTerm{t_0}{x},
    \ldots, \diffTerm{t_n}{x})) : \D^{\wordLength{\word}+1} B$, which concludes the proof.
\end{itemize}
\end{proof}

\subsection{Semantics}
\label{sec:semantics}

Let $\category$ be a CCDC.
For the sake of simplicity, we assume that
$\D (X \with Y) = \D X \with \D Y$ and $\prodSwap=\id$%
\footnote{This assumption is by no mean necessary but it simplifies
the notations and the results}.
Assume that we are given an object $\sem{\alpha}$ of \(\category\) for
any ground type symbol $\alpha$.
Then one can interpret any type as an object:
$\sem{\D^h \alpha} = \D^h \sem{\alpha}$ and
$\sem{A \with B} = \sem{A} \with \sem{B}$.
It follows by a straightforward induction that
$\sem{\D A} = \D \sem{A}$.
This interpretation extends as usual to contexts, setting
$\sem{x_0 : A_0, \ldots, x_n : A_n} = \sem{A_0} \with \cdots \with
\sem{A_n}$.
The semantics of the empty context is $\top$.

Assume that we are given a $(n+1)$-linear morphism
$\sem{\multiLinVar} \in \category(\sem{A_0} \with \cdots \with
\sem{A_n}, \sem{B})$ for any function symbol
$\multiLinVar : A_0, \ldots, A_n \arrow B$.
Then any function $\multiLin : A_0, \ldots, A_n \arrow B$ can be
interpreted as an $(n+1)$-linear morphism $\sem{f}$ by setting
$\sem{\Dproj_i} = \Dproj_i$, $\sem{\Dinj_i} = \Dinj_i$,
$\sem{\DmonadSum_n} = \DmonadSum^n$ 
(as defined in \Cref{def:natural-trans-iterate}) 
and $\sem{\prodProj_i} = \prodProj_i$.

\begin{remark}\label{rem:pair-differential} %
  Since $\prodSwap=\id$, we have
  $\D \prodProj_i = \D \prodProj_i \comp (\prodSwap)^{-1} = \D
  \prodProj_i \comp \Dpair{ \Dproj_0 \with \Dproj_0}{\Dproj_1 \with
    \Dproj_1} = \Dpair{\prodProj_i \comp (\Dproj_0 \with
    \Dproj_0)}{\prodProj_i \comp (\Dproj_1 \with \Dproj_1)} =
  \Dpair{\Dproj_0 \comp \prodProj_i}{\Dproj_1 \comp \prodProj_i} =
  \prodProj_i$.
  Notice also that
  $\prodPair{\D f_0}{\D f_1} = \D \prodPair{f_0}{f_1}$ by
  \Cref{prop:prodpair-derivative}
\end{remark}

\begin{theorem}
  For any term $t$ such that $\Gamma \vdash t : A$, we can define
  $\sem{t}_{\Gamma} \in \category(\sem{\Gamma}, \sem{A})$.
  %
  %
  %
\end{theorem}

\begin{proof}
  We proceed by induction on the term. 
  \begin{itemize}

  \item 
  If $t = x$ then the last typing rule must be (Var) so that
  $\Gamma = \Gamma_0, x : A, \Gamma_1$.
  Define
  $\sem{x}_{\Gamma} = \prodProj_{|\Gamma_0|} \in
  \category(\sem{\Gamma_0} \with \sem{A} \with \sem{\Gamma_1},
  \sem{A})$.

  \item 
  If $t = \prodPair{t_0}{t_1}$ then the last typing rule must be
  (Pair), so $t$ is of type $A \with B$, $\Gamma \vdash t_0 : A$ and
  $\Gamma \vdash t_1 : B$.
  By induction, one can define
  $\sem{t_0}_{\Gamma} \in \category(\sem{\Gamma}, \sem{A})$ and
  $\sem{t_1}_{\Gamma} \in \category(\sem{\Gamma}, \sem{B})$.
  Then we define
  $\sem{\prodPair{t_0}{t_1}}_{\Gamma} =
  \prodPair{\sem{t_0}_{\Gamma}}{\sem{t_1}_{\Gamma}} \in
  \category(\sem{\Gamma}, \sem{A \with B})$.

  \item 
  If $t = f^{\word}(t_0, \ldots, t_n)$ with
  $f : A_0, \ldots, A_n \arrow B$ then the last typing rule must be
  (App).
  That is, $t$ must be of type $D^{\wordLength{\word}} B$ for some
  type $B$ and for \(i=0,\dots,n\) we have a derivation of
  $\Gamma \vdash t_i : \D^{\wordLetter{\word}{i}}A_i$.
  By inductive hypothesis, we can define
  $\sem{t_i}_{\Gamma} \in \category(\sem{\Gamma},
  \sem{D^{\wordLetter{\word}{i}} A_i})$.
  But
  $\sem{D^{\wordLetter{\word}{i}} A_i} = \D^{\wordLetter{\word}{i}}
  \sem{A_i}$ and
  $\D_{\word_k} \ldots \D_{\word_1} \sem{f} \in
  \category(\D^{\wordLetter{\word}{0}} \sem{A_0} \with \cdots \with
  \D^{\wordLetter{\word}{n}} \sem{A_n}, \D^{\wordLength{\word}}
  \sem{B})$.
  Thus, we can set
  $\sem{f^{\word_1 \cdots \word_k}(t_0, \ldots, t_n)}_{\Gamma} =
  \D_{\word_k} \ldots \D_{\word_1} \sem{f} \comp
  \prodPair{\sem{t_0}_{\Gamma}, \ldots} {\sem{t_n}_{\Gamma}}$.
  \end{itemize}
\end{proof}

\begin{notation}
  We use $\sem{x}_{\Gamma} = \prodProj_x$ for the projection on
  $\sem{\Gamma}$ to the coordinate where $x$ appears in \(\Gamma\).
\end{notation}

\begin{remark}
  In particular,
  $\sem{\depth{\Dproj_i}{d}(t)} = \D^d \Dproj_i \comp \sem{t}$,
  $\sem{\depth{\Dinj_i}{d}(t)} = \D^d \Dinj_i \comp \sem{t}$,
  $\sem{\depth{\DmonadSum_n}{d}(t)} = \D^d{\DmonadSum^n} \comp
  \sem{t}$.
  More importantly,
  $\sem{\depth{\prodProj_i}{d}(t)} = \D^d \prodProj_i \comp \sem{t} =
  \prodProj_i \comp \sem{t}$ because of our assumption that
  $\prodSwap$ is the identity.
\end{remark}

\begin{notation}
  For any word $\word = \word_1 \cdots \word_k$ in 
  $\interval{0}{n}^k$, define
  $\D_{\word} \defEq \D_{\word_k} \ldots \D_{\word_1}$. Then for any
  $f \in \category(X_0 \with \ldots \with X_n, Y)$,
  $\D_{\word} f \in \category(\D^{\wordLetter{\word}{0}} X_0 \with
  \cdots \with \D^{\wordLetter{\word}{n}} X_n, \D^{\wordLength{\word}}
  Y)$.
  Note that $\D_{\word \cdot \word[1]} = \D_{\word[1]} \D_{\word}$.
  Then, \Cref{prop:DDproj-commute-partial} can be seen as the property
  that for any $f$ $(n+1)$-linear, for any word $\word[1]$ of length
  $d$,
  $\D^d \Dproj_i \comp \D_{\word[1]} \D_j f = \D_{\word[1]} f \comp
  \singleApp{j}{\id}{\D^{\wordLetter{\word[1]}{j}} \Dproj_i}$ %
\end{notation}

The main result of this section on the calculus consists in showing
that the semantics of this syntactical derivative operation
corresponds to the derivative in the model.

\begin{theorem}
  If $\Gamma, x : A \vdash t : B$ then
  $\sem{\diffTerm{t}{x}}_{\Gamma, x : \D A} = \D_1 \sem{t}_{\Gamma, x
    : A}$ where $\sem{t}_{\Gamma, x : A}$ is seen as a morphisms of
  $\category(\sem{\Gamma} \with \sem{A}, \sem{B})$.
\end{theorem}

\begin{proof}
  By induction on $t$.
  \begin{itemize}
  \item If $t = x$ then
    $\sem{t}_{\Gamma, x : A} = \prodProj_1 \in \category(\sem{\Gamma}
    \with \sem{A}, \sem{A})$.
    Then
    $\D_1 \prodProj_1 = \D \prodProj_1 \comp \strengthR = \D
    \prodProj_1 \comp \Dpair{\id \with \Dproj_0}{0 \with \Dproj_1} =
    \Dpair{\prodProj_1 \comp (\id \with \Dproj_0)}{\prodProj_1 \comp
      (0 \with \Dproj_1)} = \Dpair{\Dproj_0 \comp
      \prodProj_1}{\Dproj_1 \comp \prodProj_1} = \prodProj_1$ using
    \Cref{prop:strength} and the linearity of $\prodProj_1$.
  \item If $t = y \neq x$ then
    $\sem{t}_{\Gamma, x : A} = \sem{y}_{\Gamma} \comp \prodProj_0 =
    \prodProj_y \comp \prodProj_0 \in \category(\sem{\Gamma} \with
    \sem{A}, \sem{B})$.
    Then
    $\D_1 (\prodProj_y \comp \prodProj_0) = \D \prodProj_y \comp \D
    \prodProj_0 \comp \strengthR = \D \prodProj_y \comp \D \prodProj_0
    \comp \Dpair{\id \with \Dproj_0}{0 \with \Dproj_1} = \D
    \prodProj_y \comp \Dpair{\prodProj_0 \comp (\id \with \Dproj_0)}
    {\prodProj_0 \comp (0 \with \Dproj_1)} =\D \prodProj_y \comp
    \Dpair{\prodProj_0}{0} = \Dpair{\prodProj_y \comp \prodProj_0}{0}
    = \sem{\Dinj_0(y)} = \sem{\diffTerm{y}{x}}$.
  \item If $t = \prodPair{t_0}{t_1}$, then
    $\sem{\diffTerm{t}{x}} =
    \sem{\prodPair{\diffTerm{t_0}{x}}{\diffTerm{t_1}{x}}} =
    \prodPair{\sem{\diffTerm{t_0}{x}}}{\sem{\diffTerm{t_1}{x}}}
    $. By inductive hypothesis,
    $\sem{\diffTerm{t}{x}} = \prodPair{\D_1 \sem{t_0}}{\D_1 \sem{t_1}}$.
    But
    $\prodPair{\D_1 \sem{t_0}}{\D_1 \sem{t_1}} = \prodPair{\D
      \sem{t_0} \comp \strengthR}{\D \sem{t_1} \comp \strengthR} =
    \prodPair{\D \sem{t_0}}{\D \sem{t_1}} \comp \strengthR$.
    By \Cref{rem:pair-differential}, this is equal to
    $\D \prodPair{\sem{t_0}}{\sem{t_1}} \comp \strengthR = \D_1
    \prodPair{\sem{t_0}}{\sem{t_1}} = \D_1 \sem{t}$.
  \item If $t = f^{\word}(t_0, \ldots, t_n)$ then by definition
    $\diffTerm{t}{x} = \DmonadSum_n(f^{\word n \cdots 1 0}(
    \diffTerm{t_0}{x}, \ldots, \diffTerm{t_n}{x}))$.
    Thus,
    $\sem{\diffTerm{t}{x}} = \DmonadSum^n \comp \D_{n \cdots 1 0}
    \D_{\word} f \comp \prodPair{\sem{\diffTerm{t_0}{x}},
      \ldots}{\sem{\diffTerm{t_n}{x}}} = \DmonadSum^n \comp \D_{n
      \cdots 1 0} \D_{\word} f \comp \prodPair{\D_1 \sem{t_0},
      \ldots}{\D_1 \sem{t_n}}$ by inductive hypothesis.
    But then, the Leibniz rule (\Cref{prop:leibniz-n}) states that
    $\DmonadSum^n \comp \D_{n \cdots 1 0} \D_{\word} f = \D \D_{\word}
    f$.
    Thus,
    $\sem{\diffTerm{t}{x}} = \D \D_{\word} f \comp \prodPair{\D
      \sem{t_0} \comp \strengthR, \ldots}{\D \sem{t_n} \comp
      \strengthR} = \D \D_{\word} f \comp \prodPair{\D \sem{t_0},
      \ldots}{\D \sem{t_n}} \comp \strengthR = 
      \D (\D_{\word} f \comp \prodPair{\sem{t_0},
      \ldots}{\sem{t_n}}) \comp \strengthR =
      \D \sem{t} \comp \strengthR = \D_1 \sem{t}$.
\end{itemize}
\end{proof}

\subsection{Reduction}

We introduce in this section a set of reduction rules that deals with
the differential content of the terms.
The set of rules is more compact than the one given
in~\cite{Ehrhard22-pcf}, but covers all of the rules concerning the
fragment we are looking at.

\begin{remark} 
  We could have added a construct $t[u/x]$ for explicit substitutions,
  with the typing rule
  \begin{center}
    \begin{prooftree}
      \hypo{\Gamma, x : A \vdash t : B} \hypo{\Gamma \vdash u : A}
      \infer2{\Gamma \vdash t[u/x] : B}
    \end{prooftree}(Cut)
  \end{center}
  %
  %
  as well as reduction rules that performs the substitution steps (for
  example, $x[u/x] \red u$).
  We decided not to do so because, in a higher order
  \(\lambda\)-calculus setting, such explicit substitutions are not
  necessary.
\end{remark}

The main difference with the differential lambda-calculus
of~\cite{Ehrhard03} is the absence of sum, because we do not want a
non deterministic typing rules such as
\begin{center}
\begin{prooftree}
	\hypo{\Gamma \vdash t : A}
	\hypo{\Gamma \vdash u : A}
	\infer2{\Gamma \vdash t + u : A}
\end{prooftree}
\end{center}
But the reduction of a $\Dproj_1$ against a $\DmonadSum$ will
introduce sums.
Handling sum without the typing rule above is tricky, because of
subject reduction.
There will be no guarantee indeed that if $\Gamma \vdash t + u : A$
and $t \red t'$ then $\Gamma \vdash t' + u : A$.
For this reason, we chose a conservative approach, by keeping sums as
a formal multiset on top of the terms.

\begin{definition}
  A term multiset is a finite multiset of term.
\end{definition}
See \Cref{sec:apcoh} for the notations we use on multisets.
We define a reduction $\red$ from terms to term multisets.
The reduction rules are given in \Cref{fig:reduction-diff}.
Then we define $\redSym$ as the ``reflexive'' closure of $\red$.
That is, $t \redSym \multisetVar$ if $t \red \multisetVar$ or if
$\multisetVar = [t]$.
It allows to lifts $\red$ to a reduction from a term multiset to a
term multiset in a monadic fashion: if $t_1 \red L_1$ and for all
$i \neq 1$, $t_i \redSym L_i$, then
\[
  [t_1, \ldots, t_n] \redMultiset \sum_{i=1}^{n} L_i
\]
where $\sum$ is the multiset union, that is, the pointwise sum of the
functions $L_i : \terms \arrow \N$.


\begin{figure}[h]
\begin{align*}
\depth{\prodProj_i}{d}(\prodPair{t_0}{t_1}) \red & [t_i] \\
\depth{\Dproj_i}{d} (f^{\word j \word[1]}(t_0, \ldots, t_n)) \red & 
[f^{\word \word[1]}(t_0, \ldots, \depth{\Dproj_i}{\wordLetter{\word[1]}{j}}(t_j), \ldots, t_n)] 
 \text{\quad where $\wordLength{\word[1]} = d$} \\
\depth{\Dproj_i}{d}(\depth{\Dinj_i}{d}(t)) \red & [t] \\
\depth{\Dproj_i}{d}(\depth{\Dinj_{1-i}}{d}(t)) \red & [\,] \\
\depth{\Dproj_0}{d}(\depth{\DmonadSum_n}{d}(t)) \red & 
[(\depth{\Dproj_0}{d})^{n+1} (t) ] \\
\depth{\Dproj_1}{d}(\depth{\DmonadSum_n}{d}(t)) \red & 
\sum_{k=0}^{n} [ (\depth{\Dproj_0}{d})^k \depth{\Dproj_1}{d} (\depth{\Dproj_0}{d})^{n-k} (t) ]\,.
\end{align*}
Here, $(\depth{\Dproj_i}{d})^n$ is a notation for $n$ successive applications of 
$\depth{\Dproj_i}{d}$.
\caption{Reduction rules}
\label{fig:reduction-diff}
\end{figure}

\begin{definition} A term multiset $[t_1, \ldots, t_n ]$ of type $A$ in context $\Gamma$
is $\category$-summable if
$\sem{t_1}_{\Gamma}, \ldots, \sem{t_n}_{\Gamma}$ are summable 
(in the sense of \Cref{prop:arbitrary-sum}).
Then, we define 
$\sem{[t_1, \ldots, t_n ]}_{\Gamma} = \sem{t_1}_{\Gamma} + \cdots + \sem{t_n}_{\Gamma}$.
Note that $[\,]$ is always $\category$-summable, and $\sem{[\,]} = 0$.
\end{definition}

The main point of coherent differentiation is that the reduction
$\red$ will always introduce term multisets that are
$\category$-summable, for any model $\category$.

\begin{theorem}[Invariance of semantics under reduction]
\label{prop:semantic-invariance}
For any $\Gamma \vdash t : A$, if $t \red L$ then $L$ is $\category$-summable and 
$\sem{L}_{\Gamma} = \sem{t}_{\Gamma}$.
\end{theorem}

\begin{proof}
Let us consider every application of the rule $\red$. Note that when a term multiset 
has one element, it is always $\category$-summable and
$\sem{[t]} = \sem{t}$.
\begin{align*}
\sem{\depth{\prodProj_i}{d}(\prodPair{t_0}{t_1})}
&= \D^d \prodProj_i \comp \prodPair{\sem{t_0}}{\sem{t_1}} \\
&= \prodProj_i \comp \prodPair{\sem{t_0}}{\sem{t_1}} \\
&= \sem{t_i}\,.
\end{align*}

The rule below is the one where most of the differential content appears. 
Recall that $\sem{f}$ is assumed to be multilinear, for any function $f$. 
It implies that $\D_{\word} \sem{f}$ is also multilinear by 
\Cref{prop:partial-preserve-multilinearity}, so it is possible to 
apply \Cref{prop:DDproj-commute-partial} on it.
\begin{align*}
& \sem{\depth{\Dproj_i}{d} (f^{\word j \word[1]}(t_0, \ldots, t_n))} \\
&\quad= \D^d \Dproj_i \comp \D_{\word[1]} \D_j \D_{\word} \sem{f} \comp 
\prodPair{\sem{t_0}, \ldots}{\sem{t_n}} \\
&\quad=  \D_{\word[1]} \D_{\word} \sem{f} \comp 
\singleApp{j}{\id}{\D^{\wordLetter{\word[1]}{j}} \Dproj_i} \comp 
\prodPair{\sem{t_0}, \ldots}{\sem{t_n}} \quad 
\text{by \Cref{prop:DDproj-commute-partial}} \\
&\quad= \D_{\word \word[1]} \sem{f} \comp 
\prodPair{\sem{t_0}, \ldots, \D^{\wordLetter{\word[1]}{j}} \Dproj_i \comp \sem{t_j}, \ldots}{\sem{t_n}} \\
&\quad= \sem{f^{\word \word[1]}
(t_0, \ldots, \depth{\Dproj_i}{\wordLetter{\word[1]}{j}}(t_j), \ldots, t_n)}\,.
\end{align*}
The three next rules are rather standard and are consequence of the definition
of $\Dproj_i$, $\Dinj_j$ and $\DmonadSum^n$.
\begin{align*}
\sem{\depth{\Dproj_i}{d}(\depth{\Dinj_i}{d}(t))} 
&= \D^d \Dproj_i \comp \D^d \Dinj_i \comp \sem{t} \\
&= \D^d(\Dproj_i \comp \Dinj_i) \comp \sem{t} \text{\quad by \ref{ax:D-chain}} \\
&= \D^d \id \comp \sem{t} = \sem{t}  \text{\quad by \ref{ax:D-chain}}
\end{align*}
\begin{align*}
\sem{\depth{\Dproj_i}{d}(\depth{\Dinj_{1-i}}{d}(t))} 
&= \D^d \Dproj_i \comp \D^d \Dinj_{1-i} \comp \sem{t} \\
&= \D^d(\Dproj_i \comp \Dinj_{1-i}) \comp \sem{t} \text{\quad by \ref{ax:D-chain}} \\
&= \D^d 0 \comp \sem{t} = 0 \comp \sem{t} \text{\quad by \ref{ax:Dsum-lin}} \\
&= 0 = \sem{[\,]} 
\end{align*}
\begin{align*}
\sem{\depth{\Dproj_0}{d}(\depth{\DmonadSum_n}{d}(t))} 
&= \D^d \Dproj_0 \comp \D^d \DmonadSum^n \comp \sem{t} \\
&= \D^d(\Dproj_0 \comp \DmonadSum^n) \comp \sem{t} \text{\quad by \ref{ax:D-chain}} \\
&= \D^d(\Dproj_0^{n+1}) \comp \sem{t} \\
&= (\D^d \Dproj_0)^{n+1} \comp \sem{t} \text{\quad by \ref{ax:D-chain}} \\
&= \sem{(\depth{\Dproj_0}{d})^{n+1}(t)}\,.
\end{align*}
The last rule is where finite multisets of size greater than $1$ are
introduced.
Most lines in the following sequence of equations should be understood
as follows: ``the sum above is well defined, so the sum below is well
defined and they are equal''.
\begingroup \allowdisplaybreaks
\begin{align*}
\sem{\depth{\Dproj_1}{d}(\depth{\DmonadSum_n}{d}(t))} 
&= \D^d \Dproj_1 \comp \D^d \DmonadSum^n \comp \sem{t} \\
&= \D^d(\Dproj_1 \comp \DmonadSum^n) \comp \sem{t} \quad \text{by \ref{ax:D-chain}} \\
&= \left( \D^d(\sum_{k=0}^n \Dproj_0^k \comp \Dproj_1 \comp \Dproj_0^{n-k}) \right) 
\comp \sem{t} \\
&= \left( \sum_{k=0}^n \D^d (\Dproj_0^k \comp \Dproj_1 \comp \Dproj_0^{n-k}) \right) \comp \sem{t}
 \quad \text{by \ref{ax:Dsum-lin} and \Cref{prop:D-sum-com}} \\
&= \sum_{k=0}^n \D^d (\Dproj_0^k \comp \Dproj_1 \comp \Dproj_0^{n-k}) \comp \sem{t} 
\quad \text{by \Cref{prop:sum-left-compatible}} \\
&= \sum_{k=0}^n (\D^d \Dproj_0)^k \comp \D^d \Dproj_1 \comp (\D^d \Dproj_0)^{n-k} 
\comp \sem{t} \quad \text{by \ref{ax:D-chain}} \\
&= \sum_{k=0}^n \sem{(\depth{\Dproj_0}{d})^k \depth{\Dproj_1}{d} 
(\depth{\Dproj_0}{d})^{n-k}(t)]} \,.
\end{align*}
\endgroup
Thus, 
$\sum_{k=0}^n [(\depth{\Dproj_0}{d})^k \depth{\Dproj_1}{d} 
(\depth{\Dproj_0}{d})^{n-k}(t)]$ is $\category$-summable of semantics
$\sem{\depth{\Dproj_1}{d}(\depth{\DmonadSum_n}{d}(t))}$\,.
\end{proof}

\begin{cor} \label{prop:semantic-invariance-lift}
For any term multiset $\Gamma \vdash L : A$ that is $\category$-summable,
if $L \redMultiset L'$ then $L'$ is $\category$-summable and $\sem{L'}_{\Gamma} = \sem{L}_{\Gamma}$.
\end{cor}

\begin{proof} Assume that $[t_1, \ldots, t_n]$ is $\category$-summable and that
$[t_1, \ldots, t_n] \redMultiset \multisetVar$. That is,
for any $i$, $t_i \redSym [t_i^1, \ldots, t_i^{k_i}]$ and 
$\multisetVar = \sum_{i=1}^{n} [t_i^1, \ldots, t_i^{k_i}]$.
Then by \Cref{prop:semantic-invariance}, for any $i$, $\sem{t_i^1}, \ldots, \sem{t_i^{k_i}}$ are summable of sum  $\sem{t_i}$. By assumption, $\sem{t_1}, \ldots, \sem{t_n}$ are summable, 
that is, $\sum_{j=1}^{k_1} \sem{t_1^j}, \ldots, \sum_{j=1}^{k_n} \sem{t_n^j}$ are
summable. By \Cref{prop:arbitrary-sum}, it means that the family
$\sem{t_1^1}, \ldots, \sem{t_1^{k_1}}, \ldots, \sem{t_n^1}, \ldots, \sem{t_n^{k_n}}$ is summable
of sum 
\[ \sum_{i=1}^{n} \sum_{j=1}^{k_i} \sem{t_i^j} = \sum_{i=1}^n \sem{t_i} \]
Thus $\multisetVar$ is $\category$-summable and 
$ \sem{\multisetVar} = \sem{[t_1, \ldots, t_n]}$.
\end{proof}

The usage of such term multisets may seem somewhat non deterministic. 
But any multiset generated by reductions of the calculus can be
interpreted as a summable family in deterministic models such as
probabilistic coherence spaces\footnote{Probabilistic branching is by
  no mean a form of non determinism} (see Section~\ref{sec:apcoh}) or
non uniform coherence spaces.
This determinism of the models allows to prove in \cite{Ehrhard22-pcf}
a result that roughly state that whenever a closed term of \emph{type
  integer} reduces to a term multiset $C + [\intTerm]$ (where
$\intTerm$ are the usual integer variables of \PCF{}), then
$\sem{C} = 0$.
That is, only one of the branches of the reduction rule 
\[ \depth{\Dproj_1}{d}(\depth{\DmonadSum_n}{d}(t)) \red  
\sum_{k=0}^{n} [ (\depth{\Dproj_0}{d})^k \depth{\Dproj_1}{d} (\depth{\Dproj_0}{d})^{n-k} (t)] \]
produces a non empty multiset.
The proof relies on the fact that any term of type integer will be
interpreted in $\ACOH$ as a Dirac distribution $\dirac{n}$ on $\N$ or
as the zero distribution, because the calculus does not feature any
form of probabilistic branching. Thus, a term multiset of type integer
is $\ACOH$-summable if and only if there is at most one term in the
multiset whose semantic is not $0$.
In particular, $\semPcoh{\intTerm} = \dirac{\intVar}$ and
$C + [\intTerm]$ is $\ACOH$-summable (by
\cref{prop:semantic-invariance-lift}) so $\semPcoh{C} = 0$.
%
%
One can also use non-uniform coherence spaces for proving the same result
in a similar way.
This observation led to the development of a completely deterministic
Krivine Machine for a coherent differential version of PCF
in~\cite{Ehrhard22-pcf}, extending the projections path with a
writable memory structure.

\section*{Conclusion}

We have introduced and studied a general categorical framework for
coherent differentiation, a new approach to the differential calculus
which does not require the ambient category to be (left-)additive.
We have also proposed some basic syntactical constructs accounting in
a term language for these new categorical constructs.
These are the foundations for a principled and systematic approach to
the denotational semantics of functional programming languages like
(probabilistic) PCF extended with coherent differentiation.
As shown in~\cite{Ehrhard22-pcf} such an extension can perfectly
feature general recursive definitions as well as deterministic or
probabilistic behaviors, in sharp contrast with the Differential
\(\lambda\)-calculus~\cite{EhrhardRegnier02} which is inherently
non-deterministic.
Accordingly, the next step will be to specialize the present general
axiomatization to the case where the category is cartesian closed.

\section*{Acknowledgment}
We thank the reviewers for their careful reading and helpful comments.
This work was partly supported by the ANR project %
\emph{Probabilistic Programming Semantics (PPS)} ANR-19-CE48-0014.

\bibliographystyle{IEEEtran}
\bibliography{IEEEabrv, biblio.bib}

\end{document}